\documentclass[letterpaper,11pt]{article}
\usepackage{fullpage}
\usepackage[ansinew]{inputenc}
\usepackage{verbatim}
\usepackage{graphicx}
\usepackage{tabularx}
\usepackage{array}
\usepackage{delarray}
\usepackage{dcolumn}
\usepackage{float}
\usepackage{amsmath}
\usepackage{psfrag}
\usepackage{booktabs}
\usepackage{multirow,hhline}
\usepackage{amstext}
\usepackage{amssymb}
\usepackage{fancyhdr}
\usepackage{hyperref}
\usepackage{setspace}
\usepackage{amsthm}
\usepackage{makeidx}
\usepackage{tikz}
\usepackage{lipsum}
\usepackage{url}
\usepackage{adjustbox}
\usepackage{enumitem}
\usepackage{threeparttable}
\usepackage{adjustbox}
\usepackage{bbold}
\usepackage{pifont}
\usepackage{bookmark}
\usepackage{algorithm}
\usepackage{algpseudocode}
\usepackage{natbib}

\setcitestyle{round,authoryear,longnamesfirst}
\usetikzlibrary{arrows,calc}
\hypersetup{colorlinks=true,allcolors=blue}
\makeindex
\allowdisplaybreaks
\numberwithin{equation}{section}
\newtheorem{lem}{Lemma}[section]
\newtheorem{thm}{Theorem}[section]
\newtheorem{cor}{Corollary}[section]
\newtheorem{defn}{Definition}[section]
\newtheorem{ass}{Assumption}
\newtheorem{prop}{Proposition}[section]

\newcommand{\diag}{\text{diag}}
\newcommand{\rk}{\text{rk}}

\newcommand{\indep}{\perp\!\!\!\perp}

\newcommand{\Prob}{\ensuremath{\mathbb{P}}}
\newcommand{\convP}{\stackrel{\Prob}{\longrightarrow}}
\newcommand{\convD}{\rightsquigarrow}

\newcommand{\N}{\mathcal{N}}
\newcommand{\eps}{\varepsilon}
\newcommand{\tr}{\text{tr}}

\renewcommand{\epsilon}{\varepsilon}

\DeclareMathOperator*{\argmin}{arg\,min}

\newcommand{\hidefromtoc}{%
  \setcounter{oldtocdepth}{\value{tocdepth}}%
  \addtocontents{toc}{\protect\setcounter{tocdepth}{-10}}%
}
\newcommand{\unhidefromtoc}{%
  \addtocontents{toc}{\protect\setcounter{tocdepth}{\value{oldtocdepth}}}%
}
\newcounter{oldtocdepth}

\theoremstyle{definition}

\newtheorem{rem}{Remark}[section]

{}

\title{Honest Inference for High-Dimensional Linear Models\thanks{We thank Yichong Zhang for the helpful comments. The financial support of the Ministry of Education Singapore AcRF Tier 1 grants RG132/23 and RG133/25 is gratefully acknowledged.}}
\author{
Dennis Lim\thanks{Singapore Management University. \ E-mail: dennis.lim.2019@phdecons.smu.edu.sg} 
\and 
Pengjin Min\thanks{Singapore Management University. \ E-mail: pjmin.2022@phdecons.smu.edu.sg}
\and 
Jonathan H. W. Tan\thanks{Nanyang Technological University. \ E-mail: j.tan@ntu.edu.sg}
\and 
Jiezheng Wei\thanks{Singapore Management University. \ E-mail: jzwei.2022@phdecons.smu.edu.sg}
}
\date{\today}

\begin{document}
\maketitle

\begin{abstract}
\noindent We propose a heteroskedasticity-robust and dimention-adaptive test for multiple linear restrictions in sparse high-dimensional linear regression models. It combines desparsified LASSO estimation with a ridge-regularized quadratic form statistic and critical values that adapt to the restriction dimension. Under stated regularity conditions, our approach is asymptotically valid under a fixed or diverging number of restrictions without specifying the asymptotic regime in advance. Our random-design theory allows independent non-identically distributed covariate rows with marginal sub-Weibull tails and heteroskedastic errors satisfying finite polynomial moments. Monte Carlo simulations demonstrate improved finite-sample size control relative to conventional Wald tests, particularly as the number of restrictions increases. \\


\end{abstract}

\hidefromtoc
\section{Introduction}

We consider the linear model
\begin{align}
    Y = {\bf X}\beta^0 + \epsilon,
    \label{eq:model}
\end{align}
where $Y \in \mathbb{R}^n$ is the outcome vector,
${\bf X} \in \mathbb{R}^{n \times p}$ is the \textit{design matrix},
which may be either fixed or random, and $\epsilon \in \mathbb{R}^n$ is
the error vector. We assume that, conditional on ${\bf X}$, the individual
errors $\{\epsilon_i\}_{i\in[n]}$ are independent but not necessarily
identically distributed, thereby allowing for heteroskedasticity. In the
random-design setting, we also allow the rows
$\{X_i\}_{i\in[n]}$ of ${\bf X}$ to be non-identically distributed.

We are concerned with the high-dimensional setting in which $p$ may be
much larger than $n$ and $\beta^0$ is sparse or approximately
sparse.\footnote{We focus on the case in which $\beta^0$ is exactly sparse,
but our theory can be extended to the approximately sparse case provided
that the approximation error is sufficiently small.} Suppose an analyst is interested in testing a system of linear
restrictions involving the subvector
$\beta_G^0 := (\beta_j^0)_{j\in G}$, where
$G\subseteq\{1,\ldots,p\}$ is an arbitrary index set with $|G|=m$.
Specifically, consider $r$ restrictions of the form
\begin{align*}
    R_G\beta_G^0=q,
    \qquad
    R_G\in\mathbb{R}^{r\times m},
    \qquad
    \operatorname{rank}(R_G)=r\leq m,
    \qquad
    q\in\mathbb{R}^r,
\end{align*}
where $R_G$ and $q$ are specified by the analyst. Let the
\textit{restriction matrix} $R\in\mathbb{R}^{r\times p}$ be obtained by
padding $R_G$ with zero columns outside $G$. Then
$R\beta^0=R_G\beta_G^0$, and the testing problem can be written as
\begin{align}
    H_0:R\beta^0=q
    \qquad \text{vs.} \qquad
    H_1:R\beta^0\neq q.
    \label{eq:H_0}
\end{align}
By construction, $R$ has the same full row rank as $R_G$ and has nonzero
columns only on the coordinates indexed by $G$. When $m=r=1$, the testing
problem reduces to inference on a single linear restriction involving one
coefficient. More generally, we allow both $m$ and $r$ to diverge with
$n$, subject to the rate conditions imposed below. An important special
case is the joint test that a group of coefficients is zero, for which
$R_G=I_m$, $r=m$, and $q=0_m$.

\subsection{Contribution to the Literature}

In this paper, we propose a heteroskedasticity-robust dimension-adaptive test for a broad class of linear restrictions in sparse high-dimensional regressions. This procedure is asymptotically valid across a fixed or diverging number of restrictions, avoiding the need to classify problems into low or increasing dimensional asymptotic regimes. We also establish the power properties of the proposed test under alternatives for fixed designs, complementing the asymptotic size control results.

Our contribution lies at the intersection of three strands of literature: 1) debiased inference in high-dimensional linear models, 2) testing many linear restrictions, and 3) heteroskedasticity-robust inference. We start by building on the literature on inference for low-dimensional parameters in high-dimensional linear models \citep[e.g.][]{belloni2012,belloni2014inferenceRES,van-de-geer(2014),javanmard2014confidence,zhang2014confidence,ning2017general}. Our work is positioned within the general framework of ``desparsified'' or ``debiased'' LASSO pioneered by \cite{van-de-geer(2014)}, \cite{javanmard2014confidence}, and \cite{zhang2014confidence}. However, existing results are typically derived under relatively homogeneous designs with relatively light-tailed conditions on regressors or errors. We relax these stringent requirements to allow for heteroskedasticity both in the error terms and independent potentially non-independent covariate rows. 

A common approach in the literature is to first establish the validity of a test under a fixed design and then show that an $i.i.d.$ random design satisfies the required fixed-design conditions with high probability. We depart from this approach in two respects. First, we provide high-level conditions for the validity of our test under a fixed design when general LASSO estimators are employed. These conditions are not tied to a particular random-design specification and may therefore be of independent interest. Second, we provide sufficient conditions under which a random design matrix
satisfies these fixed-design requirements. Importantly, our random-design results do not require the observations to be identically distributed. In
particular, we allow the observations to have different marginal distributions, subject to marginal sub-Weibull tail conditions. Our work accommodates the joint testing of multiple linear restrictions, with the number of restrictions allowed to increase with sample size. Finally, we characterize the asymptotic power of the proposed test under alternative hypotheses, complementing our theoretical results on its validity under the null.

Our work is closely related to \cite{caner2018asymptotically}, who also consider inference within the desparsified LASSO framework and allow for heteroskedastic, non-sub-Gaussian covariates and errors under finite-moment conditions. Their asymptotic normality result concerns a single linear contrast, although the number of coefficients entering that contrast is allowed to increase with the sample size. In contrast, we consider a general system of linear restrictions \(R\beta^0=q\), where the number of restrictions itself may diverge. Their formal results are established under i.i.d. covariates, although they note that their arguments could be extended to non-identically distributed covariates. Our random-design theory directly accommodates independent but non-identically distributed covariates under marginal sub-Weibull tail conditions.

Our work also speaks to research on testing many restrictions in high-dimensional linear regressions. \cite{shi2019linear} consider testing for a subvector of the regression coefficients for sparse generalized linear models, including sparse linear models as special cases, where the number of restrictions is smaller than the sample size of order $o(n^{1/3})$. They allow for heteroskedastic error terms but restrict them to be sub-exponential. \cite{zhu2018linear} propose a test statistic for a single contrast of the whole vector of the regression coefficients; \cite{zhang2025novel} extend this test to testing many restrictions where the number of restrictions can be larger than the sample size. Their analysis assumes that the regressors and errors are homoskedastic and the errors are sub-Gaussian, and require the number of restrictions to diverge for their test statistic to be valid. Building on recent work on concentration inequalities under general tail assumptions \citep[e.g.][]{kuchibhotla2022moving}, however, we allow a design matrix with heavier than sub-Gaussian tails and heteroskedastic errors satisfying only polynomial moments. Under favorable sparsity regimes, the restriction dimension can grow nearly at the rate of $n^{1/2}$ up to logarithmic factors. Importantly, our test statistic is valid regardless of whether the number of restrictions is bounded or diverges.

In a different context of moderate dimensional linear regressions, where the number of covariates is smaller than (but can be proportional to) the sample size, inference procedures that do not require a sparsity condition have been proposed under general heteroskedasticity \citep[e.g.][]{cattaneo2018inference,KSS2020,jochmans2022heteroscedasticity,Anatolyev-Solvsten(2023)}. In particular, \cite{Anatolyev-Solvsten(2023)} propose a test statistic for many restrictions that remains valid even when the number of restrictions is proportional to the sample size. We adapt their logic to the high-dimensional case, where the number of covariates can be much larger than the sample size and direct covariance inversion is unavailable, by combining desparsification with regularization of the restriction covariance matrix. Their procedure requires ${\bf X}^\top{\bf X}$ and
\[
\sum_{\ell\neq i,j,k} X_\ell X_\ell^\top
\]
to be invertible for every $i,j,k\in[n]$, where $X_\ell\in\mathbb{R}^p$ denotes the $\ell$-th row of ${\bf X}$. When $p\gg n$, however, these matrices are necessarily rank-deficient and hence non-invertible. Consequently, the OLS estimators underlying their procedure are not available in our setting, requiring us to modify their approach to accommodate high-dimensional covariates. \cite{Anatolyev-Solvsten(2023)} note the possibility of an alternative strategy when ${\bf X}^\top{\bf X}$ is non-invertible. However, as discussed in Footnote~7 of their paper, the resulting heteroskedasticity-robust Wald statistic still requires the inversion of a high-dimensional covariance matrix estimator, which poses substantial challenges for size control.

In the special case of testing many zero restrictions (i.e. $R = I_m$), \cite{dezeure2017high} and \cite{zhang2017simultaneous} propose bootstrap inference procedures based on the high-dimensional Gaussian approximation in \cite{chernozhukov2013}, while \cite{xia2018two} derive an analytical test statistic. These methods allow the number of restrictions to be as large as the number of covariates. In contrast, we set the number of restrictions to be relatively small compared with the sample size, but allow for a more general form of restrictions. 

We summarize our contributions in Table \ref{tab:comp}.
The distributional assumptions satisfy the following
hierarchy of implications:
joint Gaussianity $\Rightarrow$ joint sub-Gaussianity
$\Rightarrow$ joint sub-Weibullity
$\Rightarrow$ marginal sub-Weibullity,
where the sub-Weibull order satisfies $0<\nu\leq 2$.

\begin{table}[ht]
\centering
\caption{Comparison with Related Tests}
\label{tab:comp}
\footnotesize

\begin{adjustbox}{max width=\linewidth}
\begin{tabular}{llll}
\toprule
\textbf{Paper}
& \textbf{Design assumptions}
& \textbf{Restrictions}
& \textbf{Error assumptions} \\
\midrule

VBRD14
& $i.i.d.$ joint sub-Gaussian rows
& $R$ selects components of $\beta^0$
& Homoskedastic, sub-Gaussian \\

DBZ17
& $i.i.d.$ joint sub-Gaussian rows
& $R=I_m$; fixed/diverging $r$; $\log m=o(n^{1/7})$
& Heteroskedastic, finite fourth moment \\

ZB18
& $i.i.d.$ joint Gaussian rows
& General $R$; $r=1$; $m=p$
& Homoskedastic, sub-Gaussian \\

SSCL19
& $i.i.d.$ joint sub-Gaussian rows
& General $R$; fixed/diverging $r$; $m=o(n^{1/3})$
& Heteroskedastic, sub-exponential \\

ZYL25
& $i.i.d.$ joint sub-Gaussian rows
& General $R$; $r\rightarrow\infty$
& Homoskedastic, sub-Gaussian \\

\textbf{Our paper}
& Independent, non-identically distributed,
& General $R$; fixed/diverging $r$;
$m=o\!\left(n/(s_0\log p)^2\right)$
& Heteroskedastic, finite fourth moment \\& marginally sub-Weibull rows \\

\bottomrule
\end{tabular}
\end{adjustbox}

\vspace{2mm}

\begin{minipage}{\linewidth}
\footnotesize
\textit{Note: VBRD14: \cite{van-de-geer(2014)}; DBZ17: \cite{dezeure2017high}; ZB18: \cite{zhu2018linear}; SSCL19: \cite{shi2019linear}; ZYL25: \cite{zhang2025novel}} 
\end{minipage}

\end{table}

\subsection{Notation}
For a real vector $a$, let $\diag (a)$ denote a diagonal matrix with diagonal entries given by $a$. For a real matrix $A \in \mathbb{R}^{s_1 \times s_2}$, denote its transpose by $A^\top$, its trace by $\tr (A)$, its rank by $\rk (A)$, its Frobenius norm by $\lVert A \rVert_{F}$, its $(i,j)$-th entry by $A_{ij}$. For any $p \in [1,\infty]$, let $\|\cdot\|_p$ denote the usual vector $\ell_p$-norm. 
We define the induced (operator) matrix $p$-norm by $\|A\|_{p} \;:= \sup_{\|x\|_{p}=1}\|Ax\|_{p}$. We also define the entrywise max norm by $\|A\|_{\max} \;:=\; \max_{1\le i \le s_1,\; 1\le j \le s_2} |A_{ij}|$. For a square matrix $A$, let $\lambda_{\max} (A)$, $\lambda_i (A)$, $\lambda_{\min} (A)$  denote the largest, $i$-th largest and smallest eigenvalues, respectively.  For $a,b\in\mathbb{R}$, we write $a \lesssim b$ if there exists a universal constant $C>0$, independent of $n$, such that $a \le C b$, and $a \asymp b$ if $a \lesssim b$ and $b \lesssim a$. We let $\overline{c}$ and $\underline{c}$ respectively denote constants for upper and lower bound. The generic notation of probability measures is denoted by $\Prob$.

\subsection{Structure of the paper}
The remainder of the paper is organized as follows.
Section~\ref{sec:test_statistic} motivates and constructs our proposed test.
Section~\ref{section:strong_approx} establishes a strong approximation result
for a general class of quadratic-form statistics.
Section~\ref{section:size_control} develops the asymptotic theory underlying
the validity of our test. We first establish, in
Theorem~\ref{thm:size-control_fixed_design}, sufficient conditions for size
control under a fixed design generally stated, and turn to random designs.
Proposition~\ref{prop:sufficient_conditions_for_fixed_design} provides general sufficient conditions under which a random design satisfies the requirements of
Theorem~\ref{thm:size-control_fixed_design}, thereby establishing validity in
the random-design setting. Theorem~\ref{thm:size-control_random_design_extension} extends the results to
the conservative LASSO.
Section~\ref{sec:sim} examines the finite-sample performance of the proposed
test through Monte Carlo simulations. Section~\ref{sec:empirical_application} illustrates its use in an empirical application. Section~\ref{sec:conclusion} suggests future research and concludes.

\section{Test Statistic: Motivation and Construction} 
\label{sec:test_statistic}
\subsection{The General Idea} \label{subsec:general-idea}
In this section, we outline the central idea from \cite{van-de-geer(2014)} (see also \cite{javanmard2014confidence} and \cite{zhang2014confidence}), which provides motivation for both the problem and our approach. Their key insight is to invert the Karush-Kuhn-Tucker (KKT) condition characterizing the least absolute shrinkage and selection operator (LASSO) solution of \cite{tibshirani(1996)}. To formalize the problem, recall that the LASSO estimator is defined as
\begin{align}
    \widehat{\beta}(\lambda_0) \in \arg \min_{\beta \in \mathbb{R}^p} \frac{1}{n} \lVert Y - {\bf X}\beta \rVert_2^2 + 2\lambda_0 \lVert \beta \rVert_1,
    \label{eq:beta-hat}
\end{align}
where $\lambda_0 >0$ is the regularization parameter. For notational simplicity, we write $\widehat{\beta}\equiv\widehat{\beta}(\lambda_0)$ throughout the remainder of the paper, unless the dependence on $\lambda_0$ needs to be made explicit. The {\it desparsified LASSO estimator} in \cite{van-de-geer(2014)} is defined as
\begin{align}
    \widehat{b} := \widehat{\beta} + \frac{1}{n} \widehat{\Theta} {\bf X}^\top ( Y - {\bf X}\widehat{\beta}),
    \label{eq:b-hat-general}
\end{align}
where $\widehat{\Theta} \in \mathbb{R}^{p \times p}$ is the ``approximate inverse'' of $\widehat{\Sigma} := {\bf X^\top X}/n$, and we have 
\begin{align}
    \sqrt{n} (\widehat{b} - \beta^0) = W - \Delta, \quad W := \widehat{\Theta} {\bf X}^\top \eps /\sqrt{n},
     \label{eq:beta_hat_inference}
\end{align}
where $\Delta := \sqrt{n}(\widehat{\Theta} \widehat{\Sigma} - I_p)(\widehat{\beta} - \beta^0)$ is the remainder term.

To construct $\widehat{\Theta}$, we follow \cite{van-de-geer(2014)} to use the LASSO for the node-wise regression on the design matrix ${\bf X}$. For each $j = 1,...,p$, denote by $X_j\in \mathbb{R}^n$ the $j$-th column vector of ${\bf X}$ and
\begin{align}
    \widehat{\gamma}_j(\lambda_j) \in \argmin_{\gamma \in \mathbb{R}^{p-1}} \frac{1}{n} \lVert X_j - {\bf X}_{-j} \gamma  \rVert_2^2 + 2\lambda_j \lVert \gamma \rVert_1,
    \label{eq:gamma_j}
\end{align}
for some regularization parameter $\lambda_j >0$; for notational simplicity, we write
$\widehat{\gamma}_j \equiv \widehat{\gamma}_j(\lambda_j)$ throughout the remainder of the paper, unless the dependence on $\lambda_j$ needs to be made explicit. We can write the elements of $\widehat{\gamma}_j$ as $\{ \widehat\gamma_{j,k}: k=1,...,p, k \neq j\}$, and define 
\begin{align*}
    \widehat{C} :=
    \begin{bmatrix}
        1 & -\widehat{\gamma}_{1,2}& \cdots &-\widehat\gamma_{1,p} \\
        -\widehat{\gamma}_{2,1}& 1 & \cdots& -\widehat\gamma_{2,p} \\
        \vdots & \vdots & \ddots & \vdots \\
        -\widehat{\gamma}_{p,1} & -\widehat{\gamma}_{p,2} & \cdots & 1
    \end{bmatrix},
\end{align*}
and 
\begin{align*}
    \widehat{T}^2 := \diag(\widehat{\tau}_1^2,...,\widehat{\tau}^2_p),
\end{align*}
with $\widehat{\tau}_j^2 := \lVert X_j - {\bf X}_{-j} \widehat{\gamma}_j \rVert_2^2/n + \lambda_j \lVert \widehat{\gamma}_j\rVert_1 $ for each $j = 1,...,p$. We then define
\begin{align*}
    \widehat{\Theta} := \widehat{T}^{-2} \widehat{C}.
\end{align*}

In \cite{van-de-geer(2014)}, the properties of $\widehat \beta$ as well as $\widehat \Theta$ are derived when $\bf X$ is sub-Gaussian or bounded under homoskedasticity; they further assume $\eps$ to be i.i.d. sub-exponential when $\bf X$ is sub-Gaussian. Using the concentration inequalities under general tail assumptions (e.g. \cite{kuchibhotla2022moving}), we derive these properties when $\bf X$ is sub-Weibull (see Appendix \ref{sec:appendix_sub_weibull} for definition) and $\eps$ has only polynomial moments under general heteroskedasticity.

\subsection{Our Test Statistic} \label{sec:Test_Statistic}

To test the hypothesis in (\ref{eq:H_0}), we may pre-multiply both sides of \eqref{eq:beta_hat_inference} by the matrix \( R \) to obtain
\begin{align}
    \sqrt{n} (R\widehat{b} - R\beta^0) = RW - R \Delta,
    \label{eq:beta_hat_inference:R}
\end{align}
where, under the null, we have \( q = R\beta^0 \in \mathbb{R}^r \). Under stated regularity conditions, the remainder term $R \Delta$ is negligible compared to $RW$. If $r$ is fixed, it is then straightforward to derive a Wald test based on the asymptotic normality of $RW$ (e.g. \cite{van-de-geer(2014)}), or a heteroskedasticity-robust Wald test that also takes into account heteroskedasticity in error terms (e.g. \cite{caner2018asymptotically}). However, these test statistics might not be valid when we consider $r$ as diverging. In the context of moderate dimensional linear regressions, a modified Fisher's $F$ test for testing the hypothesis in (\ref{eq:H_0}) is proposed by \cite{Anatolyev-Solvsten(2023)}, which provides size-control under general heteroskedasticity regardless of whether $r$ is bounded or diverging. 

This can be attractive in practice, especially when analysts do not want to make a preliminary judgment about whether $r$ is in the low- or high-dimensional regime. Another advantage of the modified $F$ test is that, in contrast with the heteroskedasticity-robust Wald test, it does not involve the inverse of an estimated sandwich matrix whose entries carry the noisy residuals, and is thus expected to be more stable in finite samples as demonstrated by our simulations.

We adapt this $F$ test to the current setting, with two important modifications. First, we replace $\widehat{\Sigma}^{-1}$ in $R \widehat{\Sigma}^{-1} R^\top$, which is not well defined when $p > n$,  with $\widehat{\Theta} \widehat{\Sigma} \widehat{\Theta}^\top$, which can be interpreted as an approximation of $\widehat{\Sigma}^{-1}$. Second, we introduce ridge regularization to $R \widehat{\Theta} \widehat{\Sigma} \widehat{\Theta}^\top R^\top$, which is motivated by the observation that inverting $R \widehat{\Theta} \widehat{\Sigma} \widehat{\Theta}^\top R^\top$ can be unstable or even undefined in finite samples, especially when $r$ is large. We further derive a novel strong approximation result for quadratic forms, which allows us to construct a critical value that leads to valid inference under both fixed and random designs. 

We introduce the ridge-regularized projection matrix
\begin{align*}
    P_\theta := \frac{1}{n} {\bf X} \widehat{\Theta}^\top R^\top \left( R\widehat{\Theta} \widehat{\Sigma} \widehat{\Theta}^\top R^\top + \theta I_r \right)^{-1} R \widehat{\Theta} {\bf X}^\top,
\end{align*}
where the ridge regularization parameter $\theta > 0$ is determined solely by the matrix ${\bf X} \widehat{\Theta}^\top R^\top$. We denote the $(i,j)$-th entry of $P_\theta$ by $P_{\theta,ij}$. The test statistic that we propose in this paper is given by
\begin{align}
    \widehat{Q}(q) := \frac{b(q)^\top \left( R\widehat{\Theta} \widehat{\Sigma} \widehat{\Theta}^\top R^\top + \theta I_r \right)^{-1} b(q)}{\sum_{i=1}^n P_{\theta,ii} \widehat{\sigma}_i^2}, \quad b(q) := \sqrt{n}(R \widehat{b} - q) \in \mathbb{R}^r,
    \label{Q_hat_statistic_definition:1}
\end{align}
where $\widehat{\sigma}_i^2 := (Y_i - X_i \widehat{\beta})^2$ is the estimator for $\sigma_i^2: = \text{Var}(\epsilon_i|{\bf X})$ and $\widehat{\beta}$ is given in \eqref{eq:beta-hat}. Here we use $\widehat{\beta}$ in $\widehat{\sigma}_i^2$ as we can derive the consistency of $\widehat{\sigma}_i^2$ directly from the properties of $\widehat{\beta}$. We also define the effective degrees of freedom associated with the regularization level $\theta$ as $K_\theta := \sum_{i \in [n]} P_{\theta,ii}$. The regularization parameter $\theta$ is selected to ensure desirable theoretical properties. Specifically, we set
\begin{align}
\theta \equiv \text{the } \lceil \varpi r_n \rceil\text{-th largest eigenvalue of }
R \widehat{\Theta} \widehat{\Sigma} \widehat{\Theta}^\top R^\top,
\label{eq:theta}
\end{align}
with $0<\varpi < 1$ that can depend on ${\bf X} \widehat{\Theta}^\top R^\top$, and $r_n := \rk({\bf X} \widehat{\Theta}^\top R^\top)$. Note that $r_n \leq r$, where $r$ is the number of restrictions in our null hypothesis \eqref{eq:H_0}, which we allow to diverge. In our setup, $r_n$ is of the same order as $r$, but we allow for $r_n < r$, which motivates us to use the ridge regularization.

\begin{rem}
    When $\theta$ is chosen arbitrarily, it is possible that $K_\theta$ and $r_n$ are not of the same order. This can lead to non-standard limiting distribution and complicate inference (e.g. \cite{KSS2020} and \cite{LWZ2024}). The choice in (\ref{eq:theta}) ensures that the order of $K_\theta$ is the same as $r_n$. In addition, under stated regularity conditions, we show that $\varpi$ can be selected in such a way that $\theta$ is bounded away from zero asymptotically when $\bf X$ is random (see Proposition \ref{prop:ridge_parameter}).
\end{rem}

\begin{rem}
    As noted by \cite{caner2018asymptotically}, computing $R \widehat \Theta$ can be much easier than computing the whole $\widehat \Theta$, depending on the structure of $R$. For example, if $R = e_j^\top \in \mathbb{R}^{1 \times p}$ where $e_j$ is the $j$-th unit vector in $\mathbb{R}^p$, then only the $j$-th row of $\widehat \Theta$ needs to be computed. More generally, at most $m$ node-wise LASSO regressions are required for computing $R \widehat \Theta$. This reduces computational demands compared to running the node-wise LASSO regression $p$ times when $m \ll p$. 
\end{rem}

\subsection{Construction of Critical Value}\label{sec:CV}
Our analytical test compares the test statistic $\widehat Q(q)$ with a robust critical value $\mathcal{C}_{\alpha}$ that is valid when $r$ is bounded as well as when $r$ diverges, where $\alpha \in (0,1)$ is the significance level. We will reject $H_0: R\beta^0 = q$ at the $\alpha$ significance level whenever 
\begin{align}
    \widehat Q(q) > \mathcal{C}_{\alpha}.
    \label{eq:rejection-criteria}
\end{align}
Inspired by \cite{Anatolyev-Solvsten(2023)}, we define this critical value as
\begin{align}
    \mathcal{C}_{\alpha} := 1 + \frac{ \sqrt{ K_\theta \widehat{\Phi}}}{  \sum_{i \in [n]} P_{\theta,ii} \widehat{\sigma}_i^2 }  \left( \frac{ q_{1-\alpha} (F_{\widetilde{w}_n})-1}{ \sqrt{2 \sum_{i \in [r_n] } \widetilde{w}_{i,n}^2} } \right),
    \label{eq:critical_value}
\end{align}
where $\{\widetilde{w}_{i,n}\}_{i \in [r_n]}$ and $\widehat{\Phi}$ are defined later. Here $q_{1-\alpha}(F_{\widetilde{w}_n})$ is the $(1-\alpha)$-quantile of the weighted chi-square random variable 
\[
F_{\widetilde{w}_n} := \sum_{i \in [r_n]} \widetilde{w}_{i,n}\,\chi^2_{1,i}
\]
conditional on the data $(Y, {\bf X})$, where $\{\chi_{1,i}^2\}_{i \in [r_n]}$ are $r_n$ independent chi-squared random variables with one degree of freedom. This quantile can be obtained via simulation given the weights~$\{\widetilde{w}_{i,n}\}_{i \in [r_n]}$.

To see the intuition of the critical value, we need to explain the limiting distribution of our test statistic $\widehat Q(q)$ under the null $R \beta^0= q$. When $r$ is bounded, by using the strong approximation result we derive in this paper, we show that
\begin{align}
    \widehat{Q}(q) \overset{d}{\approx} \mathcal{Z}^\top \mathcal{D}_n \mathcal{Z} = \sum_{i \in [r_n]} w_{i,n} \chi^2_{1,i},
    \label{eq:heuritsics_fixed_K}
\end{align}
where $\mathcal{Z} \overset{d}{=} \mathcal{N}(0,I_{r_n}) $ and  $\mathcal{D}_n:= \diag(w_{1,n},...,w_{r_n,n})$ contains the decreasingly ordered non-zero eigenvalues of 
\begin{align*}
\Omega := \frac{  (R \widehat{\Theta} {\bf X}^\top \Lambda {\bf X} \widehat{\Theta}^\top R^\top)^{1/2} (R \widehat{\Theta} \widehat{\Sigma} \widehat{\Theta}^\top R^\top + \theta I_r)^{-1} (R \widehat{\Theta} {\bf X}^\top \Lambda {\bf X} \widehat{\Theta}^\top R^\top)^{1/2} }{n \sum_{i \in [n]} P_{\theta,ii} \sigma_i^2},
\end{align*}
with $\Lambda := \diag(\sigma_1^2,\cdots,\sigma_n^2)$. Note that $\tr (\Omega) = 1$ and $\Omega$ is positive semi-definite, so that $w_n := (w_{1,n},\cdots,w_{r_n,n})^\top$ are positive and sum up to one. Using the ``plug-in" rule, we let $\widetilde{w}_n := (\widetilde{w}_{1,n},\cdots,\widetilde{w}_{r_n,n})^\top$ be the first $r_n$ decreasingly ordered eigenvalues of 
\begin{align*}
\widehat{\Omega}:= \frac{ (R \widehat{\Theta} {\bf X}^\top \widehat\Lambda {\bf X} \widehat{\Theta}^\top R^\top)^{1/2} (R \widehat{\Theta} \widehat{\Sigma} \widehat{\Theta}^\top R^\top + \theta I_r)^{-1} (R \widehat{\Theta} {\bf X}^\top \widehat\Lambda {\bf X} \widehat{\Theta}^\top R^\top)^{1/2} }{n \sum_{i \in [n]} P_{\theta,ii} \widehat{\sigma}_i^2},
\end{align*}
where $\widehat \Lambda := \diag(\widehat{\sigma}_1^2,\cdots ,\widehat{\sigma}_n^2)$. Then, note that the critical value $\mathcal{C}_{\alpha}$ can be rearranged as
\begin{align*}
    q_{1-\alpha}(F_{\widetilde{w}_n}) + (q_{1-\alpha}(F_{\widetilde{w}_n})-1) \left( \frac{  {\sqrt{ K_\theta \widehat{\Phi}}}  }{\left({\sum_{i \in [n]} P_{\theta,ii} \widehat{\sigma}_i^2}\right)\sqrt{2 \sum_{i \in [r_n]} \widetilde{w}_{i,n}^2 } } -1 \right),
\end{align*}
and we show that the second term is asymptotically negligible. In addition, the first term is close to $q_{1-\alpha}(F_{w_n})$, the $(1-\alpha)$-quantile of the weighted chi-square random variable   
\[
F_{w_n} := \sum_{i \in [r_n]} w_{i,n}\,\chi^2_{1,i}
\]
conditional on $\bf X$. This guarantees that our analytical test achieves the correct asymptotic size under bounded $r$. 

If, however, $r$ diverges, the weighted chi-square random variable converges in distribution to the standard normal distribution
\begin{align*}
   \left( \sum_{i \in [r_n]} 2\tilde w_{i,n}^2 \right)^{-1/2} (F_{\tilde{w}_n}-1) \convD \N(0,1),
\end{align*}
and we show that the standardized $\widehat Q(q)$ also converges in distribution to the standard normal distribution
\begin{align*}
   \Phi^{-1/2} \left[\frac{1}{\sqrt{K_\theta}}\sum_{i \in [n]} P_{\theta,ii} \widehat{\sigma}_i^2\right] \left(\widehat Q(q)-1\right) \convD \N(0,1),
\end{align*}
where 
\begin{align*}
    \Phi := \frac{2}{K_\theta} \sum_{i \in [n]}\sum_{j \neq i}P_{\theta,ij}^2 \sigma_i^2 \sigma_j^2.
\end{align*}
Motivated by the work of \cite{crudu2021} and \cite{dovi2024ridge}, we consider the estimator
\begin{align}
\widehat{\Phi} := \frac{2}{K_\theta} \sum_{i \in [n]} \sum_{j \neq i} P_{\theta,ij}^2 \widehat{\sigma}_i^2 \widehat{\sigma}_j^2 \label{Phihat_def}
\end{align}
for $\Phi$. As we show the consistency of $\widehat \Phi$, the asymptotic size-control of our test is also guaranteed under diverging $r$.

\begin{rem}
In our simulations, we use the following finite-sample adjustment of
$\widehat{\Phi}$, motivated by the diagonal components:
\begin{align*}
    \widetilde{\Phi}
    :=
    \widehat{\Phi}
    +
    \frac{1}{K_\theta}
    \sum_{i\in[n]}
    P_{\theta,ii}^2\widehat{\sigma}_i^4.
\end{align*}
The motivation for this adjustment is that, although the diagonal component of
$\eps^\top P_\theta \eps/\sqrt{K_\theta}$ is asymptotically negligible under
our assumptions, it may still contribute non-negligible variability in finite
samples. In particular, the diagonal component is given by
\begin{align*}
    \frac{1}{\sqrt{K_\theta}}
    \sum_{i\in[n]}P_{\theta,ii}\eps_i^2,
\end{align*}
and the additional term
\begin{align*}
    \frac{1}{K_\theta}
    \sum_{i\in[n]}P_{\theta,ii}^2\widehat{\sigma}_i^4
\end{align*}
estimates its contribution to the variance. Thus, $\widetilde{\Phi}$ accounts
for a source of finite-sample variability that is omitted from
$\widehat{\Phi}$ but vanishes asymptotically. Since
\begin{align*}
    \frac{1}{K_\theta}
    \sum_{i\in[n]}P_{\theta,ii}^2\widehat{\sigma}_i^4
    = o_p(1)
\end{align*}
under the conditions of our asymptotic analysis, replacing
$\widehat{\Phi}$ by $\widetilde{\Phi}$ does not affect any of our
asymptotic results.
\end{rem}

\section{Strong Approximation of Quadratic Forms} \label{section:strong_approx}
This section focuses on approximating the quadratic form in the numerator of \eqref{Q_hat_statistic_definition:1} by a Gaussian counterpart, without imposing any restriction on the dimension \( r \). As this result may be of independent interest to readers, we present it in a standalone section with greater generality than required for the results of this paper. For this purpose, let $P({\bf X})$ be a generic $n \times n$ symmetric positive semi-definite matrix that depends only on $\bf X$, and satisfies the projection property that $\sum_{j \in [n]} P_{ij}({\bf X})^2 \leq P_{ii}({\bf X}), \forall i \in [n]$. We treat $\bf X$ as fixed throughout the section for simplicity. We also impose the following assumptions on the data generating process (DGP).

\begin{ass}

The following assumptions hold:

\begin{enumerate}
    \item The errors $\{ \epsilon_i \}_{i \in [n]}$ are independent and mean-zero conditional on {\bf X}.
    
    \item There exists a universal constant $\underline{c}>0$ such that $\inf_{i \in [n]}\sigma_i^2 \geq \underline{c}$.
    
    \item There exists a universal constant $\overline{c}<\infty$ such that $\sup_{i \in [n]} \mathbb{E}[|\epsilon_i|^h|{\bf X}] \leq \overline{c}$ for some $h \geq 4$.
    
    \item The maximum diagonal element of $P({\bf X})$, $\tilde p_n := \max_{i \in [n]} P_{ii}({\bf X})$, satisfies ${\tilde p_n}/{K_n} = O (n^{-1})$ with $K_n := \sum_{i \in [n]} P_{ii}({\bf X})$.
\end{enumerate}
\label{ass:mean_zero}
\end{ass}

\begin{rem}
    Assumptions \ref{ass:mean_zero}.1--\ref{ass:mean_zero}.3 are standard regularity conditions and allow for general heteroskedasticity in error terms. Assumption \ref{ass:mean_zero}.4 requires the maximum leverage score $\tilde p_n$ to vanish sufficiently quickly. Consider the special case where $P({\bf X})$ is the ridge-regularized projection matrix $P_{\theta}$ with $\theta = 0$. If the diagonals of $P_\theta$ are well-balanced (e.g. $p_n \leq \overline{c} \, {r_n}/{n}$ for some constant $\overline{c} < \infty$), then noting $K_\theta = \tr(P_\theta) = r_n$, we obtain $p_n/K_\theta = p_n / r_n \lesssim 1/n$.
\end{rem}

\begin{thm}[Strong Approximation for Quadratic Forms]
Suppose Assumption~\ref{ass:mean_zero} holds. Let \( \{ \mathcal{E}_i \}_{i \in [n]} \) be a sequence of independent Gaussian random variables such that \( \mathcal{E}_i \overset{d}{=} \mathcal{N}(0, \sigma_i^2) \), where each \( \mathcal{E}_i \) is independent of \( \{ \epsilon_j \}_{j \in [n]} \). Suppose in addition that $$\frac{1}{K_n} \sum_{i \in [n]} \sum_{j \neq i} P_{ij}^2({\bf X}) \gtrsim 1.$$ Then for any $\zeta \in (0,1)$
\begin{align}
        \sup_{y \in \mathbb{R}}\left| \mathbb{P}_\epsilon \left( \sum_{i,j \in [n], j \neq i}  P_{ij}({\bf X}) \epsilon_i \epsilon_j  \leq y \right) - \mathbb{P}_\mathcal{E} \left( \sum_{i,j \in [n], j \neq i} P_{ij}({\bf X}) \mathcal{E}_i \mathcal{E}_j  \leq y \right) \right| \leq C_\zeta \left( \frac{\tilde p_n}{K_n} \right)^{\frac{1-\zeta}{14-2\zeta}}
\label{eq:strong-approximation:1}
\end{align}
where $C_\zeta$ depends only on $\zeta$. Furthermore, for any $\ell>0$,
\begin{align}
\sup_{y \in \mathbb{R}} \left| \mathbb{P}_\epsilon \left( \epsilon^\top P({\bf X}) \epsilon \leq y \right) 
- \mathbb{P}_\mathcal{E} \left( \mathcal{E}^\top P({\bf X}) \mathcal{E} \leq y \right) \right| \leq C_\zeta \left( \frac{\tilde p_n}{K_n} \right)^{\frac{1-\zeta}{14-2\zeta}} + C_\zeta \ell^{\frac{1-\zeta}{2}} + \overline{C}\frac{\tilde p_n}{\ell^2},
\label{eq:strong-approximation:2}
\end{align}
where $\overline{C}$ is a universal constant.
\label{theorem:strong_approx}
\end{thm}

\begin{rem} 
The term $C_\zeta \ell^{(1-\zeta)/2}$ in \eqref{eq:strong-approximation:2} arises from an anti-concentration inequality for non-Gaussian errors, which bounds the probability that the off-diagonal sum $\sum_{i \in [n]} \sum_{j \neq i} P_{ij}({\bf X}) \epsilon_i \epsilon_j$ falls within a  neighborhood of radius $\ell$ centered at $y$, uniformly over $y$; see Lemma \ref{lem:anti-concentration:main} in Appendix \ref{sec:appendix_main_thm}. The term $\tilde p_n / (\ell^2)$ captures the approximation error induced by incorporating the diagonal elements of $P({\bf X})$.
\end{rem}

If we assume that $K_n/n = o(1)$, then $\tilde p_n = O(K_n/n) = o(1)$, and we can set $\ell = (K_n/n)^{1/4} = o(1)$ in the right-hand side of \eqref{eq:strong-approximation:2}, which yields
$\tilde p_n/\ell^2=O(\sqrt{K_n/n})=o(1)$, so that the errors $\epsilon_i$ can be treated as approximately normal under such quadratic forms. For the ridge-regularized projection matrix $P_{\theta}$, it suffices to require that $r = o(n)$, since then $K_\theta/n = O(r_n/n) = o(1)$ by $r_n \leq r$. The following corollary then follows immediately.

\begin{cor} \label{cor:strong_approx}
    Suppose Assumption \ref{ass:mean_zero}.1-\ref{ass:mean_zero}.3 hold and that $p_n = O(K_\theta/n)$ where $p_n := \max_{i \in [n]} P_{\theta,ii}$, and $\sum_{i \in [n]} \sum_{j \neq i} P_{\theta,ij}^2 / K_\theta \gtrsim 1$. If we further assume $r = o(n)$, then
    \begin{align*}
        \sup_{y \in \mathbb{R}}\left| \mathbb{P}_\epsilon \left( \sum_{i,j \in [n], j \neq i}  P_{\theta,ij} \epsilon_i \epsilon_j  \leq y \right) - \mathbb{P}_\mathcal{E} \left( \sum_{i,j \in [n], j \neq i} P_{\theta,ij} \mathcal{E}_i \mathcal{E}_j  \leq y \right) \right| = o(1),
    \end{align*}
    and 
    \begin{align*}
\sup_{y \in \mathbb{R}} \left| \mathbb{P}_\epsilon \left( \epsilon^\top P_\theta \epsilon \leq y \right) 
- \mathbb{P}_\mathcal{E} \left( \mathcal{E}^\top P_\theta \mathcal{E} \leq y \right) \right| = o(1).
\end{align*}
\end{cor}

The strong approximation result implies that the distribution of $\widehat{Q}(q)$ can be approximated by a weighted chi-square distribution when $r$ is bounded (see~\eqref{eq:heuritsics_fixed_K}). This approximation is essential for establishing the validity of the critical value defined in~\eqref{eq:critical_value}. Furthermore, when $r$ is diverging, we may approximate the distribution of $\widehat{Q}(q)$ by a standard normal distribution.

\section{Asymptotic Theory} \label{section:size_control}

In this section, we show theoretically that our test has asymptotically exact size control under both fixed and random designs. In particular, our inference procedure is valid regardless of whether the number of restrictions $r$, the number of coefficients involved in the hypothesis being tested $m$, and the number of coefficients $p$ are treated as bounded or diverging, thereby providing uniform inference over a large class of DGPs.

\subsection{Fixed Design}

We first consider the fixed design. The following theorem gives general conditions for our test to have asymptotically exact size control. Random designs can then be deduced from this theorem by verifying the conditions below. Note that these conditions are not restricted to the ordinary LASSO estimator. For example, one can also replace $\widehat \beta$ and $\widehat \gamma_j, j \in [p]$ with alternative LASSO estimators, such as the adaptive LASSO estimator (e.g. \cite{zou2006adaptive}) or the conservative LASSO estimator (e.g. \cite{caner2018asymptotically}), in the definition of $\widehat{Q}(q)$, and the result here can be applied to those inference procedures as well. We provide one such extension below. 

\begin{thm}
Suppose that ${\bf X}$ is taken as fixed. Let \( \{ \Lambda_n \}_{n \in \mathbb{N}} \) be a sequence of classes of DGPs under which the following conditions hold uniformly over each \( \Lambda_n \):
\begin{enumerate}

    \item Assumption \ref{ass:mean_zero}.1-\ref{ass:mean_zero}.3 hold, $p_n = O(K_\theta/n)$ where $p_n = \max_{i \in [n]} P_{\theta,ii}$, and $r=o(n)$.
    
    \item The matrix \( R \in \mathbb{R}^{r \times p} \) has at most \( m \) non-zero column vectors, and satisfies $\lambda_{\max}(R  R^\top) \leq \overline{c} < \infty$.
    
    \item $\left(\sqrt{r} + \lVert {\bf X} (\widehat{\beta} - \beta^0 ) \rVert_{\infty} \right) \times \lVert {\bf X} (\widehat{\beta} - \beta^0 ) \rVert_2 = o_\Prob(\sqrt{n})$ and $\lVert \Delta \rVert_\infty = o_\Prob(1/\sqrt{m})$.

    \item There exists $\varpi \in (0,1)$ that is bounded away from zero, such that $\theta \gtrsim 1$ and $$\frac{1}{K_\theta} \sum_{i \in [n]} \sum_{j \neq i} P_{\theta,ij}^2 \gtrsim 1$$
    for all $n$ large enough.
\end{enumerate}

Then, under the null,
\begin{align*}
    \liminf_{n \rightarrow \infty} \inf_{\lambda_n \in \Lambda_n} \mathbb{P}_{\lambda_n} \left( \widehat{Q}(q) > \mathcal{C}_\alpha  \right) = \limsup_{n \rightarrow \infty} \sup_{\lambda_n \in \Lambda_n} \mathbb{P}_{\lambda_n} \left( \widehat{Q}(q) > \mathcal{C}_\alpha  \right) = \alpha.
\end{align*}
\label{thm:size-control_fixed_design}
\end{thm}

\begin{rem}
    We make three remarks regarding the preceding conditions.

    First, the condition on $\lambda_{\max}(RR^\top)$
    controls the magnitude of the restriction matrix $R$.
    For any full-row-rank restriction matrix
    $R\in\mathbb{R}^{r\times p}$, we may define
    $\widetilde R=(RR^\top)^{-1/2}R$ and
    $\widetilde q=(RR^\top)^{-1/2}q$.
    The null hypotheses $R\beta^0=q$ and
    $\widetilde R\beta^0=\widetilde q$ are equivalent,
    while $\widetilde R\widetilde R^\top=I_r$.
    Consequently, any full-row-rank system of linear
    restrictions admits an equivalent normalized
    representation satisfying
    $\lambda_{\max}(\widetilde R\widetilde R^\top)=1$. Throughout the construction of the test statistic, we may therefore adopt this normalized representation of the null hypothesis.

    Second, the rate conditions on
    $\lVert {\bf X}(\widehat{\beta}-\beta^0)\rVert_\infty$
    and
    $\lVert {\bf X}(\widehat{\beta}-\beta^0)\rVert_2$
    ensure the consistency of the variance estimators
    and the estimated weights of the weighted
    chi-squared distribution.
    The rate condition on $\lVert\Delta\rVert_\infty$
    ensures that the remainder term $R\Delta$
    is asymptotically negligible in $\widehat Q(q)$.

    Third, we impose a condition involving
    $\varpi\in(0,1)$ to ensure that
    $\left(
    R\widehat{\Theta}\widehat{\Sigma}
    \widehat{\Theta}^\top R^\top+\theta I_r
    \right)^{-1}$
    and $\Phi$ are well-defined in the limit.

    In the next section, we establish sufficient
    conditions under which these requirements hold
    for marginally sub-Weibull covariates and
    heteroskedastic errors satisfying finite
    polynomial moment conditions. 
\end{rem}

\begin{rem}
    While our primary focus is on establishing the asymptotic validity of the proposed test under the null hypothesis, its behavior under the alternative is also of interest. In Section~\ref{sec:power_analysis} of the Supplemental
    Appendix, we establish power properties of the proposed test under a fixed design. To maintain the focus of the main text on asymptotically exact size control, we defer this analysis to the Supplemental Appendix.
\end{rem}

\subsection{Random Design}
We now turn to the random design. Let $\widetilde{X}_i^\top \in \mathbb{R}^p$ for $i \in [n]$ be the row vectors of the design matrix $\bf X$, with $\Sigma_i := \mathbb{E}[\widetilde{X}_i \widetilde{X}_i^\top] \in \mathbb{R}^{p \times p}$. We define the population covariance matrix as $\Sigma := \sum_{i \in [n]} \Sigma_i /n \in \mathbb{R}^{p \times p}$, and $\Theta:= \Sigma^{-1} \in \mathbb{R}^{p \times p}$, which we assume exists. We also define $\eta_j:= X_j - {\bf X}_{-j} \gamma_j \in \mathbb{R}^n$ where $\gamma_j \in \arg \min_{\gamma \in \mathbb{R}^{p-1}} \mathbb{E} \lVert X_j - {\bf X}_{-j} \gamma \rVert_2^2/n$, and we assume that $\max_{i \in [n], j \in [p]} \mathbb E \left| \eta_{j}(i) \right|^l \leq \overline{c} < \infty$ for some $l \geq 4$. Finally, we define $s_0 := \left|\{j: \beta^0_j \neq 0 \}\right| \vee 1$, $s_j := \left| \{ k \neq j:\Theta_{j,k} \neq 0\} \right| $ and $s_0^* := \max_{j \in [p]} s_j \vee 1$.

\begin{ass} \label{ass:size}
    \begin{enumerate}
        \item Assumption \ref{ass:mean_zero}.1 - \ref{ass:mean_zero}.3 hold and $p_n = O_\Prob(K_\theta/n)$ where $p_n = \max_{i \in [n]} P_{\theta,ii}$. 
    
        \item The matrix \( R \in \mathbb{R}^{r \times p} \) has at most \( m \) non-zero column vectors, and satisfies $\lambda_{\max}(R  R^\top) \leq \overline{c} < \infty$ and $\tr (R \Theta R^\top)/r \geq \underline{c} > 0$.

        \item $(\tilde X_i, \eps_i)^\top \indep (\tilde X_j, \eps_j)^\top$ for $i \neq j$, and conditional on $\bf X$, $\{ \epsilon_i \}_{i \in [n]}$ are independent and mean-zero. $\tilde X_i$ are marginally sub-Weibull $(\nu, K_{n,p})$ \footnote{A random variable $Z$ is said to be sub-Weibull $(\nu, K_{n,p})$ if $\|Z\|_{\psi_\nu} \le K_{n,p}$ with $\|Z\|_{\psi_\nu} :=\inf\Bigl\{c>0:\ \mathbb{E}\,\psi_\nu\bigl(|Z|/c\bigr)\le 1\Bigr\}$, where $\psi_\nu(t):=\exp(t^\nu)-1$ for $t\ge 0$. See Appendix \ref{sec:appendix_sub_weibull} for further discussion.} for some $0 < \nu \leq 2$ and $K_{n,p} = O(1)$, $1/\lambda_{\min}(\Sigma) = O(1)$ and $\max_{j \in [p]} \Sigma_{jj} = O(1)$.

        \item The following rate conditions hold:
        \begin{enumerate}[label=(\arabic*)]
            \item $\max \{s_0 (\log(np))^{4/\nu}, (s_0^*)^2 \log(np), ms_0^* (\log(np))^{4/\nu}, m s_0^2 \log^2(np), m^2 s_0^* \log(np)\} = o(n)$, 
            \item $(\log(np))^{2/\nu - 1/2} = o(n^{1/2-1/h})$,
            \item $(\log(np))^{2/\nu - 1/2} = o(n^{1/2-1/l})$.
        \end{enumerate}
    \end{enumerate}
\end{ass}

\begin{rem}
    Assumption \ref{ass:size}.1 is the random design version of Assumption \ref{ass:mean_zero} when $\bf X$ is considered as random. The condition $\tr (R \Theta R^\top)/r \geq \underline{c} > 0$ in Assumption \ref{ass:size}.2 ensures that there are sufficiently many large eigenvalues for $R\widehat{\Theta} \widehat{\Sigma} \widehat{\Theta}^\top R^\top$, so that $\theta$ can be bounded away from zero; note that this condition is weaker than $\lambda_{\min} (\Theta) \geq \underline{c} > 0$ and $\lambda_{\min} (R R^\top) \geq \underline{c} > 0$. Assumption \ref{ass:size}.3 allows for i.n.i.d. $\tilde X_i^\top$ and quite general tail conditions for the entries of $\tilde X_i^\top$; the conditions on $\Sigma$ are natural extensions of standard conditions on the population covariance matrix of $\bf X$ (e.g. \cite{van-de-geer(2014)}) to i.n.i.d. data, and note that $1/\lambda_{\min}(\Sigma) = O(1)$ is sufficient for $\Sigma$ to satisfy the restricted eigenvalue condition. The rate conditions in Assumption \ref{ass:size}.4 extend the results in \cite{van-de-geer(2014)} for $\nu = 2$. Our results generalize beyond those in \cite{van-de-geer(2014)} by not restricting the joint distribution of $\tilde X_i^\top$ (e.g. joint sub-Gaussian), and we allow $\eps$ to have only polynomial moments. \cite{caner2018asymptotically} consider the even more general case when both $\bf X$ and $\eps$ only have polynomial moments, so that their rate conditions can only allow $p$ to be polynomial of $n$, while we still allow $p$ to be exponential in polynomial of $n$. Finally, we allow $m$ to diverge at a slower rate than $n$. Consider the important special case where $s_0^*=O(1)$, $s_0=O(n^{1/4})$, and $p=an$ for some $a>0$, then we allow $m$ to be at most $o(n^{1/2})$ up to a logarithmic factor.
\end{rem}

\begin{prop}
    Suppose Assumption \ref{ass:size} holds and $\lambda_j \asymp \sqrt{\log(np)/n}$ (uniformly in $j \in [p]$) are suitably chosen. There exist $c_{\varpi} \in (0,1]$ and $\underline{c} > 0$ such that, with probability approaching one, $\theta \geq \underline{c}$ and $\Phi \geq \underline{c}$ for all $\varpi \in (0,c_{\varpi})$. 
    \label{prop:ridge_parameter}
\end{prop}

\begin{rem}
Proposition \ref{prop:ridge_parameter} shows how we can choose $\varpi$ such that $\left( R\widehat{\Theta} \widehat{\Sigma} \widehat{\Theta}^\top R^\top + \theta I_r \right)^{-1}$ and $\Phi$ are well-defined in the limit under random design. We discuss in Appendix \ref{subsec:ridge_parameter} (see Remark \ref{rem:varpi}) how to choose $\varpi$ in a data-driven way. Our simulations suggest that both the size and the power of our test statistic are actually not very sensitive to the choice of $\varpi$. 
\end{rem}

Before establishing the validity of our test under general random-design settings, we first provide an intermediate result showing that Assumption~\ref{ass:size}, together with suitable choices of the tuning parameters, is sufficient for the conditions of Theorem~\ref{thm:size-control_fixed_design} to hold.
\begin{prop}[Sufficient conditions for Theorem~\ref{thm:size-control_fixed_design}]
    Suppose Assumption~\ref{ass:size} holds,
    $\lambda_0 \asymp \sqrt{\log(np)/n}$,
    $\lambda_j \asymp \sqrt{\log(np)/n}$ uniformly over $j \in [p]$,
    and $\varpi \in (0,1)$ is suitably chosen.
    Then all the conditions of
    Theorem~\ref{thm:size-control_fixed_design} are satisfied.
    \label{prop:sufficient_conditions_for_fixed_design}
\end{prop}
We are now ready to establish the main result for our test under random design, which follows directly from Proposition~\ref{prop:sufficient_conditions_for_fixed_design}.

\begin{thm}
Let \( \{ \Lambda_n \}_{n \in \mathbb{N}} \) be a sequence of classes of DGPs. Suppose that $\lambda_0 \asymp \sqrt{\log(np)/n}$, $\lambda_j \asymp \sqrt{\log(np)/n}$ (uniformly in $j \in [p]$) and $\varpi \in (0,1)$ are suitably chosen. If Assumption \ref{ass:size} holds uniformly over each \( \Lambda_n \), then under the null,
\begin{align*}
    \liminf_{n \rightarrow \infty} \inf_{\lambda_n \in \Lambda_n} \mathbb{P}_{\lambda_n} \left( \widehat{Q}(q) > \mathcal{C}_\alpha  \right) = \limsup_{n \rightarrow \infty} \sup_{\lambda_n \in \Lambda_n} \mathbb{P}_{\lambda_n} \left( \widehat{Q}(q) > \mathcal{C}_\alpha  \right) = \alpha.
\end{align*}
\label{thm:size-control_random_design}
\end{thm}

\subsection{Extension to Conservative LASSO}
\label{subsec:conservative_LASSO}

In practice, it may be desirable to employ alternative LASSO estimators in
the construction of our test statistic to improve its finite-sample performance. In this section, we consider such an extension using the
conservative LASSO estimator proposed by \cite{caner2018asymptotically}. We let $\widehat\beta_j$ denote the $j$th element of $\widehat\beta$ from the original LASSO estimator \eqref{eq:beta-hat}, and consider the conservative LASSO
\begin{align*}
    \widehat{\beta}_{\rm C} \in \arg \min_{\beta \in \mathbb{R}^p} \frac{1}{n} \lVert Y - {\bf X}\beta \rVert_2^2 + 2\lambda_0 \sum_{j \in [p]}\hat w_j |\beta_j|,
\end{align*}
where $\hat w_j = \lambda_{{\rm prec},0}/(|\widehat{\beta}_j| \vee \lambda_{{\rm prec},0})$ for $j \in [p]$, and $\lambda_{{\rm prec},0}$ is a prespecified tuning parameter.\footnote{We follow \cite{caner2018asymptotically} and set $\lambda_{\mathrm{prec},j}=\frac{9}{4}\lambda_n$, where $\lambda_n$ denotes the penalty parameter for the ordinary LASSO. In our implementation, $\lambda_n$ is either chosen to satisfy $\lambda_n \asymp \sqrt{\log(np)/n}$ or selected using BIC or GIC; see section \ref{sec:sim:description_of_test} for more details.} Similarly, define
\begin{align*}
    \widehat{\gamma}_{{\rm C},j} \in \argmin_{\gamma \in \mathbb{R}^{p-1}} \frac{1}{n} \lVert X_j - {\bf X}_{-j} \gamma  \rVert_2^2+ 2\lambda_j \sum_{k\neq j} \hat w_{j,k} |\gamma_k|,
\end{align*}
where $\hat w_{j,k} = \lambda_{{\rm prec},j}/(|\widehat{\gamma}_{j}(k)| \vee \lambda_{{\rm prec},j})$ and we recall that $\widehat{\gamma}_{j}(k)$ is the $k$th element of $\widehat\gamma_j \in \mathbb{R}^{p-1}$ defined in \eqref{eq:gamma_j}.

Following this, we modify $\widehat{\Theta}$, $\widehat{b}$, and $\theta$ in the test statistic as described in Section~\ref{sec:Test_Statistic}, and modify $\{\widehat{\sigma}_i^2\}_{i\in[n]}$, $\{\widetilde{w}_{i,n}\}_{i\in[r_n]}$, and $\widehat{\Phi}$ in the construction of the critical value in Section~\ref{sec:CV} accordingly. We denote the resulting test statistic and critical value by $\widehat{Q}^{\rm C}(q)$ and $\mathcal{C}^{\rm C}_\alpha$, respectively.

\begin{thm}
Let \( \{ \Lambda_n \}_{n \in \mathbb{N}} \) be a sequence of classes of DGPs. Suppose that $\lambda_{{\rm prec},0} \asymp \lambda_0 \asymp \sqrt{\log(np)/n}$, $\lambda_{{\rm prec},j} \asymp \lambda_j \asymp \sqrt{\log(np)/n}$ (uniformly in $j \in [p]$) and $\varpi \in (0,1)$ are suitably chosen. If Assumption \ref{ass:size} holds uniformly over each \( \Lambda_n \) and $s_0^* = O(1)$, then under the null,
\begin{align*}
    \liminf_{n \rightarrow \infty} \inf_{\lambda_n \in \Lambda_n} \mathbb{P}_{\lambda_n} \left( \widehat{Q}^{\rm C}(q) > \mathcal{C}^{\rm C}_\alpha  \right) = \limsup_{n \rightarrow \infty} \sup_{\lambda_n \in \Lambda_n} \mathbb{P}_{\lambda_n} \left( \widehat{Q}^{\rm C}(q) > \mathcal{C}^{\rm C}_\alpha  \right) = \alpha.
\end{align*}
\label{thm:size-control_random_design_extension}
\end{thm}

\section{Monte Carlo Study}
\label{sec:sim}

In this section, we present Monte Carlo simulation results comparing our proposed test (LMTW) with alternative tests in the literature based on the desparsified LASSO framework. In particular, we compare LMTW against the homoskedastic Wald test based on \cite{van-de-geer(2014)} and a heteroskedasticity-robust Wald test based on the conservative LASSO of \cite{caner2018asymptotically}, extended to accommodate $r \geq 1$. For fixed $r$, the theoretical results in these papers can be extended to establish the validity of the corresponding tests under their respective assumptions. It is less clear, however, whether and how their theoretical arguments can be extended to the setting where $r \rightarrow \infty$.

We also consider a variant of LMTW in which the ordinary LASSO estimators are replaced by the conservative LASSO estimators of \cite{caner2018asymptotically}. All reported results are based on $1{,}000$ Monte Carlo replications.

\subsection{Description of Tests}
\label{sec:sim:description_of_test}

We compare the finite-sample performance of our proposed test with two benchmarks, namely homoskedastic and heteroskedasticity-robust Wald tests. Our tests consider the Bayesian Information Crietiron (BIC) and Generalized Information Criterion (GIC). 

\begin{enumerate}

    \item {\bf Wald-HE}: The heteroskedasticity-robust Wald test of \cite{caner2018asymptotically} is implemented using ordinary LASSO estimation. The tuning parameter selected for the initial regression is also used for the node-wise regression.

    \item {\bf Wald-HO}: The homoskedastic Wald test is implemented using the same approximate inverse as Wald-HE and the scaled LASSO estimator of \cite{sun2012scaled} for the initial regression following \cite{van-de-geer(2014)}.

    \item {\bf LMTW}: Our proposed test defined in \eqref{eq:rejection-criteria} is implemented using ordinary LASSO estimation for both the initial and node-wise regressions. Tuning parameters are seperately selected using BIC and GIC.

    \item {\bf LMTW-CL}: Our proposed test defined in \eqref{eq:rejection-criteria} is implemented using conservative LASSO estimation described in Section \ref{subsec:conservative_LASSO} for both the initial and node-wise regressions. Tuning parameters are seperately selected using BIC and GIC.
    
\end{enumerate}

Implementation details such as information criteria, tuning parameter selection and procedures, conservative LASSO weights, and critical value computations are found in Appendix \ref{sec:appendix_implement}.

\subsection{Simulation Design}
We consider two dimensional settings: $(n,p)=(200,100)$ and $(n,p)=(200,300)$. In both settings, the number of nonzero elements of $\beta^0$ is fixed at $s_0=5$. Each nonzero coefficient is set equal to $1$, with their locations equally spaced across the coordinates of $\beta^0$.

We consider two restriction matrices. The first imposes restrictions on one nonzero and four zero elements of $\beta^0$, so that $r=m=5$. The second imposes restrictions on all five nonzero and twenty zero elements of $\beta^0$, so that $r=m=25$.

For the power analysis, we consider deviations from the null by testing against
\[
    \widetilde q
    =
    q
    +
    \delta
    \left(
    \underbrace{\frac{1}{\sqrt r},\ldots,
    \frac{1}{\sqrt r}}_{r}
    \right)^\top,
    \qquad
    \delta\in\{0.2,0.4\}.
\]

We focus on the random-design setting, so that both ${\bf X}$ and $\varepsilon$ are independently resampled in each Monte Carlo replication. For the design matrix ${\bf X}$, we consider Gaussian  and exponential designs, both standardized to have zero mean and unit variance, with covariance structure
\[
    \Sigma_{ij}=(0.5)^{|i-j|},
    \qquad i,j\in[p].
\]
For the errors, we consider Gaussian and Student's $t$ distributions with five degrees of freedom, both standardized to have zero mean and unit variance. We consider both homoskedastic and heteroskedastic specifications. Under heteroskedasticity, the conditional error variance is given by
\[
    \sigma_i^2
    =
    \frac{0.5+|{\bf X}_{i,1}|}
         {0.5+\mathbb{E}|{\bf X}_{i,1}|},
\]
so that $\mathbb{E}[\sigma_i^2]=1$, ensuring that the unconditional error variance remains equal to one, as in the homoskedastic specification.

Tables \ref{tab:size_power_n200_p100_r5_homo}-\ref{tab:size_power_n200_p100_r25_hetero} report the simulation results for $(n,p)=(200,100)$. The results for $(n,p)=(200,300)$ are reported in Appendix \ref{sec:additional_simulation}, and the pattern is similar to the one we observe here. The proposed test statistic delivers substantially better size control than the homoskedastic Wald benchmark under heteroskedasticity, especially when the design is non-Gaussian. Under homoskedasticity and small $r$, the homoskedastic Wald test can be competitive, but its size deteriorates more severely as the number of restrictions increases. For the heteroskedasticity-robust Wald test, the size control is competitive with our test statistic when $r$ is small under both homoskedasticity and heteroskedasticity. However, when $r$ is large, it has even more severe size distortion than the homoskedastic Wald test. This suggests that estimating and inverting the full sandwich covariance matrix can be very unstable in this regime. We also note that the use of conservative LASSO in our test statistic can improve size control when $r$ is large. In addition, our test statistic with ordinary LASSO has better power than the homoskedastic Wald test under homoskedasticity and small $r$ where the sizes of all tests are close to each other. Our proposed test with conservative LASSO exhibits higher rejection rates under alternatives than under ordinary LASSO, although these rates remain below those of the heteroskedasticity-robust Wald test under both homoskedasticity and heteroskedasticity.\footnote{Note, however, that the heteroskedasticity-robust Wald test exhibits greater size distortions when $r$ is large. Rejection rates under alternatives are not adjusted for differences in empirical size and therefore should not be interpreted as size-adjusted power comparisons}. Finally, note that the power of our test statistic is slightly better when BIC is used to select the tuning parameters. Given our emphasis on size control, our proposed test obtains smaller size distortions than the benchmark Wald tests in heteroskedastic designs with large $r$, with the conservative LASSO implementation generally reducing distortions further.

\begin{table}[H]
\centering
\begin{threeparttable}
\caption{Empirical Size and Power: $(n,p,r)=(200,100,5)$, Homoskedastic Errors}
\label{tab:size_power_n200_p100_r5_homo}
\scriptsize
\setlength{\tabcolsep}{2.7pt}
\begin{tabular}{lcccccccccccc}
\toprule\toprule
& \multicolumn{3}{c}{D1} & \multicolumn{3}{c}{D2} & \multicolumn{3}{c}{D3} & \multicolumn{3}{c}{D4} \\
\cmidrule(lr){2-4}\cmidrule(lr){5-7}\cmidrule(lr){8-10}\cmidrule(lr){11-13}
& $\delta=0$ & $\delta=0.2$ & $\delta=0.4$ & $\delta=0$ & $\delta=0.2$ & $\delta=0.4$ & $\delta=0$ & $\delta=0.2$ & $\delta=0.4$ & $\delta=0$ & $\delta=0.2$ & $\delta=0.4$ \\
\midrule
Wald-HO (BIC)
& 0.066 & 0.291 & 0.953 & 0.062 & 0.309 & 0.950 & 0.056 & 0.327 & 0.950 & 0.067 & 0.329 & 0.923 \\
Wald-HO (GIC)
& 0.071 & 0.266 & 0.934 & 0.070 & 0.278 & 0.935 & 0.060 & 0.299 & 0.941 & 0.074 & 0.305 & 0.907 \\
\addlinespace
Wald-HE (BIC)
& 0.059 & 0.503 & 0.992 & 0.051 & 0.554 & 0.994 & 0.052 & 0.564 & 0.992 & 0.065 & 0.593 & 0.993 \\
Wald-HE (GIC)
& 0.061 & 0.476 & 0.986 & 0.050 & 0.522 & 0.991 & 0.055 & 0.530 & 0.989 & 0.067 & 0.555 & 0.987 \\
\addlinespace
LMTW (BIC)
& 0.060 & 0.388 & 0.976 & 0.059 & 0.385 & 0.979 & 0.060 & 0.386 & 0.976 & 0.062 & 0.393 & 0.958 \\
LMTW (GIC)
& 0.069 & 0.342 & 0.970 & 0.062 & 0.337 & 0.961 & 0.061 & 0.349 & 0.967 & 0.062 & 0.349 & 0.939 \\
\addlinespace
LMTW-CL (BIC)
& 0.056 & 0.443 & 0.991 & 0.043 & 0.467 & 0.993 & 0.046 & 0.475 & 0.986 & 0.056 & 0.465 & 0.981 \\
LMTW-CL (GIC)
& 0.056 & 0.417 & 0.986 & 0.043 & 0.449 & 0.989 & 0.046 & 0.457 & 0.984 & 0.058 & 0.446 & 0.975 \\
\bottomrule
\end{tabular}
\begin{tablenotes}[flushleft]
\scriptsize
\item \textit{Notes:} The nominal size is $\alpha=0.05$. The column $\delta=0$ reports empirical size, while $\delta=0.2$ and $\delta=0.4$ report empirical power. D1: Gaussian $X$ and Gaussian errors; D2: centered-exponential $X$ and Gaussian errors; D3: Gaussian $X$ and standardized $t_5$ errors; D4: centered-exponential $X$ and standardized $t_5$ errors. The errors are homoskedastic in all four designs. Wald-HO is the homoskedastic Wald test; Wald-HE is the heteroskedasticity-robust Wald test using conservative-LASSO; LMTW and LMTW-CL use ordinary/plain- and conservative-LASSO respectively.
\end{tablenotes}
\end{threeparttable}
\end{table}

\begin{table}[H]
\centering
\begin{threeparttable}
\caption{Empirical Size and Power: $(n,p,r)=(200,100,5)$, Heteroskedastic Errors}
\label{tab:size_power_n200_p100_r5_hetero}
\scriptsize
\setlength{\tabcolsep}{2.7pt}
\begin{tabular}{lcccccccccccc}
\toprule\toprule
& \multicolumn{3}{c}{D1} & \multicolumn{3}{c}{D2} & \multicolumn{3}{c}{D3} & \multicolumn{3}{c}{D4} \\
\cmidrule(lr){2-4}\cmidrule(lr){5-7}\cmidrule(lr){8-10}\cmidrule(lr){11-13}
& $\delta=0$ & $\delta=0.2$ & $\delta=0.4$ & $\delta=0$ & $\delta=0.2$ & $\delta=0.4$ & $\delta=0$ & $\delta=0.2$ & $\delta=0.4$ & $\delta=0$ & $\delta=0.2$ & $\delta=0.4$ \\
\midrule
Wald-HO (BIC)
& 0.118 & 0.373 & 0.936 & 0.231 & 0.470 & 0.951 & 0.100 & 0.385 & 0.945 & 0.246 & 0.504 & 0.926 \\
Wald-HO (GIC)
& 0.124 & 0.345 & 0.915 & 0.247 & 0.451 & 0.936 & 0.106 & 0.355 & 0.924 & 0.260 & 0.482 & 0.907 \\
\addlinespace
Wald-HE (BIC)
& 0.064 & 0.458 & 0.973 & 0.095 & 0.540 & 0.974 & 0.067 & 0.514 & 0.979 & 0.107 & 0.592 & 0.970 \\
Wald-HE (GIC)
& 0.065 & 0.433 & 0.964 & 0.099 & 0.505 & 0.965 & 0.065 & 0.486 & 0.967 & 0.114 & 0.566 & 0.963 \\
\addlinespace
LMTW (BIC)
& 0.062 & 0.307 & 0.934 & 0.083 & 0.287 & 0.835 & 0.056 & 0.333 & 0.934 & 0.091 & 0.313 & 0.821 \\
LMTW (GIC)
& 0.060 & 0.277 & 0.889 & 0.079 & 0.243 & 0.798 & 0.056 & 0.285 & 0.904 & 0.090 & 0.287 & 0.783 \\
\addlinespace
LMTW-CL (BIC)
& 0.051 & 0.355 & 0.952 & 0.087 & 0.314 & 0.872 & 0.045 & 0.392 & 0.955 & 0.088 & 0.331 & 0.838 \\
LMTW-CL (GIC)
& 0.050 & 0.341 & 0.943 & 0.086 & 0.304 & 0.867 & 0.045 & 0.373 & 0.948 & 0.088 & 0.323 & 0.825 \\
\bottomrule
\end{tabular}
\begin{tablenotes}[flushleft]
\scriptsize
\item \textit{Notes:} The nominal size is $\alpha=0.05$. The column $\delta=0$ reports empirical size, while $\delta=0.2$ and $\delta=0.4$ report empirical power. D1: Gaussian $X$ and Gaussian errors; D2: centered-exponential $X$ and Gaussian errors; D3: Gaussian $X$ and standardized $t_5$ errors; D4: centered-exponential $X$ and standardized $t_5$ errors. The errors are heteroskedastic in all four designs. Wald-HO is the homoskedastic Wald test; Wald-HE is the heteroskedasticity-robust Wald test using conservative-LASSO; LMTW and LMTW-CL use ordinary/plain-- and conservative-LASSO respectively.
\end{tablenotes}
\end{threeparttable}
\end{table}

\begin{table}[H]
\centering
\begin{threeparttable}
\caption{Empirical Size and Power: $(n,p,r)=(200,100,25)$, Homoskedastic Errors}
\label{tab:size_power_n200_p100_r25_homo}
\scriptsize
\setlength{\tabcolsep}{2.7pt}
\begin{tabular}{lcccccccccccc}
\toprule\toprule
& \multicolumn{3}{c}{D1} & \multicolumn{3}{c}{D2} & \multicolumn{3}{c}{D3} & \multicolumn{3}{c}{D4} \\
\cmidrule(lr){2-4}\cmidrule(lr){5-7}\cmidrule(lr){8-10}\cmidrule(lr){11-13}
& $\delta=0$ & $\delta=0.2$ & $\delta=0.4$ & $\delta=0$ & $\delta=0.2$ & $\delta=0.4$ & $\delta=0$ & $\delta=0.2$ & $\delta=0.4$ & $\delta=0$ & $\delta=0.2$ & $\delta=0.4$ \\
\midrule
Wald-HO (BIC)
& 0.147 & 0.267 & 0.734 & 0.149 & 0.241 & 0.772 & 0.137 & 0.248 & 0.736 & 0.137 & 0.256 & 0.764 \\
Wald-HO (GIC)
& 0.172 & 0.271 & 0.712 & 0.173 & 0.245 & 0.737 & 0.162 & 0.261 & 0.699 & 0.164 & 0.266 & 0.730 \\
\addlinespace
Wald-HE (BIC)
& 0.199 & 0.521 & 0.976 & 0.192 & 0.551 & 0.975 & 0.185 & 0.581 & 0.976 & 0.177 & 0.590 & 0.990 \\
Wald-HE (GIC)
& 0.214 & 0.510 & 0.963 & 0.202 & 0.526 & 0.970 & 0.195 & 0.571 & 0.973 & 0.192 & 0.567 & 0.988 \\
\addlinespace
LMTW (BIC)
& 0.092 & 0.202 & 0.753 & 0.089 & 0.171 & 0.748 & 0.082 & 0.189 & 0.721 & 0.071 & 0.164 & 0.726 \\
LMTW (GIC)
& 0.096 & 0.177 & 0.695 & 0.087 & 0.151 & 0.675 & 0.086 & 0.164 & 0.649 & 0.069 & 0.137 & 0.642 \\
\addlinespace
LMTW-CL BIC
& 0.071 & 0.209 & 0.840 & 0.069 & 0.193 & 0.841 & 0.056 & 0.184 & 0.805 & 0.050 & 0.171 & 0.807 \\
LMTW-CL (GIC)
& 0.069 & 0.198 & 0.818 & 0.071 & 0.179 & 0.822 & 0.058 & 0.179 & 0.777 & 0.051 & 0.167 & 0.784 \\
\bottomrule
\end{tabular}
\begin{tablenotes}[flushleft]
\scriptsize
\item \textit{Notes:} The nominal size is $\alpha=0.05$. The column $\delta=0$ reports empirical size, while $\delta=0.2$ and $\delta=0.4$ report empirical power. D1: Gaussian $X$ and Gaussian errors; D2: centered-exponential $X$ and Gaussian errors; D3: Gaussian $X$ and standardized $t_5$ errors; D4: centered-exponential $X$ and standardized $t_5$ errors. The errors are homoskedastic in all four designs. Wald-HO is the homoskedastic Wald test; Wald-HE is the heteroskedasticity-robust Wald test using conservative-LASSO; LMTW and LMTW-CL use ordinary/plain-- and conservative-LASSO respectively.
\end{tablenotes}
\end{threeparttable}
\end{table}

\begin{table}[H]
\centering
\begin{threeparttable}
\caption{Empirical Size and Power: $(n,p,r)=(200,100,25)$, Heteroskedastic Errors}
\label{tab:size_power_n200_p100_r25_hetero}
\scriptsize
\setlength{\tabcolsep}{2.7pt}
\begin{tabular}{lcccccccccccc}
\toprule\toprule
& \multicolumn{3}{c}{D1} & \multicolumn{3}{c}{D2} & \multicolumn{3}{c}{D3} & \multicolumn{3}{c}{D4} \\
\cmidrule(lr){2-4}\cmidrule(lr){5-7}\cmidrule(lr){8-10}\cmidrule(lr){11-13}
& $\delta=0$ & $\delta=0.2$ & $\delta=0.4$ & $\delta=0$ & $\delta=0.2$ & $\delta=0.4$ & $\delta=0$ & $\delta=0.2$ & $\delta=0.4$ & $\delta=0$ & $\delta=0.2$ & $\delta=0.4$ \\
\midrule
Wald-HO (BIC)
& 0.184 & 0.291 & 0.778 & 0.244 & 0.361 & 0.831 & 0.168 & 0.278 & 0.758 & 0.237 & 0.364 & 0.834 \\
Wald-HO (GIC)
& 0.213 & 0.302 & 0.747 & 0.265 & 0.364 & 0.791 & 0.194 & 0.284 & 0.732 & 0.268 & 0.367 & 0.815 \\
\addlinespace
Wald-HE (BIC)
& 0.215 & 0.612 & 0.984 & 0.252 & 0.686 & 0.993 & 0.209 & 0.651 & 0.992 & 0.249 & 0.742 & 0.997 \\
Wald-HE (GIC)
& 0.221 & 0.597 & 0.981 & 0.264 & 0.666 & 0.989 & 0.221 & 0.632 & 0.983 & 0.270 & 0.721 & 0.994 \\
\addlinespace
LMTW (BIC)
& 0.087 & 0.168 & 0.722 & 0.096 & 0.183 & 0.668 & 0.074 & 0.164 & 0.676 & 0.066 & 0.158 & 0.641 \\
LMTW (GIC)
& 0.091 & 0.152 & 0.653 & 0.097 & 0.156 & 0.591 & 0.073 & 0.142 & 0.614 & 0.067 & 0.138 & 0.576 \\
\addlinespace
LMTW-CL (BIC)
& 0.069 & 0.184 & 0.786 & 0.083 & 0.194 & 0.734 & 0.043 & 0.151 & 0.754 & 0.053 & 0.157 & 0.709 \\
LMTW-CL (GIC)
& 0.070 & 0.175 & 0.764 & 0.079 & 0.188 & 0.703 & 0.045 & 0.145 & 0.739 & 0.058 & 0.144 & 0.682 \\
\bottomrule
\end{tabular}
\begin{tablenotes}[flushleft]
\scriptsize
\item \textit{Notes:} The nominal size is $\alpha=0.05$. The column $\delta=0$ reports empirical size, while $\delta=0.2$ and $\delta=0.4$ report empirical power. D1: Gaussian $X$ and Gaussian errors; D2: centered-exponential $X$ and Gaussian errors; D3: Gaussian $X$ and standardized $t_5$ errors; D4: centered-exponential $X$ and standardized $t_5$ errors. The errors are heteroskedastic in all four designs. Wald-HO is the homoskedastic Wald test; Wald-HE is the heteroskedasticity-robust Wald test using conservative-LASSO; LMTW and LMTW-CL use ordinary/plain-- and conservative-LASSO respectively.
\end{tablenotes}
\end{threeparttable}
\end{table}

\section{Empirical Application}
\label{sec:empirical_application}

We illustrate the proposed procedure using a high-dimensional version of the Romer-Romer regression for identifying monetary policy shocks. The empirical problem is to explain changes in the federal funds rate using information available to the Federal Reserve at each Federal Open Market Committee (FOMC) meeting. Omitting information contained in the Fed's own forecasts can confound the regression residuals used to measure policy shocks. The ``narrative identification approach'' of \citet{romer1989does, romer2004new} thus regresses changes in federal funds rate on its prevailing level and Fed's real-time forecasts of macroeconomic conditions, with the residual interpreted as the monetary policy shock.

\cite{aruoba2024identifying} augment these controls with sentiment indicators constructed by scoring economic-concept keywords in the FOMC documents.\footnote{The text consists of Greenbook Parts I and II (merged and renamed Tealbook A in 2010) and the Redbook (renamed Beige Book in 1983). Greenbook Part I summarizes recent economic conditions and presents the Fed staff's economic forecasts. Greenbook part II discusses recent economic and financial developments, whereas the Redbook summarizes economic conditions across the 12 Federal Reserve Districts.} Their data therefore provide a natural high-dimensional setting for illustrating joint inference on the coefficients associated with economically meaningful groups of text-based predictors. We analyze their dataset, which covers $n=210$ FOMC meetings from October 5 1982 to October 29 2008, and contains sentiment indicators for the 296 selected economic-concept keywords.\footnote{The sentiment data is available on the authors' website: \href{https://thomas-drechsel.com/research.html}{https://thomas-drechsel.com/research.html}. The sample terminates in 2008 to avoid the zero-lower-bound period.} 


Our benchmark regression is 
\begin{align}
    \Delta i_t
    =
    \alpha + \gamma i_{t-1} + \delta^\top N_t + \zeta^\top S_t + \varepsilon_t,
    \qquad t=1,\ldots, 210,
    \label{RR_reg}
\end{align}
where $i_{t-1}$ denotes the level of the federal funds rate prior to the $t$-th FOMC announcement, $\Delta i_t$ denotes the Fed's policy decision, $N_t$ represents the 132 numerical forecasts in the FOMC documents, $S_t$ is the vector of 296 sentiment indicators, and $\varepsilon_t$ is conceptually the monetary policy shock. Results for $p=858$, constructed by additionally including quadratic terms of $N_t$ and $S_t$, are reported in Appendix \ref{ERR_RC}. Let $\mathbf{X}_t=(1,i_{t-1},N_t^\top, S_t^\top)^\top$, $\mathbf{X}=(\mathbf{X}_1,\ldots,\mathbf{X}_n)^\top$, $\beta^0 = (\alpha, \gamma, \delta^\top, \zeta^\top)^\top$. We assume that $\beta^0$ is sparse, and let $\Delta i = (\Delta i_1, \ldots, \Delta i_n)^\top$, and $\varepsilon = (\varepsilon_1, \ldots, \varepsilon_n)^\top$. Then the regression (\ref{RR_reg}) can be written in the form (\ref{eq:model}):
\begin{align*}
    \Delta i = \mathbf{X} \beta^0 + \varepsilon, \qquad \text{with } (n,p) = (210, 430).
\end{align*}
We emphasize that our objective is not to use the high-dimensional linear regression to identify monetary policy shocks, but rather to conduct inference on its regression coefficients under the sparsity assumption. This extends the regression to a high-dimensional setting, from the conventional low-dimensional setting by \cite{romer2004new} and \cite{croushore2018fiscal}, and examine joint hypotheses concerning its coefficients. The error term $\varepsilon_t$ in \eqref{RR_reg}, interpreted as the monetary policy shocks, is conceptually treated as serially uncorrelated (\cite{bauer2023reassessment}). For the purpose of applying our procedure, we assume independence across FOMC meetings, which is stronger than serial uncorrelatedness. Note that substantial heteroskedasticity in $\varepsilon_t$ has been documented by \citet{coibion2012effects}; see also \citet{schlaak2023monetary}. These features provide a natural setting for applying heteroskedasticity-robust high-dimensional inference. 


To provide empirical context for the regularity conditions underlying our test, Appendix \ref{cali_ERR} reports tail diagnostics for the standardized covariates and moment diagnostics based on the LASSO residuals. We adopt $(\nu,h)=(1,4)$ as conservative working values, corresponding to marginal sub-exponential covariates and the minimum fourth-moment requirement for the errors. These choices avoid imposing stronger distributional assumptions, such as sub-Gaussian covariates or higher-order error moments. We emphasize that the finite-sample diagnostics do not verify the population assumptions or the asymptotic rate conditions in Assumption \ref{ass:size}. They instead provide descriptive evidence to accompany our empirical application, with the robustness of the findings further examined in Appendix \ref{ERR_RC}.

Let $G$ be one of the six clusters of economic concepts: ``oil'', ``inflation'', ``labor'', ``credit'', ``demand'', and ``housing'', with the constituent keywords of each cluster detailed in Table \ref{tab:lag0-linear-main-concepts}. We apply our test to examine which clusters of sentiment indicators are jointly significant for the Fed's monetary policy decisions. For example, when $G = \text{``oil''}$, the parameter of interest is the subvector $\beta^0_G$ consisting of the regression coefficients associated with the sentiment indicators for the keywords ``crude oil'', ``oil imports'', ``oil prices'', and ``petroleum'', so that $m = |G| = 4$. Recalling the Introduction, our testing problem of interest in subvector form is
\begin{align*}
 R_G \beta^0_{G} = 0 
 \qquad  vs. \qquad 
 R_G \beta^0_{G} \neq 0.
\end{align*}
Equivalently, the null hypothesis can be written as $R\beta^0=0$, where $R$ is obtained by embedding $R_G$ into an $\mathbb{R}^{r \times p}$ matrix and padding the columns corresponding to coefficients outside $G$ with zeros. Our proposed test can then be applied directly. We consider two choices of the restriction matrix:
\begin{align*}
R_G
= I_m \qquad \text{or} \qquad R_G = [ - \mathbb{1}_{m-1}, I_{m-1}]
\end{align*}
where $\mathbb{1}_{m-1}$ is the vector of ones of dimension $m-1$. The first choice, $R_G=I_m$, tests the joint null that all coefficients associated with cluster $G$ are equal to zero. The second choice, $R_G=[-\mathbb{1}_{m-1},I_{m-1}]$, tests whether all coefficients within the cluster are equal, thereby examining whether the sentiment indicators within a given cluster have homogeneous effects on the monetary policy outcome.

\begin{table}[!t]
\centering
\footnotesize
\setlength{\tabcolsep}{5.0pt}
\renewcommand{\arraystretch}{1.05}
\begin{threeparttable}
\caption{Economic Concept Keywords underlying Sentiment Indicators in Each Cluster $G$}
\label{tab:lag0-linear-main-concepts}
\begin{tabular}{@{}>{\raggedright\arraybackslash}p{1.65cm}>{\raggedright\arraybackslash}p{14.25cm}@{}}
\toprule
Group $G$ & Related concepts \\
\midrule
credit & ``banks'', ``credit'', ``credit standards'', ``financing'', ``lending'', ``liquidity'', ``loan demand'', ``loans'' \\
demand & ``aggregate demand'', ``business activity'', ``consumption'', ``domestic demand'', ``economic growth'', ``final sales'', ``industrial production'', ``investment'', ``manufacturing activity'', ``output gap'', ``retail sales'' \\
housing & ``home prices'', ``home sales'', ``household debt'', ``household net worth'', ``housing'', ``housing activity'', ``housing starts'', ``mortgage'', ``residential investment'' \\
inflation & ``consumer prices'', ``core inflation'', ``headline inflation'', ``inflation'', ``inflation expectations'', ``inflationary pressures'', ``price pressures'', ``producer prices'', ``unit labor cost'', ``wage pressures'' \\
labor & ``employment'', ``jobs'', ``labor force participation'', ``labor market'', ``payrolls'', ``unemployment'', ``wage pressures'', ``wages'', ``workers'' \\
oil & ``crude oil'', ``oil imports'', ``oil prices'', ``petroleum'' \\
\bottomrule
\end{tabular}
\begin{tablenotes}[flushleft]
\scriptsize
\item Notes: The concepts are from \cite{aruoba2024identifying}; they are single-, double-, or triple-word expressions extracted from FOMC documents and scored by a dictionary-based procedure.
\end{tablenotes}
\end{threeparttable}
\end{table}

Table \ref{tab:lag0-linear-main-p429} summarizes our empirical findings. The first $p$-value column shows that the joint null of zero coefficients is rejected at the 5\% level for the credit, demand, and housing clusters across both ordinary/plain and conservative LASSO implementations and under both tuning criteria. These results indicate that sentiment indicators in these clusters are jointly associated with monetary policy decisions conditional on the remaining controls. The evidence is particularly strong for demand, with the 1\% significance level in all four implementations. Viewed in the context of monetary-policy history, this pattern is consistent with the sample period, which spans Volcker's disinflation legacy, Greenspan's risk-management approach, and Bernanke's response to mounting housing and credit stress. Against this backdrop, the joint significance of demand, credit, and housing suggests that qualitative information on real activity and financial conditions contains incremental information beyond that captured by the numerical staff forecasts $N_t$. By contrast, the insignificant results for concept groupings, inflation and oil, need not imply that these factors were unimportant for policy; a possibility is that if they were already central to the Fed's forecasting framework, then less incremental information would remain in the associated text. Labor presents a less robust pattern: BIC selects ``jobs'' but yields no joint significance, whereas GIC selects no individual labor concept yet strongly rejects the joint null. This contrast highlights the distinction between variable selection and joint inference, suggesting that conclusions regarding the incremental information in labor market sentiment are sensitive to the chosen regularization criterion. Importantly, these findings illustrate how joint inference can reveal statistically significant associations within groups of predictors even when individual variables are not selected by the initial penalized estimator.

The coefficient-equality tests in the second $p$-value column complement this analysis by examining whether the coefficients associated with sentiment indicators are equal within each cluster. For demand, the equality null is rejected at the 5\% level in all four implementations. Thus, not only is demand sentiment jointly relevant, but different dimensions of demand appear to carry different policy implications. Credit shows a similar, though less robust, pattern, with equality rejected at the 5\% level under GIC but not under BIC. Housing is more distinctive. The cluster is jointly significant in all four implementations, yet coefficient equality is not rejected. This finding is consistent with the possibility that the housing indicators capture a common underlying housing market signal, although the failure to reject equality in coefficients does not establish that these effects are identical.

More generally, the two tests provide complementary insights into cluster-level relevance and within-cluster heterogeneity. When both the joint-zero and coefficient-equality null hypotheses are rejected, as in the case of demand, the results suggest that the constituent sentiment concepts are jointly relevant but exert heterogeneous effects, highlighting the importance
of their composition for policy interpretation.
In contrast, joint significance accompanied by non-rejection of coefficient equality, as observed for housing, is consistent with a more parsimonious common-factor interpretation, in which the constituent sentiment concepts have similar estimated effects. For inflation and oil, the joint-zero null hypothesis is not rejected at the 5\% significance level, limiting the additional
economic insights provided by the coefficient-equality tests. Labor presents a more nuanced case: under GIC, the results indicate joint significance without detectable coefficient heterogeneity, resembling the findings for housing. However, this conclusion is sensitive to the choice of regularization criterion.

\begin{table}[H]
\centering
\footnotesize
\setlength{\tabcolsep}{3.0pt}
\renewcommand{\arraystretch}{1.0}
\begin{threeparttable}
\caption{Concept-group tests of zero coefficients and coefficient equality, $(n,p)=(210,430)$}
\label{tab:lag0-linear-main-p429}
\begin{tabular*}{\textwidth}{@{\extracolsep{\fill}}lcllll@{}}
\toprule
& & & & \multicolumn{2}{c}{$p$-value} \\
\cmidrule(l){5-6}
Group $G$ & $m$ & Method & Selected concepts
& $R_G=I_m$
& $R_G=[-\mathbb{1}_{m-1},\,I_{m-1}]$ \\
\midrule
credit & 8 & LMTW (BIC) & ``banks'', ``lending'' & $0.0271^{**}$ & $0.0716^{*}$ \\
 & & LMTW (GIC) & ``banks'' & $<0.001^{***}$ & $0.0217^{**}$ \\
 & & LMTW-CL (BIC) & ``banks'', ``lending'' & $0.0283^{**}$ & $0.0750^{*}$ \\
 & & LMTW-CL (GIC) & ``banks'' & $<0.001^{***}$ & $0.0283^{**}$ \\
\midrule
demand & 11 & LMTW (BIC) & ``business activity'', ``investment'' & $<0.001^{***}$ & $0.0123^{**}$ \\
 & & LMTW (GIC) & ``business activity'' & $<0.001^{***}$ & $0.0026^{***}$ \\
 & & LMTW-CL (BIC) & ``business activity'', ``investment'' & $<0.001^{***}$ & $0.0123^{**}$ \\
 & & LMTW-CL (GIC) & ``business activity'' & $<0.001^{***}$ & $0.0026^{***}$ \\
\midrule
housing & 9 & LMTW (BIC) & ``household net worth'' & $0.0105^{**}$ & $0.4638$ \\
 & & LMTW (GIC) & None & $<0.001^{***}$ & $0.6572$ \\
 & & LMTW-CL (BIC) & ``household net worth'' & $0.0105^{**}$ & $0.4638$ \\
 & & LMTW-CL (GIC) & None & $<0.001^{***}$ & $0.6572$ \\
\midrule
inflation & 10 & LMTW (BIC) & None & $0.5488$ & $0.5927$ \\
 & & LMTW (GIC) & None & $0.2133$ & $0.1805$ \\
 & & LMTW-CL (BIC) & None & $0.5488$ & $0.5927$ \\
 & & LMTW-CL (GIC) & None & $0.2133$ & $0.1805$ \\
\midrule
labor & 9 & LMTW (BIC) & ``jobs'' & $0.1673$ & $0.1720$ \\
 & & LMTW (GIC) & None & $0.0025^{***}$ & $0.1475$ \\
 & & LMTW-CL (BIC) & ``jobs'' & $0.1673$ & $0.1720$ \\
 & & LMTW-CL (GIC) & None & $0.0025^{***}$ & $0.1475$ \\
\midrule
oil & 4 & LMTW (BIC) & None & $0.3089$ & $0.1457$ \\
 & & LMTW (GIC) & None & $0.2420$ & $0.0592^{*}$ \\
 & & LMTW-CL (BIC) & None & $0.2998$ & $0.1592$ \\
 & & LMTW-CL (GIC) & None & $0.2229$ & $0.0623^{*}$ \\
\bottomrule
\end{tabular*}
\begin{tablenotes}[flushleft]
\scriptsize
\item Notes: For coefficient-equality tests $R_G = [-\mathbb{1}_{m-1}, I_{m-1}]$, the reference is the first
concept listed for its group in Table~\ref{tab:lag0-linear-main-concepts}. LMTW and LMTW-CL denote the proposed test implemented with ordinary and conservative LASSO, respectively. BIC and GIC denote the information criteria used to select the penalty parameters in the main and node-wise regressions. ``Selected concepts'' lists the concept selected by the main LASSO within each group. Simulated $p$-values and critical values use $\widetilde\Phi$ and 200,000 weighted chi-square draws.
Stars denote rejection at the 10\% ($^{*}$), 5\% ($^{**}$), and 1\% ($^{***}$) levels.
\end{tablenotes}
\end{threeparttable}
\end{table}

\section{Conclusion}
\label{sec:conclusion}

We propose a novel test for multiple restrictions in high-dimensional linear regression models, allowing for general heteroskedasticity and heavy-tailed covariates and errors. The method combines desparsified estimation with a ridge-regularized quadratic-form statistic and a critical value construction that remains valid under both fixed and diverging restriction dimensions. Our theoretical framework accommodates independent but non-identically distributed covariate rows with marginal sub-Weibull tails and heteroskedastic errors satisfying finite polynomial moment conditions. Simulation results demonstrate substantially improved size control relative to conventional homoskedastic and heteroskedasticity-robust Wald tests, particularly when the number of restrictions is large. Our empirical application illustrates joint inference for groups of text-based predictors in a high-dimensional monetary policy regression.

Several directions for future research remain. One possibility is to extend the proposed framework to generalized linear and nonparametric models. Another is to weaken the conditions relating the effective degrees of freedom of the projection matrix to its rank. Developing alternative critical value approximations, potentially based on resampling methods, may provide a way to relax these conditions, although establishing their validity in high-dimensional settings would require further theoretical investigation.

Overall, our results provide a unified framework for heteroskedasticity-robust joint inference in sparse high-dimensional models, without requiring practitioners to distinguish between fixed and diverging restriction dimensions when implementing the proposed procedure.


\clearpage

\bibliographystyle{chicago}
\singlespacing
\bibliography{reference}

@STRING{At = {Aussenwirtschaft}}

@STRING{MA = {Moving Average}}

@STRING{NBER = {N.B.E.R.}}

@STRING{OR = {Operations Research}}

@Book{VW96,
  author    = {van der Vaart, Aad and Wellner, Jon A},
  publisher = {Springer, New York},
  title     = {Weak Convergence and Empirical Processes},
  year      = {1996},
}

@article{crudu2021,
  title={Inference in instrumental variable models with heteroskedasticity and many instruments},
  author={Crudu, Federico and Mellace, Giovanni and S{\'a}ndor, Zsolt},
  journal={Econometric Theory},
  volume={37},
  number={2},
  pages={281--310},
  year={2021},
  publisher={Cambridge University Press}
}

@article{KSS2020,
  title={Leave-out estimation of variance components},
  author={Kline, Patrick and Saggio, Raffaele and S{\o}lvsten, Mikkel},
  journal={Econometrica},
  volume={88},
  number={5},
  pages={1859--1898},
  year={2020},
  publisher={Wiley Online Library}
}

@article{belloni2012,
  title={Sparse models and methods for optimal instruments with an application to eminent domain},
  author={Belloni, Alexandre and Chen, Daniel and Chernozhukov, Victor and Hansen, Christian},
  journal={Econometrica},
  volume={80},
  number={6},
  pages={2369--2429},
  year={2012},
  publisher={Wiley Online Library}
}

@article{chernozhukov2013,
  title={GAUSSIAN APPROXIMATIONS AND MULTIPLIER BOOTSTRAP FOR MAXIMA OF SUMS OF HIGH-DIMENSIONAL RANDOM VECTORS},
  author={Chernozhukov, Victor and Chetverikov, Denis and Kato, Kengo},
  journal={The Annals of Statistics},
  volume={41},
  number={6},
  pages={2786--2819},
  year={2013}
}

@misc{LWZ2024,
      title={A Dimension-Agnostic Bootstrap Anderson-Rubin Test For Instrumental Variable Regressions}, 
      author={Dennis Lim and Wenjie Wang and Yichong Zhang},
      year={2024},
      eprint={2412.01603},
      archivePrefix={arXiv},
      primaryClass={econ.EM},
      url={https://arxiv.org/abs/2412.01603}, 
}

@article{zhu2018linear,
  title={Linear hypothesis testing in dense high-dimensional linear models},
  author={Zhu, Yinchu and Bradic, Jelena},
  journal={Journal of the American Statistical Association},
  volume={113},
  number={524},
  pages={1583--1600},
  year={2018},
  publisher={Taylor \& Francis}
}

@article{dezeure2017high,
  title={High-dimensional simultaneous inference with the bootstrap},
  author={Dezeure, Ruben and B{\"u}hlmann, Peter and Zhang, Cun-Hui},
  journal={Test},
  volume={26},
  number={4},
  pages={685--719},
  year={2017},
  publisher={Springer}
}

@article{cattaneo2018inference,
  title={Inference in linear regression models with many covariates and heteroscedasticity},
  author={Cattaneo, Matias D and Jansson, Michael and Newey, Whitney K},
  journal={Journal of the American Statistical Association},
  volume={113},
  number={523},
  pages={1350--1361},
  year={2018},
  publisher={Taylor \& Francis}
}

@article{jochmans2022heteroscedasticity,
  title={Heteroscedasticity-robust inference in linear regression models with many covariates},
  author={Jochmans, Koen},
  journal={Journal of the American Statistical Association},
  volume={117},
  number={538},
  pages={887--896},
  year={2022},
  publisher={Taylor \& Francis}
}

@article{zhang2025novel,
  title={A novel approach of high dimensional linear hypothesis testing problem},
  author={Zhang, Zhe and Yu, Xiufan and Li, Runze},
  journal={Journal of the American Statistical Association},
  volume={120},
  number={552},
  pages={2235--2249},
  year={2025},
  publisher={Taylor \& Francis}
}

@article{shi2019linear,
  title={Linear hypothesis testing for high dimensional generalized linear models},
  author={Shi, Chengchun and Song, Rui and Chen, Zhao and Li, Runze},
  journal={Annals of Statistics},
  volume={47},
  number={5},
  pages={2671},
  year={2019}
}

@article{belloni2014inferenceRES,
  title={Inference on treatment effects after selection among high-dimensional controls},
  author={Belloni, Alexandre and Chernozhukov, Victor and Hansen, Christian},
  journal={Review of Economic Studies},
  volume={81},
  number={2},
  pages={608--650},
  year={2014},
  publisher={Oxford University Press}
}

@article{dovi2024ridge,
  title={A ridge-regularized jackknifed Anderson-Rubin test},
  author={Dov{\`\i}, Max-Sebastian and Kock, Anders Bredahl and Mavroeidis, Sophocles},
  journal={Journal of Business and Economic Statistics},
  volume={42},
  number={3},
  pages={1083--1094},
  year={2024},
  publisher={Taylor \& Francis}
}

@book{Horn-Johnson-Matrix-Analysis,
  title     = {Matrix Analysis},
  author    = {Horn, Roger A. and Johnson, Charles R.},
  edition   = {2},
  publisher = {Cambridge University Press},
  address   = {Cambridge},
  year      = {2013},
  isbn      = {9780521548236}
}

@article{caner2018asymptotically,
  title={Asymptotically honest confidence regions for high dimensional parameters by the desparsified conservative Lasso},
  author={Caner, Mehmet and Kock, Anders Bredahl},
  journal={Journal of Econometrics},
  volume={203},
  number={1},
  pages={143--168},
  year={2018},
  publisher={Elsevier}
}

@article{kuchibhotla2022moving,
  title={Moving beyond sub-Gaussianity in high-dimensional statistics: Applications in covariance estimation and linear regression},
  author={Kuchibhotla, Arun Kumar and Chakrabortty, Abhishek},
  journal={Information and Inference: A Journal of the IMA},
  volume={11},
  number={4},
  pages={1389--1456},
  year={2022},
  publisher={Oxford University Press}
}

@article{javanmard2014confidence,
  title={Confidence intervals and hypothesis testing for high-dimensional regression},
  author={Javanmard, Adel and Montanari, Andrea},
  journal={The Journal of Machine Learning Research},
  volume={15},
  number={1},
  pages={2869--2909},
  year={2014},
  publisher={JMLR. org}
}

@article{zhang2014confidence,
  title={Confidence intervals for low dimensional parameters in high dimensional linear models},
  author={Zhang, Cun-Hui and Zhang, Stephanie S},
  journal={Journal of the Royal Statistical Society Series B: Statistical Methodology},
  volume={76},
  number={1},
  pages={217--242},
  year={2014},
  publisher={Oxford University Press}
}

@article{chernozhukov2018double,
  title={Double/debiased machine learning for treatment and structural parameters},
  author={Chernozhukov, Victor and Chetverikov, Denis and Demirer, Mert and Duflo, Esther and Hansen, Christian and Newey, Whitney and Robins, James},
  journal={The Econometrics Journal},
  volume={21},
  number={1},
  pages={C1--C68},
  year={2018},
  publisher={JSTOR}
}

@article{rudelson2015small,
  title={Small ball probabilities for linear images of high-dimensional distributions},
  author={Rudelson, Mark and Vershynin, Roman},
  journal={International Mathematics Research Notices},
  volume={2015},
  number={19},
  pages={9594--9617},
  year={2015},
  publisher={Oxford University Press}
}

@article{ning2017general,
  title   = {A general theory of hypothesis tests and confidence regions for sparse high dimensional models},
  author  = {Ning, Yang and Liu, Han},
  journal = {The Annals of Statistics},
  year    = {2017},
  volume  = {45},
  number  = {1},
  pages   = {158--195},
}

@article{zhang2017simultaneous,
  title={Simultaneous inference for high-dimensional linear models},
  author={Zhang, Xianyang and Cheng, Guang},
  journal={Journal of the American Statistical Association},
  volume={112},
  number={518},
  pages={757--768},
  year={2017},
  publisher={Taylor \& Francis}
}

@article{aruoba2024identifying,
  title       = {Identifying Monetary Policy Shocks: A Natural Language Approach},
  author      = {Aruoba, S. Bora{\u{g}}an and Drechsel, Thomas},
  year        = {2026},
  journal = {American Economic Journal: Macroeconomics},
  note        = {Forthcoming.}
}

@article{romer2004new,
  title={A new measure of monetary shocks: derivation and implications},
  author={Romer, Christina D and Romer, David H},
  journal={American economic review},
  volume={94},
  number={4},
  pages={1055--1084},
  year={2004},
  publisher={American Economic Association}
}

@article{xia2018two,
  title={Two-sample tests for high-dimensional linear regression with an application to detecting interactions},
  author={Xia, Yin and Cai, Tianxi and Cai, T Tony},
  journal={Statistica Sinica},
  volume={28},
  pages={63},
  year={2018}
}

@article{sun2012scaled,
  title={Scaled sparse linear regression},
  author={Sun, Tingni and Zhang, Cun-Hui},
  journal={Biometrika},
  volume={99},
  number={4},
  pages={879--898},
  year={2012},
  publisher={Oxford University Press}
}

@article{zou2006adaptive,
  title={The adaptive lasso and its oracle properties},
  author={Zou, Hui},
  journal={Journal of the American statistical association},
  volume={101},
  number={476},
  pages={1418--1429},
  year={2006},
  publisher={Taylor \& Francis}
}

@article{croushore2018fiscal,
  title={Fiscal forecasts at the FOMC: evidence from the Greenbooks},
  author={Croushore, Dean and Van Norden, Simon},
  journal={Review of Economics and Statistics},
  volume={100},
  number={5},
  pages={933--945},
  year={2018},
  publisher={MIT Press One Rogers Street, Cambridge, MA 02142-1209, USA journals-info~…}
}

@article{coibion2012effects,
  title={Are the effects of monetary policy shocks big or small?},
  author={Coibion, Olivier},
  journal={American Economic Journal: Macroeconomics},
  volume={4},
  number={2},
  pages={1--32},
  year={2012},
  publisher={American Economic Association}
}

@article{schlaak2023monetary,
  title={Monetary policy, external instruments, and heteroskedasticity},
  author={Schlaak, Thore and Rieth, Malte and Podstawski, Maximilian},
  journal={Quantitative Economics},
  volume={14},
  number={1},
  pages={161--200},
  year={2023},
  publisher={Wiley Online Library}
}

@article{bauer2023reassessment,
  title={A reassessment of monetary policy surprises and high-frequency identification},
  author={Bauer, Michael D and Swanson, Eric T},
  journal={NBER Macroeconomics Annual},
  volume={37},
  number={1},
  pages={87--155},
  year={2023},
  publisher={The University of Chicago Press Chicago, IL}
}

@article{romer1989does,
  title={Does monetary policy matter? A new test in the spirit of Friedman and Schwartz},
  author={Romer, Christina D and Romer, David H},
  journal={NBER macroeconomics annual},
  volume={4},
  pages={121--170},
  year={1989},
  publisher={MIT Press}
}

@article{Fan-Tang2013,
 ISSN = {13697412, 14679868},
 URL = {http://www.jstor.org/stable/24772736},
 author = {Yingying Fan and Cheng Yong Tang},
 journal = {Journal of the Royal Statistical Society. Series B (Statistical Methodology)},
 number = {3},
 pages = {531--552},
 publisher = {[Royal Statistical Society, Wiley]},
 title = {Tuning parameter selection in high dimensional penalized likelihood},
 urldate = {2026-10-09},
 volume = {75},
 year = {2013}
}

\newpage
\appendix
\unhidefromtoc 
\tableofcontents
\newpage
\section{Properties of Sub-Weibull Random Variables}
\label{sec:appendix_sub_weibull}

\cite{kuchibhotla2022moving} introduce the general framework for analyzing sub-Weibull random variables that allows for more general exponential type tail behavior than the traditional sub-Gaussian or sub-exponential random variable. In this section, we provide the definition and various properties that we explore in the paper.

\begin{defn}[Sub-Weibull random variables and vectors] \label{def:sub_weibull}
Let $\nu>0$. Define $\psi_\nu(t):=\exp(t^\nu)-1$ for $t\ge 0$.  
A real-valued random variable $Z$ is \emph{sub-Weibull of order $\nu$ with parameter $K$}, written as
$Z \sim \text{sub-Weibull}(\nu,K)$, if its $\psi_\nu$-(quasi)norm
\[
\|Z\|_{\psi_\nu}
:=\inf\Bigl\{c>0:\ \mathbb{E}\,\psi_\nu\bigl(|Z|/c\bigr)\le 1\Bigr\}
\]
satisfies $\|Z\|_{\psi_\nu}\le K<\infty$. For a random vector $Z=(Z_1,\ldots,Z_p)^\top\in\mathbb{R}^p$, define the \emph{joint} and \emph{marginal} $\psi_\nu$ norms by
\[
\|Z\|_{J,\psi_\nu}
:=\sup_{\|v\|_2=1}\ \|\langle v,Z\rangle\|_{\psi_\nu},
\qquad
\|Z\|_{M,\psi_\nu}
:=\max_{j\in[p]}\ \|Z_j\|_{\psi_\nu}.
\]
We say $Z$ is \emph{joint sub-Weibull}$(\nu,K)$ if $\|Z\|_{J,\psi_\nu}\le K$, and
\emph{marginal sub-Weibull}$(\nu,K)$ if $\|Z\|_{M,\psi_\nu}\le K$. Note that $\|Z\|_{M,\psi_\nu}\le \|Z\|_{J,\psi_\nu}$ (by taking $v=e_j$), so that the joint sub-Weibull property implies the marginal property, making the latter a weaker condition. 
\end{defn}

Throughout the paper, we focus on the case when $0 < \nu \leq 2$. Special cases include $\nu=1$ (sub-exponential) and $\nu=2$ (sub-Gaussian). For instance, \citet[Theorem~2.4]{van-de-geer(2014)} assume independent and identically distributed rows, where each row is jointly sub-Gaussian. We focus instead on the marginal sub-Weibull case. Assuming joint sub-Weibull could improve our rate conditions, but it appears to be much more stringent than marginal sub-Weibull, as discussed in \cite{kuchibhotla2022moving}.

\begin{lem} [Maximum inequality for $\bf X$]
\label{lem:bound_of_X}
Suppose the random matrix ${\bf X}=(X_{ij})\in\mathbb{R}^{n\times p}$ satisfies the entrywise (row-wise marginal) sub-Weibull condition
\[
\max_{i\in[n],\,j\in[p]}\|X_{ij}\|_{\psi_\nu}\le K_{n,p}.
\]
Then for any $s\ge 0$,
\[
\mathbb{P}\!\left(
\max_{i\in[n],\,j\in[p]} |X_{ij}|
> K_{n,p}\bigl(\log(1+np)+s\bigr)^{1/\nu}
\right)\le e^{-s}.
\]
\end{lem}

\noindent\textbf{Proof of Lemma \ref{lem:bound_of_X}.}
Fix $t \geq 1$ and set $u:=K_{n,p}\psi_\nu^{-1}(tnp)$, where
$\psi_\nu^{-1}(y)=(\log(1+y))^{1/\nu}$.
If $\max_{i,j}|X_{ij}|>u$, then for some $(i,j)$ we have
$|X_{ij}|/K_{n,p}>\psi_\nu^{-1}(tnp)$, whence
$\psi_\nu(|X_{ij}|/K_{n,p})>tnp$. Therefore
\[
\Bigl\{\max_{i,j}|X_{ij}|>u\Bigr\}
\subseteq
\left\{\sum_{i=1}^n\sum_{j=1}^p \psi_\nu(|X_{ij}|/K_{n,p})>tnp\right\}.
\]
By Markov inequality and the assumption $\|X_{ij}\|_{\psi_\nu}\le K_{n,p}$ (so that
$\mathbb{E}\,\psi_\nu(|X_{ij}|/K_{n,p})\le 1$ for all $i,j$),
\[
\mathbb{P}\!\left(\sum_{i=1}^n\sum_{j=1}^p \psi_\nu(|X_{ij}|/K_{n,p})>tnp\right)
\le \frac{1}{tnp}\sum_{i=1}^n\sum_{j=1}^p \mathbb{E}\,\psi_\nu(|X_{ij}|/K_{n,p})
\le \frac{1}{t}.
\]
Now set $t=e^{s}$ to obtain
\[
\mathbb{P}\!\left(\max_{i,j}|X_{ij}|>K_{n,p}\psi_\nu^{-1}(e^{s}np)\right)\le e^{-s}.
\]
Finally, since $1+e^{s}np\le e^{s}(1+np)$, we have
$\log(1+e^{s}np)\le s+\log(1+np)$, i.e.
$\psi_\nu^{-1}(e^{s}np)\le (\log(1+np)+s)^{1/\nu}$.
This yields the stated bound. \qed

\begin{rem} \label{eq:rem:zeta_n}
Lemma \ref{lem:bound_of_X} implies that $\max_{i,j}|X_{ij}|$ admits a high-probability envelope. 
For instance, take $s=\log(1+np)$ and define
\begin{align*}
\xi_n
:= K_{n,p}\bigl(2\log(1+np)\bigr)^{1/\nu}
\asymp K_{n,p}(\log (np))^{1/\nu}.
\end{align*}
Then
\[
\mathbb{P}\!\left(\max_{i\in[n],\,j\in[p]} |X_{ij}|\le \xi_n\right)
\ge 1-(1+np)^{-1}\ \rightarrow\ 1.
\]
\end{rem}

\begin{lem}[Maximum inequality for $\bf X^\top \eps$] \label{lem:bound_of_Xeps}
    Suppose for some $h \geq 4$,
    \begin{enumerate}
        \item $\max_{i \in [n],j\in[p]} \lVert X_{ij} \rVert_{\psi_\nu} \leq K_{n,p}$,  
     \item  $\max_{i \in [n]} \mathbb{E}|\epsilon_i|^{h} \leq \overline{c} < \infty$,
     \item $(\tilde X_i, \eps_i)^\top \indep (\tilde X_j, \eps_j)^\top$ for $i \neq j$, and $\mathbb{E}\left[{\bf X}^\top  \epsilon / n\right] = 0$,
     \item $(\log(np))^{2/\nu - 1/2} = o(n^{1/2-1/h})$.
    \end{enumerate}
     Then there exists a strictly positive integer $N(\nu,h)$ such that, for all $n \geq N(\nu,h)$, we have
    \begin{align*}
        \mathbb P \left(  \left\Vert \frac{1}{n} \bf X^\top \eps \right\Vert_{\infty} > C_{\nu,h} K_{n,p} \sqrt{\frac{\log  (np)}{n}} \right) \leq \frac{3}{np} +\frac{1}{(\log(2n))^{1/\nu}},
    \end{align*}
    for a constant $C_{\nu,h}$ that depends only on $\nu$ and $h$.
\end{lem}

\noindent \textbf{Proof of Lemma \ref{lem:bound_of_Xeps}:} The proof of \cite[Theorem 4.6]{kuchibhotla2022moving} shows that with probability at least $1 - 3/(np) - 1/L$, for any $L \geq 1$,
\begin{align}
    \left\Vert \frac{1}{n} \bf X^\top \eps \right\Vert_{\infty} \leq 7 \sqrt{2} \sigma_{n,p} \sqrt{\frac{\log (np)}{n}} + \frac{C_{\nu} K_{n,p}K_{\epsilon,h} (\log (np))^{1/\nu} \left[(\log(2n))^{1/\nu} + L \right] }{n^{1-1/h}}, \label{KC22_them4.6}
\end{align}
where $C_{\nu}$ is some constant that only depends on $\nu$, the constant $K_{\epsilon,h}:= \max_{i \in [n]} \lVert \epsilon_i \rVert_h  \leq \overline{c}^{1/h}$, and 
\begin{align*}
    \sigma^2_{n,p} := \max_{j \in [p]} \frac{1}{n} \sum_{i \in [n]} \mathbb{E} \left[ \epsilon_i^2 X_{ij}^2 \right].
\end{align*}
We can obtain, by H\"{o}lder inequality, 
\begin{align} \label{sigma_np_unif_bd}
    \mathbb E\left[ X_{ij}^2 \eps_i^2\right] \leq \left( \mathbb E|X_{ij}|^{2h/(h-2)} \right)^{(h-2)/h} \left( \mathbb E|\eps_i|^{h} \right)^{2/h} \lesssim \lVert X_{ij} \rVert_{2h/(h-2)}^2 \lesssim K_{n,p}^2.
\end{align}
For the last inequality, note for any $i\in[n],j\in[p]$, 
\begin{align*}
    \mathbb{E}\exp\left(\frac{|X_{ij}|^\nu}{K_{n,p}^\nu}\right) \leq 2 \quad \Rightarrow \quad \mathbb{P}(|X_{ij}|>t) \leq 2 e^{-t^\nu/K_{n,p}^\nu},
\end{align*}
so that by integral identity,
\begin{align*}
    \mathbb{E}|X_{ij}|^k &= \int_0^\infty \mathbb{P}(|X_{ij}|^k > s) ds \\
    &= \int_0^\infty k t^{k-1} \mathbb{P}(|X_{ij}| >t) dt \\
    &\leq 2k \int_0^\infty t^{k-1} e^{-t^\nu/K_{n,p}^\nu} dt \\
    & = \frac{ 2 k K_{n,p}^k}{\nu} \int_0^\infty u^{k/\nu-1} e^{-u} du \\
    &= \frac{2 k K_{n,p}^k}{\nu} \Gamma(k/\nu),
\end{align*}
where the Gamma function $\Gamma (x) = \int_0^\infty t^{x-1} e^{-t} d t$ and $k = 2h / (h-2)$. Using the inequality \eqref{sigma_np_unif_bd} for $\sigma_{n,p}$ in \eqref{KC22_them4.6}, therefore for any $L \geq 1$, with probability at least $1 - 3/(np) - 1/L$, we have
\begin{align*}
    \left\Vert \frac{1}{n} \bf X^\top \eps \right\Vert_{\infty} &\lesssim \sqrt{\frac{K_{n,p}^2\log (np)}{n}} + \frac{K_{n,p} (\log (np))^{1/\nu} \left[(\log(2n))^{1/\nu} + L \right] }{n^{1-1/h}}.
\end{align*}
Furthermore, with large enough $n$,  by choosing $L = (\log(2n))^{1/\nu}$, the second term has order 
\begin{align*}
    \frac{K_{n,p} (\log (np))^{2/\nu}}{n^{1-1/h}},
\end{align*}
and note that 
\begin{align*}
    \frac{K_{n,p}(\log(np))^{2/\nu}/n^{1-1/h}}{K_{n,p}(\log(np))^{1/2}/n^{1/2}} = \frac{(\log(np))^{2/\nu - 1/2}}{n^{1/2-1/h}} = o(1)
\end{align*}
by assumption. Putting everything together, for large enough $n$ (such that $L \geq 1$), we have with probability at least $1-3/(np)-(\log(2n))^{-1/\nu}$,
\begin{align*}
    \left\Vert \frac{1}{n} \bf X^\top \eps \right\Vert_{\infty} \leq C_{\nu,h} K_{n,p}\sqrt{\frac{\log (np)}{n}}.
\end{align*}  
\qed

\begin{rem}
    Note that we only require $\eps$ to have polynomial moments; in contrast, \cite{van-de-geer(2014)} requires $\eps$ to be sub-exponential when $\bf X$ is sub-Gaussian. The condition
    \begin{align*}
    \frac{(\log(np))^{2/\nu - 1/2}}{n^{1/2-1/h}} = o(1)
    \end{align*}
    in Lemma \ref{lem:bound_of_Xeps} may not seem intuitive at first, but if we assume $X_{ij}$ is {\it sub-Gaussian} and finite fourth moment of error (so $\nu=2$ and $h=4$), then we obtain
    \begin{align*}
        \frac{(\log(np))^{2/\nu - 1/2}}{n^{1/2-1/h}}  = \sqrt{\frac{\log(np)}{\sqrt{n}}} =  \sqrt{
        \frac{\log p + \log n}{\sqrt{n}}}.
    \end{align*}
    Under the standard assumption that $ \log p/\sqrt{n} = o(1)$, the above term converges to zero. In the case of {\it sub-exponential} $X_{ij}$, we have $\nu=1$ and
    \begin{align*}
        \frac{(\log(np))^{2/\nu - 1/2}}{n^{1/2-1/h}}  = \frac{(\log n + \log p)^{3/2}}{n^{1/2-1/h}} = \frac{(\log p)^{3/2}}{n^{1/2-1/h}} + o(1).
    \end{align*}
    If we continue to only require $h=4$, then in order to satisfy the condition of Lemma \ref{lem:bound_of_Xeps}, we need ${(\log p)^6}/{n} = o(1)$. If instead we can allow for higher finite moments of error, say $h = 8$, then we only need ${(\log p)^4}/{n} = o(1)$. 
\end{rem}

\begin{lem}[Sum of sub-Weibull is sub-Weibull] \label{lem:sum_of_sub-Weibull_is_sub-Weibull}
    Consider random variables $\{ Z_i\}_{i=1}^n$ satisfying
    \begin{align*}
        \lVert Z_i \rVert_{\psi_\nu} \leq K_{i,n}
    \end{align*}
    for each $i \in [n]$. Then
    \begin{align*}
        \left\lVert \sum_{i \in [n]} Z_i \right\rVert_{\psi_\nu} \leq 
        \begin{cases}
            \sum_{i \in [n]} K_{i,n}, & \text{for }\; \nu \geq 1; \\
            C(\nu) \sum_{i \in [n]} K_{i,n}, & \text{for } \; 0<\nu<1,
        \end{cases}
    \end{align*}
    where $C(\nu) = e (4+4/\nu)^{1/\nu}$.
\end{lem}

\noindent {\bf Proof of Lemma \ref{lem:sum_of_sub-Weibull_is_sub-Weibull}:} When $\nu \geq 1$, $\psi_\nu$ is an orlicz norm, c.f. \cite[Chapter 2.2]{VW96}. Therefore the triangle inequality applies, implying
\begin{align*}
     \left\lVert \sum_{i \in [n]} Z_i \right\rVert_{\psi_\nu} \leq \sum_{i \in [n]} \lVert Z_i \rVert_{\psi_\nu} \leq \sum_{i \in [n]} K_{i,n}.
\end{align*}

\vspace{5mm}
For $0<\nu<1$, the proof is similar to \cite[Proposition A.5]{kuchibhotla2022moving}. We first prove an intermediate result:
\begin{align}
    C_1(\nu) \sup_{k \geq 1} \frac{\lVert Z \rVert_k}{k^{1/\nu}} \leq \lVert Z \rVert_{\psi_\nu} \leq C_2(\nu) \sup_{k \geq 1} \frac{\lVert Z \rVert_k}{k^{1/\nu}}, 
    \label{eq:sum_of_sub-Weibull_is_sub-Weibull}
\end{align}
for some $C_1(\nu)$ and $C_2(\nu)$ depending only on $\nu$. For the first inequality, assume that $\lVert Z \rVert_{\psi_\nu} < \infty$. Fix a $\eta > \lVert Z \rVert_{\psi_\nu}$, by definition
\begin{align*}
    \mathbb E \left[\exp \left( (|Z|/\eta)^\nu \right) - 1 \right] \leq 1 \quad \Longrightarrow \quad \mathbb E \left| (|Z|/\eta)^\nu \right|^m / (m!) \leq 1, 
\end{align*}
for all $m \in \mathbb N^+$. Then for any $k \geq 1$, we have $k/\nu \geq 1$, so that letting $k' = \lceil k/\nu \rceil$ yields
\begin{align*}
    \lVert Z \rVert^\nu_k / \eta^\nu = \lVert(|Z|/\eta)^\nu \rVert_{k/\nu}  \leq \lVert(|Z|/\eta)^\nu \rVert_{k'} \leq (k'!)^{1/k'} \leq k' \leq 1 + k/\nu \leq k (1+1/\nu ),
\end{align*}
implying
\begin{align*}
    \lVert Z \rVert_k \leq \eta (1+ 1/\nu)^{1/\nu} k^{1/\nu}.
 \end{align*}
Letting $C_1(\nu) = (\nu/(\nu+1))^{1/\nu}$ and $\eta \downarrow \lVert Z \rVert_{\psi_\nu}$ gives the first inequality of \eqref{eq:sum_of_sub-Weibull_is_sub-Weibull}. For the second inequality, assume that $K :=  \sup_{k \geq 1} \lVert Z \rVert_k / k^{1/\nu} < \infty$. By Markov inequality, we have for $t \geq 1$,
\begin{align*}
    \mathbb P \left( \left|Z \right| \geq eK t^{1/\nu} \right) \leq \left(\frac{K t^{1/\nu}}{eK t^{1/\nu}} \right)^t \leq \exp(-t).
\end{align*}
For $0 \leq t < 1$, we have
\begin{align*}
    \mathbb P \left( \left|Z \right| \geq eK t^{1/\nu} \right) \leq 1 \leq e \exp(-t),
\end{align*}
and thus
\begin{align*}
    \mathbb P \left( \left|Z \right| \geq eK t^{1/\nu} \right) \leq e \exp(-t)
\end{align*}
for all $t \geq 0$. Let $C_2(\nu) = e 4^{1/\nu}$, then
\begin{align*}
    \mathbb E \left[ \psi_\nu \left( |Z|/ C_2(\nu) K \right) \right] &= \int_0^\infty \mathbb P \left( |Z| \geq C_2(\nu) K  \psi_\nu^{-1}(t) \right) dt \\
    &= \int_0^\infty \mathbb P \left( |Z| \geq C_2(\nu) K  \left(\log (1+t) \right)^{1/\nu} \right) dt \\
    &= \int_0^\infty \mathbb P \left( |Z| \geq e K  \left(\log (1+t)^4 \right)^{1/\nu} \right) dt \\
    &\leq e \int_0^\infty \frac{1}{(1+t)^4} dt < 1,
\end{align*}
and the RHS of \eqref{eq:sum_of_sub-Weibull_is_sub-Weibull} follows. Finally, by $\lVert Z_i \rVert_{\psi_\nu} \leq K_{i,n}$ and applying \eqref{eq:sum_of_sub-Weibull_is_sub-Weibull}, we have
\begin{align*}
    \left\lVert \sum_{i \in [n]} Z_i \right\rVert_{\psi_\nu} &\leq C_2(\nu) \sup_{k \geq 1} \frac{\lVert \sum_{i \in [n]} Z_i \rVert_k}{k^{1/\nu}} \\
    &\leq C_2(\nu) \sup_{k \geq 1} \frac{\sum_{i \in [n]} \lVert  Z_i \rVert_k}{k^{1/\nu}} \\
    &\leq C_2(\nu) \sum_{i \in [n]} \sup_{k \geq 1} \frac{\lVert  Z_i \rVert_k}{k^{1/\nu}} \\
    &\leq \frac{C_2(\nu)}{C_1(\nu)} \sum_{i \in [n]} \lVert Z_i \rVert_{\psi_\nu} \\
     &\leq C(\nu) \sum_{i \in [n]} K_{i,n},
\end{align*}
where $C(\nu) = e (4+4/\nu)^{1/\nu}$.
\qed

\newpage
\section{Auxiliary Lemmas for LASSO Estimation}


In this section, we consider a more general LASSO estimator $\widehat \beta$ defined as
\begin{align}
    \widehat{\beta} \in \arg \min_{\beta \in \mathbb{R}^p} \frac{1}{n} \lVert Y - {\bf X}\beta \rVert_2^2 + 2\lambda_0 \sum_{j \in [p]}\hat w_j |\beta_j|,
    \label{eq:beta-hat-general}
\end{align}
with $0 \leq \hat w_j \leq 1, j \in [p]$. The ordinary LASSO estimator, which we now denote as $\widehat{\beta}_{\rm L}$, corresponds to $\hat w_j = 1, j \in [p]$. The conservative LASSO estimator in \cite{caner2018asymptotically}, denoted as $\widehat{\beta}_{\rm C}$, corresponds to $\hat w_j = \lambda_{{\rm prec},0}/(|\widehat{\beta}_{L,j}| \vee \lambda_{{\rm prec},0})$, where $\lambda_{{\rm prec},0} \geq \left\Vert \widehat{\beta}_{\rm L}-\beta^0 \right\Vert_{\infty}$ with high probability. Note that although the adaptive LASSO in \cite{zou2006adaptive} has the same weighted-$\ell_1$ structure in Equation \eqref{eq:beta-hat-general}, its weights are not uniformly bounded and therefore do not belong to the class considered here. In particular, the cone argument below might fail when a relevant first-stage coefficient is close to zero. See \cite{caner2018asymptotically} for further discussion.

\subsection{Fixed Design}

We first provide the result when $\bf X$ is treated as fixed, which can then be applied to quite general random designs. We require the following restricted eigenvalue (RE) condition.

\begin{defn}[Restricted eigenvalue condition]
    For a matrix ${\bf X} \in \mathbb{R}^{n \times p}$, we say that it satisfies the {\it restricted eigenvalue condition} with parameter $(s,\kappa)$ whenever there exists a constant $\phi_0>0$ with
        \begin{align*}
            \inf_{\underset{|S| \leq s}{S \subseteq \{1,...,p\}}} \inf_{\underset{u \in \mathcal{C}(S,\kappa)}{u \neq 0}} \frac{u^\top \left({\bf X}^\top {\bf X}/n \right) u}{ \lVert u \rVert_2^2} \geq \phi_0,
        \end{align*}
        where $\mathcal{C}(S,\kappa):=\{ u \in \mathbb{R}^p: \lVert u_{S^c} \rVert_1 \leq \kappa \lVert u_{S} \rVert_1 \}$.
    \label{defn:sparse-RE}
\end{defn}

\begin{ass} \label{ass:2}
    Assume 
    \begin{enumerate}
        \item $\max_{i \in [n],j\in[p]} |X_{ij}| \leq \xi_n$ for some $\xi_n>0$.
        \item ${\bf X}$ satisfies the RE condition with parameter $(s_0, 3)$ where $s_0 = \lVert \beta^0 \rVert_0$. 
    \end{enumerate}
\end{ass}

\begin{lem}[Fixed design bound] \label{lem:main_LASSO}
    Suppose Assumption \ref{ass:2} holds and $\lambda_0 \geq 2\lVert {\bf X}^\top \epsilon/n \rVert_\infty $ where $\lambda_0 = C_n \sqrt{\log(np)/n}$ for some $C_n>0$. Assume also that $\hat w_j = 1, j \in S_0^c$ where $S_0 = \{ s \in[p] : \beta^0_s \neq 0\}$. Then we have
    \begin{enumerate}
        \item $\left\lVert {\bf X} (\widehat\beta - \beta^0) /\sqrt{n} \right\rVert_{2} = O\left( C_n\sqrt{\frac{s_0 \log (np)}{n}}\right)$,
        \item $\lVert \widehat\beta-\beta^0 \rVert_2 = O\left(C_n \sqrt{\frac{s_0 \log (np)}{n}}\right)$,
        \item $\lVert \widehat\beta-\beta^0 \rVert_1 = O\left(C_n \sqrt{\frac{s_0^2 \log (np)}{n}}\right)$,
        \item $\lVert {\bf X} (\widehat\beta - \beta^0) \rVert_{\infty} = O\left( C_n\sqrt{\frac{\xi_n^2 s_0^2 \log (np)}{n}}\right)$.
    \end{enumerate}
\end{lem}

\begin{rem}
    The assumption that $\hat w_j = 1, j \in S_0^c$ holds trivially for the original LASSO. It also holds for the conservative LASSO with probability approaching one as we are going to show in Theorem \ref{thm:size-control_random_design_extension}. We also use a slightly stronger condition on the restricted eigenvalue than the compatibility condition in \cite{van-de-geer(2014)}, which allows us to derive the $\ell_2$-norm bound for the node-wise LASSO estimator directly. Both this condition and the compatibility condition hold with probability approaching one under our assumption that $0 < \underline{c} \leq\lambda_{\min}(\Sigma)$.
\end{rem}

\noindent \textbf{Proof of Lemma \ref{lem:main_LASSO}:} By the definition of $\widehat\beta$, we have
\begin{align*}
    \frac{1}{n} \lVert Y - {\bf X} \widehat\beta \rVert_2^2 + 2\lambda_0 \sum_{j \in [p]}\hat w_j |\widehat{\beta}_j| \leq \frac{1}{n} \lVert Y - {\bf X} \beta^0 \rVert_2^2 + 2\lambda_0 \sum_{j \in [p]}\hat w_j |\beta^0_j|.
\end{align*}
By the preceding equation, re-arranging and substituting $Y = {\bf X} \beta^0 +\epsilon$ yields the {\it basic inequality}
\begin{align*}
    \frac{1}{n} \sum_{i \in [n]} \left(\widetilde{X}_i^\top(\widehat\beta-\beta^0)\right)^2 + 2\lambda_0 \sum_{j \in [p]}\hat w_j |\widehat{\beta}_j| \leq \frac{2}{n} \sum_{i \in [n]} \widetilde{X}_i^\top(\widehat\beta - \beta^0)\epsilon_i  + 2\lambda_0  \sum_{j \in [p]}\hat w_j |\beta^0_j|.
\end{align*}
Note that by H\"older inequality,
\begin{align*}
    \left| \left(\frac{2}{n} \sum_{i \in [n]} \widetilde{X}_i^\top\epsilon_i \right) (\widehat\beta - \beta^0) \right| \leq 2 \left\lVert\frac{1}{n} \sum_{i \in [n]} \widetilde{X}_i^\top\epsilon_i \right\rVert_\infty \left\lVert \widehat\beta - \beta^0 \right\rVert_1 \leq \lambda_0 \left\lVert \widehat\beta - \beta^0 \right\rVert_1
\end{align*}
by assumption, so that 
\begin{align*}
    \lVert {\bf X} (\widehat\beta - \beta^0)/\sqrt{n}\rVert_{2}^2 \leq \lambda_0 \left\lVert \widehat\beta - \beta^0 \right\rVert_1 + 2\lambda_0 \sum_{j \in [p]}\hat w_j |\beta^0_j| - 2\lambda_0 \sum_{j \in [p]}\hat w_j |\widehat{\beta}_j|.
\end{align*}
Define $\widehat{\Delta} := \widehat\beta - \beta^0$ and note that $\beta^0_j = 0$ and $\hat w_j = 1$ on $S_0^c$, we have
\begin{align*}
    \sum_{j \in S_0^c}\hat w_j |\widehat{\beta}_j| - \sum_{j \in S_0^c}\hat w_j |\beta^0_j| = \lVert \widehat{\Delta}_{S_0^c} \rVert_1
\end{align*}
and 
\begin{align*}
    \sum_{j \in S_0}\hat w_j |\widehat{\beta}_j| - \sum_{j \in S_0}\hat w_j |\beta^0_j|\geq  - \max_{j \in S_0} \hat w_j \lVert \widehat{\Delta}_{S_0} \rVert_1, 
\end{align*}
and thus
\begin{align*}
    \lVert {\bf X} \widehat{\Delta}/\sqrt{n}\rVert_{2}^2 &\leq \lambda_0 \lVert \widehat{\Delta} \rVert_1 + 2\lambda_0 \max_{j \in S_0} \hat w_j \lVert \widehat{\Delta}_{S_0} \rVert_1 - 2\lambda_0 \lVert \widehat{\Delta}_{S_0^c} \rVert_1  \\
    &= \lambda_0 (2\max_{j \in S_0} \hat w_j +1) \lVert \widehat{\Delta}_{S_0} \rVert_1 - \lambda_0 \lVert \widehat{\Delta}_{S_0^c} \rVert_1,
\end{align*}
which further implies that
\begin{align*}
    \lVert \widehat{\Delta}_{S_0^c} \rVert_1 \leq (2\max_{j \in S_0} \hat w_j +1) \lVert \widehat{\Delta}_{S_0} \rVert_1 \leq 3 \lVert \widehat{\Delta}_{S_0} \rVert_1.
\end{align*}
By the preceding equation, $\widehat{\Delta} \in \mathcal{C}(S,3)$, so that by Assumption \ref{ass:2},
\begin{align*}
    \lVert {\bf X}\widehat{\Delta}/\sqrt{n}\rVert_{2}^2 = n^{-1}\widehat{\Delta}^\top {\bf X^\top X} \widehat{\Delta} \geq \phi_0 \lVert \widehat{\Delta} \rVert_2^2 \geq \phi_0 \lVert \widehat{\Delta}_{S_0} \rVert_2^2 \geq \phi_0 \lVert \widehat{\Delta}_{S_0} \rVert_1^2/s_0,
\end{align*}
whence
\begin{align*}
    &\lVert {\bf X} \widehat{\Delta}/\sqrt{n}\rVert_{2}^2 \leq 3\lambda_0 \lVert \widehat{\Delta}_{S_0} \rVert_1 \leq \frac{3\lambda_0\sqrt{s_0}}{ \phi_0^{1/2}} \lVert {\bf X} \widehat{\Delta}/\sqrt{n}\rVert_{2} \\
    &\quad \Rightarrow \quad  \lVert {\bf X} \widehat{\Delta}/\sqrt{n}\rVert_{2} = O(\lambda_0 \sqrt{s_0}) = O\left(C_n\sqrt{\frac{s_0 \log (np)}{n}} \right).
\end{align*}
We also have
\begin{align*}
    \phi_0 \lVert \widehat{\Delta} \rVert_2^2 \leq \lVert {\bf X} \widehat{\Delta}/\sqrt{n}\rVert_{2}^2 = O(\lambda_0^2 s_0) \quad \Rightarrow \quad \lVert \widehat{\Delta} \rVert_2 = O(\lambda_0 \sqrt{s_0}) = O\left(C_n\sqrt{\frac{s_0 \log (np)}{n}} \right).
\end{align*}
Furthermore, we have
\begin{align*}
    \lVert \widehat{\Delta} \rVert_1 = \lVert \widehat{\Delta}_{S_0} \rVert_1 + \lVert \widehat{\Delta}_{S_0^c} \rVert_1 \leq 4 \lVert \widehat{\Delta}_{S_0} \rVert_1 \leq 4 \sqrt{s_0} \lVert \widehat{\Delta}_{S_0} \rVert_2 \leq 4 \sqrt{s_0} \lVert \widehat{\Delta} \rVert_2 = O\left(C_n\sqrt{\frac{s_0^2 \log (np)}{n}} \right).
\end{align*}
Finally, the last claim of the lemma follows from observing
\begin{align*}
    \lVert {\bf X} \widehat{\Delta} \rVert_{\infty} \leq \max_{i \in [n],j \in [p]} |X_{ij}| \times \lVert\widehat{\Delta} \rVert_1 \leq \xi_n  \lVert \widehat{\Delta} \rVert_1 = O\left(C_n \sqrt{\frac{\xi_n^2 s_0^2\log (np)}{n}} \right).
\end{align*}
\qed

\begin{rem} \label{rem:LASSO_tuning_rate}
     
Lemma \ref{lem:bound_of_Xeps} suggests that we can set $C_n \equiv 2 C_{\nu,h} K_{n,p}$ in Lemma \ref{lem:main_LASSO}, so that $\lambda_0 \asymp C_{\nu,h} K_{n,p} \sqrt{\log (np)/n}$ and $\lambda_0 \geq 2\lVert {\bf X}^\top \epsilon/n \rVert_\infty $ with probability approaching one.
\end{rem}

\subsection{Random Design}
We show here that under very general conditions, when we instead treat ${\bf X}$ as a random matrix, then it satisfies Assumption \ref{ass:2} with probability approaching one, which implies the same rate. 

\begin{lem}[Sufficient conditions for Assumption \ref{ass:2}] \label{lem:random_design_X}
    Suppose 
    \begin{enumerate}
        \item $ \max_{i \in [n], j \in [p]} \lVert X_{ij} \rVert_{\psi_\nu} \leq K_{n,p}$,
        \item $K_{n,p}^2 \max\left\{ \sqrt{\frac{s_0^2 \log (np)}{n}}, \frac{s_0 (\log(np))^{4/\nu}}{n} \right\} = o(1)$, 
        \item $\widetilde{X}_i^\top \indep \widetilde{X}_j^\top$ for $i \neq j$, 
        \item $0 < \underline{c} \leq\lambda_{\min}(\Sigma)$.
    \end{enumerate}
    Then Assumption \ref{ass:2} is satisfied for ${\bf X}$ with probability approaching one with $\xi_n \asymp K_{n,p} (\log (np))^{1/\nu}$ and $\phi_0 = \underline{c}/2$. 
\end{lem}

\noindent \textbf{Proof of Lemma \ref{lem:random_design_X}:}
Define
\begin{align*}
\xi_n
:= K_{n,p}\bigl(2\log(1+np)\bigr)^{1/\nu}.
\end{align*}
Note that by Remark \ref{eq:rem:zeta_n}, we have with probability approaching one
\begin{align*}
    \max_{i \in [n], j \in [p]}| X_{ij} | \leq \xi_n \asymp K_{n,p} (\log (np))^{1/\nu},
\end{align*}
so that Assumption \ref{ass:2}.1 is satisfied. 

\vspace{5mm}
For Assumption \ref{ass:2}.2, let $\delta_n := \lVert \widehat{\Sigma} - \Sigma \rVert_{\max}$. For any $u \in \mathcal{C}(S,3)$ with $S \subseteq \{1,...,p\}$ and 
$|S| \leq s_0$, we have
\begin{align*}
    u^\top \widehat{\Sigma} u = u^\top \Sigma u + u^\top (\widehat{\Sigma}-\Sigma) u \geq \underline{c} \lVert u \rVert_2^2 - \delta_n  \lVert u \rVert_1^2 \geq \underline{c} \lVert u \rVert_2^2 - 16 \delta_n s_0  \lVert u \rVert_2^2,
\end{align*}
where the last inequality is by
\begin{align*}
    \lVert u \rVert_1 \leq 4 \lVert u_S \rVert_1 \leq 4 \sqrt{s_0} \lVert u_S \rVert_2 \leq 4 \sqrt{s_0} \lVert u \rVert_2.
\end{align*}
By \cite[Theorem 4.1]{kuchibhotla2022moving}, we have
\begin{align*}
    \delta_n = O_\Prob\left( \max\left\{A_{n,p} \sqrt{\frac{\log (np)}{n}}, K_{n,p}^2 \frac{(\log (np))^{4/\nu}}{n} \right\} \right),
\end{align*}
where 
\begin{align*}
    A_{n,p}^2 := \max_{1 \leq j\leq k \leq p} \frac{1}{n} \sum_{i \in [n]} \text{Var}(X_{ij} X_{ik}).
\end{align*}
Note that by the proof of Lemma \ref{lem:bound_of_Xeps}, we show that 
\begin{align*}
    \text{Var}(X_{ij} X_{ik}) \leq \mathbb{E}(X_{ij}^2 X_{ik}^2) \leq \frac{\mathbb{E}X_{ij}^4}{2} + \frac{\mathbb{E}X_{ik}^4}{2} \lesssim K_{n,p}^4,
\end{align*}
uniformly over ${i \in [n], j,k \in [p]}$, implying that $A_{n,p}^2 \lesssim K_{n,p}^4$. It follows that $\delta_n s_0 = o_\Prob(1)$ by assumption, so that $16 \delta_n s_0 \leq \underline{c}/2$ with probability approaching one. The proof is complete.

\qed

\newpage
\section{Auxiliary Results for Nodewise LASSO Estimation}
\label{sec:node-wise}

As in the previous section, we consider here a more general LASSO estimator $\widehat{\gamma}_j $ defined as 
\begin{align}
    \widehat{\gamma}_j \in \argmin_{\gamma \in \mathbb{R}^{p-1}} \frac{1}{n} \lVert X_j - {\bf X}_{-j} \gamma  \rVert_2^2 + 2\lambda_j \sum_{k\neq j} \hat w_{j,k} |\gamma_k|,
    \label{eq:approx_inver_LASSO}
\end{align}
with $0 \leq \hat w_{j,k} \leq 1, k \neq j$. Again, this covers both the ordinary LASSO estimator and the conservative LASSO estimator in \cite{caner2018asymptotically}. We also redefine $\widehat{\tau}_j^2 := \lVert X_j - {\bf X}_{-j} \widehat{\gamma}_j \rVert_2^2/n + \lambda_j \sum_{k\neq j} \hat w_{j,k} |\widehat{\gamma}_{j,k}|$ in this section.

\subsection{Properties of \texorpdfstring{\(\gamma_j\) and \(\tau_j^2\)}{Properties of gamma\_j and tau\_j squared}}\label{subsec:1}

Consider the minimization problem
\begin{align*}
    \gamma_j \in \arg \min_{\gamma \in \mathbb{R}^{p-1}} \mathbb{E} \lVert X_j - {\bf X}_{-j} \gamma \rVert_2^2/n,
\end{align*}
so that $\gamma_j$ is the population minimizer. We define the matrix 
\begin{align*}
    C := 
    \begin{bmatrix}
        1 & -\gamma_{1,2}& \cdots &-\gamma_{1,p} \\
        -\gamma_{2,1}& 1 & \cdots& -\gamma_{2,p} \\
        \vdots & \vdots & \ddots & \vdots \\
        -\gamma_{p,1} & -\gamma_{p,2} & \cdots & 1
    \end{bmatrix},
\end{align*}
and
\begin{align*}
    \tau_j^2 := \mathbb{E} \lVert X_j - {\bf X}_{-j} \gamma_j \rVert_2^2/n.
\end{align*}
Furthermore, note that for $j \in [p]$, there exists permutation matrix $\mathcal P_j$ such that
\begin{align*}
    \mathcal P_j \Sigma \mathcal P_j^\top = 
\begin{bmatrix}
\Sigma_{j,j} & \Sigma_{j,-j} \\
\Sigma_{-j,j} & \Sigma_{-j,-j}
\end{bmatrix},
\end{align*}
where $\Sigma_{j,j} \in \mathbb{R}$ is the average of variance for $X_j$, $\Sigma_{-j,-j} \in \mathbb{R}^{(p-1) \times (p-1)}$ is the average of covariance matrix for ${\bf X}_{-j}$, $\Sigma_{j,-j} = \Sigma_{-j,j}^\top \in \mathbb{R}^{1 \times (p-1)}$ are the average of cross-covariances between $X_j$ and ${\bf X}_{-j}$. Using the standard formula for the inverse of a block matrix, we obtain
\begin{align}
\Sigma^{-1} = \mathcal P_j^\top  
\begin{bmatrix}
(\Sigma_{j,j} - \Sigma_{j,-j} \Sigma_{-j,-j}^{-1} \Sigma_{-j,j})^{-1} & 
- (\Sigma_{j,j} - \Sigma_{j,-j} \Sigma_{-j,-j}^{-1} \Sigma_{-j,j})^{-1} \Sigma_{j,-j} \Sigma_{-j,-j}^{-1} \\
- \Sigma_{-j,-j}^{-1} \Sigma_{-j,j} (\Sigma_{j,j} - \Sigma_{j,-j} \Sigma_{-j,-j}^{-1} \Sigma_{-j,j})^{-1} &
\ast
\end{bmatrix} \mathcal P_j.
\label{eq:Sigma-inverse}
\end{align}
The following results are standard in the literature and we provide them here for completeness.

\begin{lem}[Relationship between \texorpdfstring{$\gamma_j$}{gamma\_j} and \texorpdfstring{$\Theta$}{Theta}]
Suppose $0<\lambda_{\min}(\Sigma)$. Then for every $j \in [p]$,
\begin{align*}
    (a)\;\lVert \gamma_j \rVert_0 = s_j, \quad and\quad (b)\; \gamma_j = -\frac{\Theta_{j,-j}}{{\Theta_{j,j}}}.
\end{align*}
    \label{lem:aux:1}
\end{lem}
{\noindent \textbf{Proof of Lemma \ref{lem:aux:1}:} 
By \eqref{eq:Sigma-inverse}, we have
\begin{align*}
\Theta_{j,j} = \left(\Sigma_{j,j} - \Sigma_{j,-j} \Sigma_{-j,-j}^{-1} \Sigma_{-j,j} \right)^{-1},
\end{align*}
and
\begin{align*}
\Theta_{j,-j} = -\Theta_{j,j} \Sigma_{j,-j} \Sigma_{-j,-j}^{-1}.
\end{align*}
Recall that the regression coefficients from regressing $X_j$ on ${\bf X}_{-j}$ are given by
\begin{align*}
    \gamma_j = \Sigma_{-j,-j}^{-1} \Sigma_{-j,j} \quad \Rightarrow \quad \gamma_j^\top = \Sigma_{j,-j} \Sigma_{-j,-j}^{-1}.
\end{align*}
Substituting into the formula for $\Theta_{j,-j}$ gives:
\begin{align*}
\Theta_{j,-j} = -\Theta_{j,j} \gamma_j^\top \quad \Rightarrow \quad \gamma_j = -\frac{1}{\Theta_{j,j}} \Theta_{-j,j}.
\end{align*}
Since the cardinality of $\Theta_{j,-j}$ is $s_j$, $\gamma_j$ will also have cardinality $s_j$.
\qed

\vspace{5mm}
\begin{lem}[Relationship between \texorpdfstring{$\tau_j^2$}{tau_j^2} and \texorpdfstring{$\Theta$}{Theta}]
Suppose $0<\lambda_{\min}(\Sigma)$. Then for every $j \in [p]$,
\begin{align*}
    (a) \;\tau_j^2 = 1/\Theta_{j,j}, \quad \text{and} \quad (b) \; \lambda_{\min}(\Sigma) \leq \tau_j^2 \leq \Sigma_{j,j}.  
\end{align*} \label{lem:aux:2}
\end{lem}
\noindent\textbf{Proof of Lemma \ref{lem:aux:2}:} We have
\begin{align*}
    \tau_j^2 &= \mathbb{E}[X_j^\top X_j]/n + \gamma_j^\top \mathbb{E}[{\bf X}_{-j}^\top{\bf X}_{-j}] \gamma_j/n - 2 \gamma_j^\top \mathbb{E}[{\bf X}_{-j}^\top X_j]/n \\
    & \overset{(i)}{=} \Sigma_{j,j} - \gamma_j^\top \mathbb{E}[{\bf X}_{-j}^\top X_j]/n \\ 
    & \overset{(ii)}{=} \Sigma_{j,j} - \Sigma_{-j,j}^\top \Sigma_{-j,-j}^{-1} \Sigma_{-j,j} \\
    &\leq \Sigma_{j,j},
\end{align*}
where $(i)$ and $(ii)$ follow from $\gamma_j = \mathbb{E}[{\bf X}_{-j}^\top {\bf X}_{-j}]^{-1} \mathbb{E}[{\bf X}_{-j}^\top X_j]$, which yields the second inequality in $(b)$; by recalling the top-left of \eqref{eq:Sigma-inverse}, we see that the second-last term of the preceding equation equals $\Theta_{j,j}^{-1}$, proving part (a). Finally, $\Theta_{j,j} \leq \lambda_{\max}(\Theta) = 1/\lambda_{\min}(\Sigma)$ so that the first inequality in $(b)$ follows by $(a)$.
\qed 

\vspace{5mm}
\begin{lem}[$\ell_1$ and $\ell_2$ bound for $\gamma_j$]
    Suppose $0<\lambda_{\min}(\Sigma)$. Then for every $j \in [p]$,
    \begin{align*}
        (a) \;\lVert \gamma_j \rVert_2 \leq \sqrt{\Sigma_{j,j}/ \lambda_{\min}(\Sigma)},  \quad and \quad (b) \;\lVert \gamma_j \rVert_1 \leq \sqrt{s_j \Sigma_{j,j} / \lambda_{\min}(\Sigma)}.
    \end{align*}
    \label{lem:aux:3}
\end{lem}
\noindent \textbf{Proof of Lemma \ref{lem:aux:3}:} Note that from the proof of Lemma \ref{lem:aux:2}, we see that
\begin{align*}
    0 \leq \tau^2_j = \Sigma_{j,j} - \Sigma_{j,-j} \Sigma_{-j,-j}^{-1} \Sigma_{-j,j},
\end{align*}
which implies that
\begin{align*}
    \Sigma_{j,j} \geq \Sigma_{j,-j} \Sigma_{-j,-j}^{-1} \Sigma_{-j,j}.
\end{align*}
Therefore,
\begin{align*}
    \gamma_j^\top\gamma_j = \Sigma_{j,-j} \Sigma_{-j,-j}^{-2} \Sigma_{-j,j} \leq  \frac{1}{\lambda_{\min}(\Sigma)} \Sigma_{j,-j} \Sigma_{-j,-j}^{-1} \Sigma_{-j,j} \leq \Sigma_{j,j}/\lambda_{\min}(\Sigma),
\end{align*}
yielding (a); for (b), simply note
\begin{align*}
    \lVert \gamma_j \rVert_1 \leq \sqrt{s_j} \lVert \gamma_j \rVert_2,
\end{align*}
where we note that the cardinality of $\gamma_j$ is $s_j$, by Lemma \ref{lem:aux:1}.
\qed

\subsection{Properties of \texorpdfstring{$\widehat \gamma_j$}{gamma\_j}}

Recall the definition of $\widehat \Theta$ and $\widehat C$, and we use $\widehat{\Theta}_j$ and $\widehat{C}_j$ to denote the $j$-th rows of $\widehat{\Theta}$ and $\widehat{C}$, respectively,
then
\begin{align}
    \widehat{\Theta}_j = \widehat{C}_j/\widehat{\tau}_j^2.
    \label{eq:Theta-hat_j}
\end{align}
The KKT condition implies that $\widehat{\tau}^2_j =  (X_j - {\bf X}_{-j} \widehat\gamma_j)^\top X_j/ n$ for $\widehat{\gamma}_j$ in Equation \eqref{eq:approx_inver_LASSO}, so that
\begin{align}
    \frac{1}{n} X_j^\top {\bf X} \widehat\Theta_j^\top = \frac{1}{n} X_j^\top (X_j - {\bf X}_{-j} \widehat{\gamma}_j) / \widehat{\tau}_j^2 = 1.
    \label{eq:KKT:1}
\end{align}
Furthermore, by the KKT condition, we have
\begin{align}
    \left\lVert \frac{1}{n} {\bf X}_{-j}^\top {\bf X} \widehat{\Theta}_j^\top \right\rVert_\infty = \left\lVert \frac{1}{n} {\bf X}_{-j}^\top (X_j - {\bf X}_{-j} \widehat{\gamma}_j)/\widehat{\tau}_j^2 \right\rVert_\infty \leq \lambda_j/\widehat{\tau}_j^2,
    \label{eq:KKT:2}
\end{align}
by $0 \leq \hat w_{j,k} \leq 1, k \neq j$. Combining \eqref{eq:KKT:1} and \eqref{eq:KKT:2} yields
\begin{align}
    \left\lVert \widehat{\Sigma} \widehat{\Theta}_j^\top - e_j\right\rVert_\infty \leq \lambda_j/\widehat{\tau}_j^2,
    \label{eq:extended-KKT}
\end{align}
where $e_j \in \mathbb{R}^p$ has element one on its $j$-th coordinate, zero otherwise. \eqref{eq:extended-KKT} is the so-called {\it extended KKT equation} in \cite{van-de-geer(2014)}. Finally, consideration over all $j=1,...,p$ yields
\begin{align}
    \left\lVert \widehat{\Sigma} \widehat{\Theta}^\top - I_p\right\rVert_{\max} \leq \max_{j \in [p]} \lambda_j/\widehat{\tau}_j^2.
    \label{node-wise_LASSO_kkt}
\end{align}
See, for example, \cite{van-de-geer(2014)} for the original LASSO and \cite{caner2018asymptotically} for the weighted LASSO. The following results are useful for deriving the properties of $\widehat \tau_j^2$ and $\widehat \Theta$. 

\begin{lem}[$\ell_1$ and $\ell_2$ bound for $\widehat \gamma_j$]\label{lem:hatgamma}
    Suppose the RE condition (see Definition \ref{defn:sparse-RE}) holds for ${\bf X}$ with parameter $(s_0^*+1,3)$ where $s_0^*$ is defined in Section \ref{subsec:1}. Assume also that $\hat w_{j,k} = 1, k \in S_{j}^c$ where $S_{j} = \{k \neq j: \gamma_{j,k} \neq 0\}$. For each $j \in [p]$, recall that $\eta_j = X_j - {\bf X}_{-j} \gamma_j$, then for any $\lambda_j \geq  2 \left\lVert {{\bf X}_{-j}^\top \eta_j}/{n}\right\lVert_\infty $, we have
    \begin{align*}
        (a) \;\lVert \widehat{\gamma}_j - \gamma_j \rVert_2 \lesssim \lambda_j \frac{\sqrt{s_j}}{\phi_0}, \quad (b) \;\lVert {\bf X}_{-j}(\widehat{\gamma}_j - \gamma_j) \rVert_2^2/n \lesssim \lambda^2_j \frac{s_j}{\phi_0}, \quad (c)\; \lVert \widehat{\gamma}_j - \gamma_j \rVert_1 \lesssim \lambda_j \frac{s_j}{\phi_0},
    \end{align*}
    where $\widehat \gamma_j$ is defined in \eqref{eq:approx_inver_LASSO}.
\end{lem}
\noindent \textbf{Proof of Lemma \ref{lem:hatgamma}:} Define $\widehat{\Delta}_j := \widehat{\gamma}_j - \gamma_j$ and let $\widetilde{\Delta}_j \in \mathbb{R}^p$ be such that $\widetilde{\Delta}_j(j)=0$ and $\widetilde{\Delta}_j(k) = \widehat{\Delta}_j(k),k\neq j$. The proof is then the same as in Lemma \ref{lem:main_LASSO} by replacing $\widehat{\Delta}$ with $\widetilde{\Delta}_j,j\in[p]$.
\qed

\begin{lem}[Implications of KKT conditions]
    For $\lambda_j \geq 2 \left\lVert {\bf X}_{-j}^\top \eta_j/n \right\rVert_\infty$ where $\eta_j=X_j - {\bf X}_{-j} \gamma_j$, we have 
    \begin{align*}
        \left\lVert {\bf X}^{\top}_{-j} {\bf X}_{-j}(\widehat{\gamma}_j - \gamma_j) \right\rVert_\infty /n = O(\lambda_j).
    \end{align*}
    \label{lem:kkt}
\end{lem}
\noindent \textbf{Proof of Lemma \ref{lem:kkt}:} An application of \eqref{eq:Theta-hat_j}, KKT conditions \eqref{eq:KKT:2}, as well as noting that ${\bf X} \widehat{C}_j = X_j - {\bf X}_{-j} \widehat{\gamma}_j$ yields
\begin{align*}
    \lVert {\bf X}_{-j}^\top (X_j - {\bf X}_{-j} \widehat{\gamma}_j) \lVert_\infty/n = \lVert{\bf X}_{-j}^\top {\bf X} \widehat{C}_j^\top \rVert_\infty/n = \widehat{\tau}_j^2 \lVert {\bf X}_{-j}^\top {\bf X} \widehat{\Theta}_j^\top \rVert_\infty/n \leq \lambda_j.
\end{align*}
Recall that $\widehat{\Delta}_j= \widehat{\gamma}_j - \gamma_j$, so that the preceding equation can be rewritten as
\begin{align*}
     \left\lVert {\bf X}^{\top}_{-j}(\eta_j - {\bf X}_{-j} \widehat{\Delta}_j) \right\rVert_\infty /n \leq \lambda_j.
\end{align*}
Since $\lVert a - b \rVert_\infty \geq \lVert a \rVert_\infty - \lVert b \rVert_\infty$, letting $a = {\bf X}_{-j}^\top {\bf X}_{-j} \widehat{\Delta}_j$ and $b = {\bf X}_{-j}^\top \eta_j$ yields
\begin{align*}
    \lVert {\bf X}_{-j}^\top {\bf X}_{-j} \widehat{\Delta}_j \rVert_\infty/n \leq \lambda_j + \lVert {\bf X}_{-j}^\top \eta_j \rVert_\infty/n \leq \frac{3}{2} \lambda_j.
\end{align*}
The proof is complete.
\qed

\subsection{Properties of \texorpdfstring{$\widehat \tau_j^2$}{tau\_j} and \texorpdfstring{$\widehat \Theta$}{Theta}}

We impose the following conditions to derive the properties of $\widehat \tau_j^2$ and $\widehat \Theta$.
\begin{ass} \label{ass:marginal_SW}
    Suppose 
    \begin{enumerate}
    \item $ \max_{i \in [n], j \in [p]} \lVert X_{ij} \rVert_{\psi_\nu} \leq K_{n,p}$,
    \item $\widetilde{X}_i^\top \indep \widetilde{X}_j^\top$ for $i \neq j$,
    \item $\max_{i \in [n], j \in [p]} \lVert \eta_{j}(i) \rVert_l \leq \overline{c} <\infty$ for some $l \geq 4$,
    \item $0 < \underline{c} \leq\lambda_{\min}(\Sigma)$ and $\max_{j \in [p]}\Sigma_{jj} = O(1)$.
    \end{enumerate}
\end{ass}

\begin{lem}  \label{lem:hattau}
    Suppose Assumption \ref{ass:marginal_SW} holds. On the event that $\hat w_{j,k} = 1, k \in S_{j}^c$ where $S_{j} = \{k \neq j: \gamma_{j,k} \neq 0\}$, $\lambda_j \geq  2 \left\lVert {{\bf X}_{-j}^\top \eta_j}/{n}\right\lVert_\infty $, and the RE condition holds for ${\bf X}$ with parameter $(s_0^*+1,3)$, we have
    \begin{align*}
        \max_{j \in [p]} |\widehat{\tau}_j^2 - \tau_j^2| = O_\Prob\left(\max_{j \in [p]} \lambda_j\sqrt{s_j} + \max_{j \in [p]} \lambda_j^2 s_j\right) + O_\Prob \left( \sqrt{\frac{\log (np)}{n} } + \frac{K_{n,p}^2 s_0^* (\log (np))^{4/\nu}}{n} \right).
    \end{align*}
\end{lem}

\noindent \textbf{Proof of Lemma \ref{lem:hattau}:} Recall that $\widehat{\Delta}_j = \widehat{\gamma}_j - \gamma_j$. We can write for each $j = 1,...,p$,
\begin{align*}
    |\widehat{\tau}_j^2 - \tau_j^2| &\leq |\eta_j^\top\eta_j/n - \tau_j^2| + |\eta_j^\top {\bf X}_{-j} \widehat{\Delta}_j/n| + | \eta_j^\top {\bf X}_{-j} \gamma_j/n| + |\gamma_j^\top {\bf X}_{-j}^\top {\bf X}_{-j } \widehat{\Delta}_j/n| \\
    & = I + II + III + IV.
\end{align*}

We first prove the order for $I$. Fix any $j \in [p]$ and note first that for any $i \in [n]$,
\begin{align*}
    \eta_{j}(i) = X_{ij} - \sum_{k \neq j, k \in [p]} X_{ik} \gamma_{j}(k),
\end{align*}
where we denote $\gamma_j(k)$ and $\eta_j(k)$ to be the $k$-th element of the vectors $\gamma_j$ and $\eta_j$ respectively. By Lemma \ref{lem:sum_of_sub-Weibull_is_sub-Weibull}, $\eta_j(i)$ is sub-Weibull with
\begin{align*}
    \lVert \eta_{j}(i) \rVert_{\psi_\nu} \leq \max\{1,C(\nu)\} K_{n,p} \left( 1+ \lVert \gamma_j \rVert_1 \right) \leq C_\nu K_{n,p} \sqrt{s_0^*},
\end{align*} 
uniformly over $i \in [n], j \in [p]$ and $C_\nu>0$ is a universal constant, where we use Lemma \ref{lem:aux:3}(b) in the last inequality. Then for every $i \in [n], j \in [p]$,
\begin{align*}
    \lVert \eta_{j}(i)^2 - \mathbb{E}(\eta_{j}(i)^2) \rVert_{\psi_{\nu/2}} \leq C_{\psi_{\nu/2}} \lVert \eta^2_{j}(i)\rVert_{\psi_{\nu/2}} = C_{\psi_{\nu/2}} \lVert \eta_{j}(i) \rVert_{\psi_\nu}^2 \lesssim K_{n,p}^2 s_0^*,
\end{align*}
so that
\begin{align*}
    \frac{1}{n}\eta_j^\top\eta_j - \tau_j^2 = \frac{1}{n} \sum_{i=1}^n \eta_{j}(i)^2 - \tau_j^2 = \frac{1}{n} \sum_{i=1}^n (\eta_{j}(i)^2 - \mathbb{E}[\eta_{j}(i)^2])
\end{align*}
is the sum of $n$ independent mean-zero random variables with sub-Weibull parameter $\nu/2$ and $O(K_{n,p}^2 s_0^*)$. Applying \cite[Theorem 3.4]{kuchibhotla2022moving} with $t= \log (np)$ yields
\begin{align*}
    \max_{j \in [p]} |\eta_j^\top\eta_j/n - \tau_j^2| = O_\Prob \left( \sqrt{\frac{\log (np)}{n} } + \frac{K_{n,p}^2 s_0^* (\log (np))^{4/\nu}}{n} \right),
\end{align*}
where we also use the assumption that $\max_{i \in [n], j \in [p]} \lVert \eta_{j}(i) \rVert_l \leq \overline{c} <\infty$ for some $l \geq 4$.

We now prove the order for the remaining terms $II,III,IV$. For $II$, by using Lemma \ref{lem:hatgamma}(c) we have
\begin{align*}
    II \leq \lVert\eta_j^\top {\bf X}_{-j}/n \rVert_\infty \lVert \widehat{\Delta}_j \rVert_1 = \lVert\eta_j^\top {\bf X}_{-j}/n \rVert_\infty \times O(\lambda_j s_j) \leq \frac{1}{2} \lambda_j \times O(\lambda_j s_j) = O(s_j \lambda_j^2).
\end{align*}
For $III$ we have
\begin{align*}
    III \leq \lVert \gamma_j \rVert_1 \lVert\eta_j^\top {\bf X}_{-j}/n \rVert_\infty = O(\sqrt{s_j}) \times\lVert\eta_j^\top {\bf X}_{-j}/n \rVert_\infty = O(\sqrt{s_j} \lambda_j).
\end{align*}
For $IV$ we have
\begin{align*}
    IV \leq \lVert \gamma_j \rVert_1 \lVert {\bf X}_{-j}^\top {\bf X}_{-j } \widehat{\Delta}_j/n \rVert_\infty = O(\sqrt{s_j}) O(\lambda_j) = O(\sqrt{s_j} \lambda_j).
\end{align*}
where the last equality follows from the KKT conditions shown in Lemma \ref{lem:kkt}. The proof is complete. 
\qed

\begin{lem}\label{lem:hatTheta}
    Suppose Assumption \ref{ass:marginal_SW} holds, and
    \begin{align*}
         \max_{j \in [p]} |\widehat{\tau}_j^2 - \tau_j^2| = o_\Prob(1).
    \end{align*}
    On the event that $\hat w_{j,k} = 1, k \in S_{j}^c$ where $S_{j} = \{k \neq j: \gamma_{j,k} \neq 0\}$, $\lambda_j \geq  2 \left\lVert {{\bf X}_{-j}^\top \eta_j}/{n}\right\lVert_\infty $, and the RE condition holds for ${\bf X}$ with parameter $(s_0^*+1,3)$, we have
    \begin{align*}
        &(a) \; \max_{j \in [p]} \lVert  \widehat{\Theta}_j - \Theta_j \rVert_1 = O_\Prob\left(\max_{j \in [p]} \lambda_js_j \right) + O_\Prob\left(\sqrt{s_0^*} \max_{j \in [p]} |\widehat{\tau}_j^2 - \tau_j^2| \right), \\
        &(b) \; \max_{j \in [p]} \lVert  \widehat{\Theta}_j - \Theta_j \rVert_2 =  O_\Prob\left(\max_{j \in [p]} \lambda_j\sqrt{s_j} \right) + O_\Prob\left(\max_{j \in [p]} |\widehat{\tau}_j^2 - \tau_j^2| \right).
    \end{align*}   
\end{lem}
\noindent \textbf{Proof of Lemma \ref{lem:hatTheta}:} Recall by Lemma \ref{lem:aux:1}{(b)} and Lemma \ref{lem:aux:2}{(a)} that $\gamma_j = - \frac{\Theta_{-j,j}}{\Theta_{j,j}} = - \tau_j^2 \Theta_{-j,j}$, so that $\Theta_j = C_j/\tau_j^2$. Therefore, by Equation \eqref{eq:Theta-hat_j},
\begin{align*}
    \lVert  \widehat{\Theta}_j - \Theta_j \rVert_1 &= \left\lVert \widehat{C}_j/\widehat{\tau}_j^2 - C_j/\tau_j^2 \right\rVert_1 \\
    &\leq \left\lVert \widehat{C}_j/\widehat{\tau}_j^2 - C_j/\widehat{\tau}_j^2 \right\rVert_1 + \left\lVert C_j/\widehat{\tau}_j^2 - C_j/\tau_j^2 \right\rVert_1 \\
    &\leq \lVert \widehat{\gamma}_j - \gamma_j \rVert_1 / \widehat{\tau}_j^2 + \left(1+\lVert \gamma_j \rVert_1 \right) \left|\frac{1}{\widehat{\tau}_j^2} - \frac{1}{\tau_j^2}\right| \\
    &= I + II.
\end{align*} 
For $I$, by an application of Lemma \ref{lem:hatgamma}{(c)} and Lemma \ref{lem:hattau}, we have
\begin{align}
    I &= \frac{O(\lambda_j s_j)}{\tau_j^2 + (\widehat{\tau_j}^2 - \tau_j^2)} = O(\lambda_j s_j) O_\Prob(1) = O_\Prob(\max_{j \in [p]} \lambda_j s_j),
    \label{eq:hatTheta:1}
\end{align}
uniformly over $j \in [p]$, where the first equality follows by assumption of $\widehat{\tau}_j^2$ converging to $\tau_j^2$ so that 
\begin{align*}
    \frac{1}{\tau_j^2 + (\widehat{\tau_j}^2 - \tau_j^2)} = \frac{1}{\tau_j^2} - \frac{\widehat{\tau_j}^2 - \tau_j^2}{\tau_j^2(\tau_j^2+(\widehat{\tau_j}^2 - \tau_j^2))} = \frac{1}{\tau_j^2} + o_\Prob(1) = O_\Prob(1),
\end{align*}
the last equality following from Lemma \ref{lem:aux:2}(b). For $II$, recall from Lemma \ref{lem:aux:3}(b) that
$\lVert \gamma_j \rVert_1 = O\left(\sqrt{s_j}\right)$. By using Lemma \ref{lem:hattau} we have 
\begin{align}
    II  = O_\Prob\left(\sqrt{s_0^*} \max_{j \in [p]} |\widehat{\tau}_j^2 - \tau_j^2| \right),
    \label{eq:hatTheta:2}
\end{align}
since
\begin{align*}
    \frac{1}{\widehat{\tau}_j^2} - \frac{1}{\tau_j^2} = -\frac{\widehat{\tau}_j^2 - \tau_j^2}{\tau_j^2 ((\widehat{\tau}_j^2 - \tau_j^2) + \tau_j^2)} = - \frac{\widehat{\tau}_j^2 - \tau_j^2}{\tau_j^2(o_\Prob(1) + \tau_j^2)} = O_\Prob(\max_{j \in [p]} |\widehat{\tau}_j^2 - \tau_j^2|).
\end{align*}
Combining \eqref{eq:hatTheta:1} and \eqref{eq:hatTheta:2} yields part (a). 

\vspace{5mm}
For part (b), we have
\begin{align*}
    \lVert  \widehat{\Theta}_j - \Theta_j \rVert_2 &= \left\lVert \widehat{C}_j/\widehat{\tau}_j^2 - C_j/\tau_j^2 \right\rVert_2 \\
    &\leq \left\lVert \widehat{C}_j/\widehat{\tau}_j^2 - C_j/\widehat{\tau}_j^2 \right\rVert_2 + \left\lVert C_j/\widehat{\tau}_j^2 - C_j/\tau_j^2 \right\rVert_2 \\
    &\leq \lVert \widehat{\gamma}_j - \gamma_j \rVert_2 / \widehat{\tau}_j^2 + \left(1+\lVert \gamma_j \rVert_2 \right) \left|\frac{1}{\widehat{\tau}_j^2} - \frac{1}{\tau_j^2}\right| \\
    &= I + II.
\end{align*}
Repeat as in \eqref{eq:hatTheta:1} and use Lemma \ref{lem:hatgamma}{(a)} (instead of \ref{lem:hatgamma}{(c)}) to obtain
\begin{align}
    I = O(\lambda_j \sqrt{s_j}) O_\Prob(1) = O_\Prob(\max_{j \in [p]} \lambda_j \sqrt{s_j}).
    \label{eq:hatTheta:3}
\end{align}
For $II$, repeat as in \eqref{eq:hatTheta:2} to obtain
\begin{align}
    II =  O_\Prob(\max_{j \in [p]} |\widehat{\tau}_j^2 - \tau_j^2|), 
    \label{eq:hatTheta:4}
\end{align}
since $\lVert \gamma_j \rVert_2 = O(1)$ by Lemma \ref{lem:aux:3}(a). Combining \eqref{eq:hatTheta:3} and \eqref{eq:hatTheta:4} yields part (b). The proof is complete.
\qed

\begin{rem}
By repeating the proof of Lemma \ref{lem:random_design_X}, we show that the RE condition holds with probability approaching one if $\delta_n s_0^* = o_\Prob(1)$, where
\begin{align*}
    \delta_n = O_\Prob\left( \max\left\{K_{n,p}^2 \sqrt{\frac{\log (np)}{n}}, K_{n,p}^2 \frac{(\log (np))^{4/\nu}}{n} \right\} \right). 
\end{align*}
In addition, we show that $\hat w_{j,k} = 1, k \in S_{j}^c$ holds for the conservative LASSO with probability approaching one in Theorem \ref{thm:size-control_random_design_extension}. Finally, the Lemmas \ref{lem:hatgamma}--\ref{lem:hatTheta} depend on the unknown parameters $\lambda_j$, and we provide this rate below.
\end{rem}

\begin{lem}[Sub-Weibull tuning rate for $\lambda_j$]
\label{lem:node-wise_tuning_rate}
Suppose Assumption \ref{ass:marginal_SW} holds and $(\log(np))^{2/\nu - 1/2} = o(n^{1/2-1/l})$. Let $\left\{\lambda_j\right\}_{j=1}^p$ be a sequence satisfying
\begin{align*}
\lambda_j \geq 2 C_{\nu,l} 
K_{n,p} \sqrt{ \frac{2\log (np)}{n}},
\qquad j\in[p],
\end{align*}
where the constant $C_{\nu,l}$ only depends on $\nu$ and $l$. Then
\begin{align*}
\mathbb{P}\!\left( 2 \left\lVert \frac{{\bf X}_{-j}^\top \eta_j}{n} \right\rVert_\infty
\leq \lambda_j, \quad \forall j \in [p]
\right)\to 1
\qquad\text{as}\; n\to\infty.
\end{align*}
\end{lem}

\noindent \textbf{Proof of Lemma \ref{lem:node-wise_tuning_rate}:} We may repeat the argument used in Lemma~\ref{lem:bound_of_Xeps} by working with the stacked object
\begin{align*}
\begin{bmatrix}
{\bf X}_{-1}^\top \eta_1/n \\
{\bf X}_{-2}^\top \eta_2/n \\
\vdots \\
{\bf X}_{-p}^\top \eta_p/n
\end{bmatrix}
\in \mathbb{R}^{p(p-1)},
\end{align*}
which stacks the vectors $\{{\bf X}_{-j}^\top \eta_j/n\}_{j=1}^p$. For each $j\in[p]$, recall that $\eta_j = X_j-{\bf X}_{-j}\gamma_j$, where $\gamma_j$ is defined at the population level by the normal equation
\begin{align*}
\mathbb E[{\bf X}_{-j}^\top X_j] = \mathbb E[{\bf X}_{-j}^\top {\bf X}_{-j}]\,\gamma_j .
\end{align*}
Hence,
\begin{align*}
\mathbb E[{\bf X}_{-j}^\top \eta_j]
&= \mathbb E[{\bf X}_{-j}^\top (X_j-{\bf X}_{-j}\gamma_j)]
 = \mathbb E[{\bf X}_{-j}^\top X_j]-\mathbb E[{\bf X}_{-j}^\top {\bf X}_{-j}]\,\gamma_j
=0.
\end{align*}
We can thus follow the proof of \cite[Theorem~4.6]{kuchibhotla2022moving} to obtain that, with probability at least $1 - 3/(np(p-1)) - 1/L$ for any $L\ge 1$,
\begin{align*}
\max_{j \in [p]}\left\Vert \frac{{\bf X}_{-j}^\top \eta_j}{n}  \right\Vert_{\infty}
\le
7\sqrt{2}\,\sigma_{n,p}\sqrt{\frac{\log(np(p-1))}{n}}
+
\frac{C_{\nu} K_{n,p}K_{\eta,l} \,(\log(np(p-1)))^{1/\nu}\Big[(\log(2n))^{1/\nu}+L\Big]}{n^{1-1/l}},
\end{align*}
where $K_{\eta,l} := \max_{i \in [n], j \in [p]} \lVert \eta_{j}(i) \rVert_l$ and
\begin{align*}
    \sigma^2_{n,p} := \max_{j \in [p]} \max_{k \in [p-1]} \frac{1}{n} \sum_{i \in [n]} \mathbb{E}\left[ \eta_{j}(i)^2 [{\bf X}_{-j}]^2_{i,k}  \right].
\end{align*}
 Repeating the proof of Lemma \ref{lem:bound_of_Xeps} reveals that there exists $N(\nu,l) \in \mathbb N^+$ such that, for all $n \geq N(\nu,l)$, we have
    \begin{align*}
        \mathbb P \left(  \max_{j \in [p]}\left\Vert \frac{{\bf X}_{-j}^\top \eta_j}{n}  \right\Vert_{\infty} > C_{\nu,l} K_{n,p} \sqrt{\frac{\log (np(p-1))}{n}} \right) \leq \frac{3}{np(p-1)}+\frac{1}{(\log(2n))^{1/\nu}}.
    \end{align*}
    Setting $\lambda_j$ as above gives the result.
\qed

\begin{rem}
    The condition $\max_{i \in [n], j \in [p]} \lVert \eta_{j}(i) \rVert_l \leq \overline{c}<\infty$ is also made in \cite{van-de-geer(2014)} for the bounded case and in \cite{caner2018asymptotically}. This condition appears to be necessary to obtain the sharp tuning rate above when we only have control over the marginal distribution of $\tilde X_i$. 
    In fact, if we assume that $\tilde X_i$ is jointly sub-Weibull, then this assumption holds for all $l > 0$.
\end{rem}

\newpage
\section{Consistency of Auxiliary Estimators}
\label{sec:appendix_other_estimators}

In this section, we derive the consistency of $\widehat{\Phi}$ (and $\widetilde{\Phi}$) as well as $\{\tilde w_{i,n}\}_{i \in [r_n]}$. We derive these results under both fixed and random designs.

\subsection{Fixed Design}

We first treat $\bf X$ as fixed. To begin with, note that the following holds for $K_\theta$ and $r_n$.

\begin{lem}
    We have
    \begin{align}
    \frac{1 \vee \varpi r_n }{2} \leq  K_\theta \leq r_n,
    \label{eq:K_lambda:summary}
    \end{align}
    under both fixed and random designs.
    \label{lem:Ktheta_and_rn}
\end{lem}

\noindent {\bf Proof of Lemma \ref{lem:Ktheta_and_rn}:} Consider the singular value decomposition 
\begin{align*}
    R \widehat{\Theta} {\bf X}^\top / \sqrt{n} = U S Q^\top,
\end{align*}
where $U \in \mathbb{R}^{r \times r}$ and $Q \in \mathbb{R}^{n \times n}$ are orthonormal matrices, and $S \in \mathbb{R}^{r \times n}$ contains singular values $D_1 \geq \cdots \geq D_{r_n} > 0 $ on its diagonal and zeros elsewhere. Consequently, $ P_\theta$ admits the spectral decomposition 
\begin{align*}
    P_\theta = Q \tilde{\mathcal{D}} Q^\top, \quad \tilde{\mathcal{D}} = \diag{ \left( \frac{D_{1}^2}{(D^2_{1} + \theta)}, \cdots, \frac{D_{r_n}^2}{(D^2_{r_n} + \theta)}, \mathbf{0}_{n - r_n} \right)}.
\end{align*}
It follows that 
\begin{align*}
    K_\theta = \tr (\tilde{\mathcal{D}}) = \sum_{i=1}^{r_n} \frac{D_{i}^2} {D^2_i + \theta}.
\end{align*}
The upper bound is immediate as
\begin{align*}
    \sum_{i=1}^{r_n} \frac{D_{i}^2}{D^2_i + \theta}\leq \sum_{i=1}^{r_n} 1 = r_n.
\end{align*}
The lower bound follows from 
\begin{align*}
    \sum_{i=1}^{r_n} \frac{D_{i}^2}{D^2_i + \theta}  \geq \frac{D_{\lceil \varpi r_n \rceil}^2}{D^2_{\lceil \varpi r_n \rceil} + D^2_{\lceil \varpi r_n \rceil}} = 1/2,
\end{align*}
and 
\begin{align*}
    \sum_{i=1}^{r_n} \frac{D_{i}^2}{D^2_i + \theta}  \geq \sum_{i=1}^{\lceil \varpi r_n \rceil}  \frac{D_{i}^2}{D^2_i + \theta} \geq \sum_{i=1}^{\lceil \varpi r_n \rceil} \frac{\theta}{\theta + \theta} \geq  \varpi r_n  /2.
\end{align*}
\qed

\begin{lem}  \label{lem:variance-estimator-consistent}
Under the setup of \eqref{eq:model}, suppose the following holds: 
\begin{enumerate}
    \item $\operatorname{Var}(\epsilon_i) \geq \underline{c}$ and $ \max_{i \in [n]} \mathbb{E}[\epsilon_i^4] \leq \overline{c}$ for some universal constants $ 0 <\underline{c} \leq\overline{c}< \infty$,
    \item $p_n = \max_{i \in [n]} P_{\theta,ii} =  O(K_\theta/n)$,
    \item $\left\lVert {\bf X} (\widehat{\beta} - \beta^0 ) / \sqrt{n}\right\rVert_2 = o_\mathbb{P}(1)$.
\end{enumerate}
Then we have
\begin{align*}
    \frac{\sum_{i \in [n]} P_{\theta,ii} \widehat{\sigma}_i^2}{K_\theta} = \frac{\sum_{i \in [n]} P_{\theta,ii} \epsilon_i^2}{K_\theta} + o_\Prob(1),
\end{align*}
and 
\begin{align*}
    \frac{\sum_{i \in [n]} P_{\theta,ii} \epsilon_i^2}{K_\theta} = \frac{\sum_{i \in [n]} P_{\theta,ii} \sigma_i^2}{K_\theta} + o_\Prob(1).
\end{align*}
In particular, we have
\begin{align*}
    \frac{\sum_{i \in [n]} P_{\theta,ii} \widehat{\sigma}_i^2}{K_\theta}  = \frac{\sum_{i \in [n]} P_{\theta,ii} \sigma_i^2}{K_\theta} + o_\Prob(1).
\end{align*}
\end{lem}

\begin{rem} \label{rem:variance-estimator-consistent}
    Condition {\it 3} of Lemma \ref{lem:variance-estimator-consistent} holds under quite general settings of random design ${\bf X}$. In particular, it holds under marginal sub-Weibull design with parameter $(\nu,K_{n,p})$ (see Definition \ref{def:sub_weibull}). Indeed by Lemmas \ref{lem:bound_of_Xeps} and \ref{lem:main_LASSO}
    \begin{align*}
    \lVert {\bf X} (\widehat{\beta} - \beta^0) \rVert_2^2 = O_\mathbb{P} \left( K_{n,p}^2 s_0 \log (np)  \right).
\end{align*}
If we assume $K_{n,p}^2 s_0 \log (np) = o(n)$, then
\begin{align*}
     \lVert {\bf X} (\widehat{\beta} - \beta^0) \rVert_2^2 = o_\mathbb{P}(n).
\end{align*}
\end{rem}

\noindent \textbf{Proof:} By noting that $Y_i = \widetilde{X}_i^\top \beta^0 + \epsilon_i$ and recalling that $\widehat{\Delta}= \widehat{\beta} - \beta^0$, we claim that
\begin{align*}
    \frac{\sum_{i \in [n]} P_{\theta,ii} \widehat{\sigma}_i^2}{K_\theta} &= \frac{\sum_{i \in [n]} P_{\theta,ii} \epsilon_i^2}{K_\theta} - 2 \frac{\sum_{i \in [n]} P_{\theta,ii} \epsilon_i d_i }{K_\theta} + \frac{\sum_{i \in [n]} P_{\theta,ii} d_i^2}{K_\theta} \notag \\
    &= \frac{\sum_{i \in [n]} P_{\theta,ii} \epsilon_i^2}{K_\theta} + I + II \\
    &= \frac{\sum_{i \in [n]} P_{\theta,ii} \epsilon_i^2}{K_\theta} + o_\mathbb{P} (1),
\end{align*}
where $d_i = \widetilde{X}_i^\top \widehat{\Delta}$. For $II$, observe 
\begin{align*}
    II  \leq \frac{p_n}{K_\theta} \sum_{i \in [n]} d_i^2 = \frac{n \, p_n}{K_\theta} \frac{\lVert {\bf X} \widehat{\Delta} \rVert_2^2 }{n} = O \left( \frac{\lVert {\bf X} \widehat{\Delta} \rVert_2^2 }{n}  \right),
\end{align*}
since $n p_n / K_\theta = O (1)$ by assumption. For $I$, by Cauchy-Schwarz inequality 
\begin{align*}
    |I| \lesssim  \sqrt{\frac{\sum_{i \in [n]} P_{\theta,ii}^2 \epsilon_i^2  }{K_\theta^2}}  \times  \sqrt{\sum_{i \in [n]} d_i^2} = O_\mathbb{P} \left( \frac{\lVert {\bf X} \widehat{\Delta} \rVert_2}{\sqrt{n}} \right),
\end{align*}
where the inequality follows from Markov inequality,
\begin{align*}
    &\mathbb{E}_{\epsilon} \left[\frac{\sum_{i \in [n]} P_{\theta,ii}^2 \epsilon_i^2 }{K_\theta^2}\right] \leq \frac{p_n}{K_\theta^2} \sum_{i \in [n]} P_{\theta,ii} \sigma_i^2 \lesssim \frac{p_n}{K_\theta} = O\left( \frac{1}{n}  \right).
\end{align*}
Therefore,
\begin{align*}
    \frac{\sum_{i \in [n]} P_{\theta,ii} \widehat{\sigma}_i^2}{K_\theta} &= \frac{\sum_{i \in [n]} P_{\theta,ii} \epsilon_i^2}{K_\theta} +  O_\mathbb{P} \left(\frac{\lVert {\bf X} \widehat{\Delta} \rVert_2^2 }{n} + \frac{\lVert {\bf X} \widehat{\Delta} \rVert_2}{\sqrt{n}} \right) \\
    &= \frac{\sum_{i \in [n]} P_{\theta,ii} \epsilon_i^2}{K_\theta} + o_\Prob(1).
\end{align*}
We also have
\begin{align*}
    \mathbb{E}_{\eps} \left( \frac{\sum_{i \in [n]} P_{\theta,ii} (\epsilon_i^2 - \sigma_i^2)}{K_\theta} \right)^2  \lesssim \frac{\sum_{i \in [n]} P_{\theta,ii}^2}{K_\theta^2} = O\left(\frac{p_n}{K_\theta}\right) = O\left(\frac{1}{n}\right) = o(1).
\end{align*}
This completes the proof.
\qed

\begin{lem}
    Under the conditions of Lemma \ref{lem:variance-estimator-consistent}, and assume that $\lVert {\bf X} (\widehat{\beta} - \beta^0 ) \rVert_{\infty} \times \lVert {\bf X} (\widehat{\beta} - \beta^0 ) \rVert_2 = o_\Prob(\sqrt{n})$ and $r=o(n)$. Then, 
    \begin{align*}
         \widehat{\Phi}-\Phi  = o_\Prob(1).
    \end{align*}
    If we further assume that $\Phi \geq \underline{c}>0$, then
    \begin{align*}
        \frac{\widehat\Phi}{\Phi} = 1 + o_\Prob(1).
    \end{align*}
    The same results hold for the adjusted variance estimator
    \begin{align*}
        \widetilde \Phi = \widehat \Phi + \frac{1}{K_\theta} \sum_{i \in [n]} P_{\theta,ii}^2 \widehat \sigma_i^4.
    \end{align*}
    \label{lem:consistency:phi-hat}
\end{lem}

\begin{rem} \label{rem:consistency:phi-hat}
    Again the conditions here hold under general marginal sub-Weibull design with parameter $(\nu,K_{n,p})$ (see Definition \ref{def:sub_weibull}), as in Remark \ref{rem:variance-estimator-consistent}. Indeed, it follows from Lemmas \ref{lem:bound_of_X}, \ref{lem:main_LASSO} and \ref{lem:random_design_X} that
    \begin{align*}
    \lVert {\bf X} (\widehat{\beta} - \beta^0 ) \rVert_{\infty} \times \lVert {\bf X} (\widehat{\beta} - \beta^0 ) \rVert_2 \lesssim K_{n,p}^3 \sqrt{s_0^3 \log(np)^{2+2/\nu}} / \sqrt{n} = o(\sqrt{n}),
    \end{align*}
    whenever $K_{n,p}^3 s_0^{3/2} \log(np)^{1+1/\nu} = o(n)$.
\end{rem}

\noindent{\bf Proof of Lemma \ref{lem:consistency:phi-hat}:}
We begin by first establishing
\begin{align}
    \frac{2}{K_\theta} \sum_{i \in [n]} \sum_{j \neq i} P_{\theta,ij}^2 \epsilon_i^2 \epsilon_j^2 - \Phi  = o_\Prob(1).
    \label{eq:consistency:phi-hat:1}
\end{align}
We have
\begin{align*}
    \frac{2}{K_\theta} \sum_{i \in [n]} \sum_{j \neq i} P_{\theta,ij}^2 \epsilon_i^2 \epsilon_j^2 - \Phi  &= \frac{2}{K_\theta} \sum_{i \in [n]} \sum_{j \neq i} P_{\theta,ij}^2 \epsilon_i^2 (\epsilon_j^2 - \sigma_j^2) + \frac{2}{K_\theta} \sum_{i \in [n]} \sum_{j \neq i} P_{\theta,ij}^2 (\epsilon_i^2 - \sigma_i^2) \sigma_j^2 = I + II.
\end{align*}
We claim that $II = o_\Prob(1)$ by Markov inequality and observing
\begin{align*}
    \mathbb{E}_{\eps} \left[ II^2 \right] &\lesssim \frac{1}{K_\theta^2} \sum_{i \in [n]} \mathbb{E}[(\epsilon_i^2 - \sigma_i^2)^2] \left(\sum_{j \neq i} P_{\theta,ij}^2 \sigma_j^2 \right)^2 \\
    &\lesssim \frac{1}{K_\theta^2} \sum_{i \in [n]} \left(  \sum_{j \neq i} P_{\theta,ij}^2 \right)^2 \\
    &\leq \frac{1}{K_\theta^2} \sum_{i \in [n]} P_{\theta,ii}^2 \\
    &\leq \frac{p_n}{K_\theta} = O(\frac{1}{n}) = o(1),
\end{align*}
where the second-last inequality follows from $0\leq P_{\theta,ii} \leq 1$ by the proof of Lemma \ref{lem:Ktheta_and_rn} and 
\begin{align}
    \sum_{j \neq i} P_{\theta,ij}^2 \leq \sum_{j \in [n]} P_{\theta,ij}^2 \leq P_{\theta,ii}
    \label{off_diagonal_ineq}
\end{align}
by Lemma 1(ii) in the Appendix of \cite{dovi2024ridge}. For $I$, first observe that
\begin{align*}
    \frac{2}{K_\theta} \sum_{i \in [n]} \sum_{j \neq i} P_{\theta,ij}^2 (\epsilon_i^2 - \sigma_i^2)(\epsilon_j^2 - \sigma_j^2) = o_\Prob(1),
\end{align*}
which follows from Markov inequality and
\begin{align*}
    \mathbb{E}_{\eps} \left[\frac{2}{K_\theta}  \sum_{i \in [n]} \sum_{j \neq i}  P_{\theta,ij}^2 (\epsilon_i^2 - \sigma_i^2)(\epsilon_j^2 - \sigma_j^2) \right]^2 &\lesssim \frac{1}{K_\theta^2} \sum_{i \in [n]} \sum_{j \neq i} P_{\theta,ij}^4 \mathbb{E}[(\epsilon_i^2 - \sigma_i^2)^2] \mathbb{E}[(\epsilon_j^2 - \sigma_j^2)^2] \\
    &\lesssim \frac{1}{K_\theta^2} \sum_{i \in [n]} \sum_{j \neq i} P_{\theta,ij}^4 \\
    &\leq \frac{p_n^2}{K_\theta^2}  \sum_{i \in [n]} \sum_{j \neq i} P_{\theta,ij}^2 \\
    &\leq \frac{p_n^2}{K_\theta} = O(\frac{1}{n}) = o(1).
\end{align*}
Therefore,
\begin{align*}
    I = \frac{2}{K_\theta} \sum_{i \in [n]} \sum_{j \neq i} P_{\theta,ij}^2 \sigma_i^2 (\epsilon_j^2 - \sigma_j^2) + o_\Prob(1) = II + o_\Prob(1) = o_\Prob(1).
\end{align*}
We have thus established \eqref{eq:consistency:phi-hat:1}.

\vspace{5mm}
Next, we will show that
\begin{align}
     \frac{2}{K_\theta} \sum_{i \in [n]} \sum_{j \neq i} P_{\theta,ij}^2 \widehat{\sigma}_i^2 \widehat{\sigma}_j^2 - \frac{2}{K_\theta} \sum_{i \in [n]} \sum_{j \neq i} P_{\theta,ij}^2 \epsilon_i^2 \epsilon_j^2 = o_\Prob(1),
     \label{eq:consistency:phi-hat:2}
\end{align}
so that combining with \eqref{eq:consistency:phi-hat:1} completes the proof of Lemma \ref{lem:consistency:phi-hat}. 
Recall that 
\begin{align*}
    \widehat{\sigma}_i^2 = \epsilon_i^2 + d_i^2 - 2 \epsilon_i d_i \equiv \epsilon_i^2 + V_i.
\end{align*}
Therefore, proving \eqref{eq:consistency:phi-hat:2} is equivalent to showing 
\begin{align*}
     &\frac{2}{K_\theta} \sum_{i \in [n]} \sum_{j \neq i} P_{\theta,ij}^2 \epsilon_i^2 V_j + \frac{2}{K_\theta} \sum_{i \in [n]} \sum_{j \neq i} P_{\theta,ij}^2 \epsilon_j^2 V_i + \frac{2}{K_\theta} \sum_{i \in [n]} \sum_{j \neq i} P_{\theta,ij}^2 V_i V_j \notag \\
     &\equiv I + II + III = o_\mathbb{P}(1).
\end{align*}
For $I$, we claim that
\begin{align*}
    I = \frac{2}{K_\theta} \sum_{i \in [n]} \sum_{j \neq i} P_{\theta,ij}^2 \sigma_i^2 V_j + o_\Prob(1) = o_\Prob(1).
\end{align*}
To begin, observe
\begin{align*}
    I - \frac{2}{K_\theta} \sum_{i \in [n]} \sum_{j \neq i} P_{\theta,ij}^2 \sigma_i^2 V_j &= \frac{2}{K_\theta} \sum_{i \in [n]} \sum_{j \neq i} P_{\theta,ij}^2 (\epsilon_i^2 -\sigma_i^2) d_j^2 - \frac{4}{K_\theta} \sum_{i \in [n]} \sum_{j \neq i} P_{\theta,ij}^2 (\epsilon_i^2 -\sigma_i^2) \epsilon_j d_j \\
    & = \frac{2}{K_\theta} \sum_{i \in [n]} \sum_{j \neq i} P_{\theta,ij}^2 (\epsilon_j^2 -\sigma_j^2) d_i^2 - \frac{4}{K_\theta} \sum_{i \in [n]} \sum_{j \neq i} P_{\theta,ij}^2 (\epsilon_j^2 -\sigma_j^2) \epsilon_i d_i \\
    & = I.1 - I.2.
\end{align*}
Then $I.1 = o_\Prob(1)$ since 
\begin{align*}
    |I.1| &\lesssim \sqrt{\frac{1}{K_\theta} \sum_{i \in [n]} \sum_{j \neq i} P_{\theta,ij}^2 \left(\epsilon_j^2 -\sigma_j^2\right)^2} \sqrt{\frac{1}{K_\theta}\sum_{i \in [n]} \sum_{j \neq i} P_{\theta,ij}^2 d_i^4} \\
    &= O_\Prob \left(\sqrt{\frac{p_n}{K_\theta} \sum_{i \in [n]} d_i^4} \right) \\
    &= O_\Prob \left(\sqrt{\frac{\lVert {\bf X} \widehat{\Delta} \rVert_{\infty}^2 \times \lVert {\bf X} \widehat{\Delta} \rVert_2^2}{n}} \right) = o_\Prob(1)
\end{align*}
by Cauchy-Schwarz inequality, where
\begin{align*}
    \mathbb{E}_\epsilon \left[ \frac{1}{K_\theta} \sum_{i \in [n]} \sum_{j \neq i} P_{\theta,ij}^2 \left(\epsilon_j^2 -\sigma_j^2\right)^2 \right] &=  \frac{1}{K_\theta} \sum_{i \in [n]} \sum_{j \neq i} P_{\theta,ij}^2 \mathbb{E}_\epsilon\left(\epsilon_j^2 -\sigma_j^2\right)^2 \\
    &\lesssim \frac{1}{K_\theta} \sum_{i \in [n]} \sum_{j \neq i} P_{\theta,ij}^2 \\
    &= O(1),
\end{align*}
and 
\begin{align*}
    \frac{1}{K_\theta}\sum_{i \in [n]} \sum_{j \neq i} P_{\theta,ij}^2 d_i^4 \leq \frac{1}{K_\theta}\sum_{i \in [n]}P_{\theta,ii} d_i^4 \leq \frac{p_n}{K_\theta} \sum_{i \in [n]} d_i^4.
\end{align*}
For $I.2$, by Cauchy-Schwarz inequality we have
\begin{align*}
    |I.2| &\lesssim \frac{1}{K_\theta} \sqrt{\sum_{i \in [n]} \epsilon_i^2 \left(  \sum_{j \neq i} P_{\theta,ij}^2 (\epsilon_j^2 - \sigma_j^2) \right)^2 } \sqrt{\sum_{i \in [n]} d_i^2} \\
    &= o_\mathbb{P}\left(\frac{\lVert {\bf X} \widehat{\Delta} \rVert_2}{\sqrt{n}}\right) = o_\mathbb{P}(1),
\end{align*}
where the first equality follows from Markov inequality and (\ref{off_diagonal_ineq})
\begin{align*}
    &\mathbb{E}_{\eps} \left[ \frac{1}{K_\theta^2} \sum_{i \in [n]}  \epsilon_i^2 \left( \sum_{j \neq i} P_{\theta,ij}^2 (\epsilon_j^2 -\sigma_j^2) \right)^2  \right] \\
    &= \frac{1}{K_\theta^2} \sum_{i \in [n]} \sum_{j \neq i} P_{\theta,ij}^4\mathbb{E}[(\epsilon_j^2 - \sigma_j^2)^2]  \sigma_i^2  \\
    & \lesssim \frac{p_n^2}{K_\theta^2} \sum_{i \in [n]} \sum_{j \neq i} P_{\theta,ij}^2 \\
    &\leq \frac{p_n^2}{K_\theta} = O(\frac{p_n}{n}) =  o(\frac{1}{n}),
\end{align*}
since $p_n = O(K_\theta/n) = O(r/n) = o(1)$. Using the same argument, we also have
\begin{align*}
    \frac{2}{K_\theta} \sum_{i \in [n]} \sum_{j \neq i} P_{\theta,ij}^2 \sigma_i^2 V_j = o_\Prob(1),
\end{align*}
whence $I = o_\Prob(1)$. Since $I = II$, we have $II = o_\Prob(1)$. Finally, it remains to show that $III = o_\Prob(1)$. To this end, express
\begin{align*}
    III &= \frac{2}{K_\theta} \sum_{i \in [n]} \sum_{j \neq i} P_{\theta,ij}^2 d_i^2 d_j^2 -  \frac{8}{K_\theta} \sum_{i \in [n]} \sum_{j \neq i} P_{\theta,ij}^2 d_i^2 \epsilon_j d_j + \frac{8}{K_\theta} \sum_{i \in [n]} \sum_{j \neq i} P_{\theta,ij}^2 \epsilon_i d_i \epsilon_j d_j \\
    &\equiv III.1 - III.2 + III.3.
\end{align*}
For $III.1$, by (\ref{off_diagonal_ineq}) we have
\begin{align*}
    III.1 \lesssim \frac{p_n}{K_\theta} \max_{i \in [n]}d_i^2 \sum_{i \in [n]} d_i^2  =  O \left(\frac{\lVert {\bf X} \widehat{\Delta} \rVert_{\infty}^2 \times \lVert {\bf X} \widehat{\Delta} \rVert_2^2}{n} \right) = o_\Prob(1).
\end{align*}
For $III.2$, by Cauchy-Schwarz inequality we have
\begin{align*}
    |III.2| &\lesssim \sqrt{\frac{1}{K_\theta} \sum_{i \in [n]} \sum_{j \neq i} P_{\theta,ij}^2 d_i^2 d_j^2} \times \sqrt{\frac{1}{K_\theta} \sum_{i \in [n]} \sum_{j \neq i} P_{\theta,ij}^2 d_i^2 \epsilon_j^2} = o_\Prob(1),
\end{align*}
where the equality follows from the proof of $III.1$ and 
\begin{align*}
    \frac{1}{K_\theta} \sum_{i \in [n]} \sum_{j \neq i} P_{\theta,ij}^2 d_i^2 \epsilon_j^2 &= \frac{1}{K_\theta} \sum_{i \in [n]} \sum_{j \neq i} P_{\theta,ij}^2 d_i^2 \sigma_j^2 + \frac{1}{K_\theta} \sum_{i \in [n]} \sum_{j \neq i} P_{\theta,ij}^2 d_i^2 (\epsilon_j^2 - \sigma_j^2) \\
    &= O \left(\frac{ \lVert {\bf X} \widehat{\Delta} \rVert_2^2}{n} \right) + o_\Prob(1) = o_\Prob(1),
\end{align*}
similar to the proof of $I.1$. Finally, by the Cauchy-Schwarz inequality,
\begin{align*}
    |III.3| &\lesssim \sqrt{\frac{1}{K_\theta} \sum_{i \in [n]} \sum_{j \neq i} P_{\theta,ij}^2 d_j^2 \epsilon_i^2} \times \sqrt{\frac{1}{K_\theta} \sum_{i \in [n]} \sum_{j \neq i} P_{\theta,ij}^2 d_i^2 \epsilon_j^2} = o_\Prob(1).
\end{align*}
This establishes the first part of the proof in Lemma \ref{lem:consistency:phi-hat}. The second part follows easily from
\begin{align*}
    \frac{\widehat\Phi}{\Phi} = \frac{\Phi + o_\Prob(1)}{\Phi} = 1 + o_\Prob(\Phi^{-1}) = 1 + o_\Prob(1).
\end{align*}

Finally, note that $\widehat \sigma_i^4 = (\eps_i - d_i)^4 \leq 8 (\eps_i^4 + d_i^4)$, so that
\begin{align*}
    \frac{1}{K_\theta} \sum_{i \in [n]} P_{\theta,ii}^2 \widehat \sigma_i^4  \lesssim \frac{1}{K_\theta} \sum_{i \in [n]} P_{\theta,ii}^2 \eps_i^4 + \frac{1}{K_\theta} \sum_{i \in [n]} P_{\theta,ii}^2 d_i^4.
\end{align*}
We have
\begin{align*}
    \mathbb E \frac{1}{K_\theta} \sum_{i \in [n]} P_{\theta,ii}^2 \eps_i^4 \lesssim p_n = o(1),
\end{align*}
and
\begin{align*}
    \frac{1}{K_\theta} \sum_{i \in [n]} P_{\theta,ii}^2 d_i^4 \lesssim \frac{r \lVert {\bf X} \widehat{\Delta} \rVert_{\infty}^2 \times \lVert {\bf X} \widehat{\Delta} \rVert_2^2}{n^2} = o_\Prob(1),
\end{align*}
and thus
\begin{align*}
    \frac{1}{K_\theta} \sum_{i \in [n]} P_{\theta,ii}^2 \widehat \sigma_i^4 = o_\Prob(1),
\end{align*}
and the result for $\widetilde \Phi$ follows.

\qed

\begin{lem}
    Under the conditions of Lemma \ref{lem:variance-estimator-consistent}, and assume that $\lVert {\bf X} (\widehat{\beta} - \beta^0 ) \rVert_{\infty} \times \lVert {\bf X} (\widehat{\beta} - \beta^0 ) \rVert_2 = o_\Prob(\sqrt{n})$ and $r=o(n)$. Then we have
    \begin{align*}
        \max_{i \in [r_n]} w_{i,n}^2 = O(1/K_\theta^2), \quad \max_{i \in [r_n]} \left(\tilde w_{i,n} - w_{i,n} \right)^2 = o_\Prob(1/K_\theta).
    \end{align*}
    In particular, in the case where $K_\theta$ is bounded, we have $\left \Vert \tilde{w}_{n} - w_{n} \right\Vert_2 = o_\Prob(1)$; in the case where $K_\theta \rightarrow \infty$, we have $ \max_{i \in [r_n]} w_{i,n}^2 = o(1/K_\theta)$ and $ \max_{i \in [r_n]} \tilde {w}_{i,n}^2 = o_\Prob(1/K_\theta)$.
    \label{lem:consistency:weights}
\end{lem}

\noindent \textbf{Proof of Lemma \ref{lem:consistency:weights}:} Note that $\Omega$ shares the same non-zero eigenvalues with
\begin{align*}
    \mathcal{O} := \frac{P_\theta^{1/2} \Lambda P_\theta^{1/2}}{\sum_{i \in [n]} P_{\theta,ii} \sigma_i^2},
\end{align*}
and similarly, $\widehat \Omega$ shares the same non-zero eigenvalues with 
\begin{align*}
    \widehat{\mathcal{O}} := \frac{P_\theta^{1/2} \widehat \Lambda P_\theta^{1/2}}{\sum_{i \in [n]} P_{\theta,ii} \hat \sigma_i^2}.
\end{align*}
The first conclusion follows because $\lambda_{\max}(\mathcal{O}) = O(1/K_\theta)$. By an application of the Weyl's theorem, we have
\begin{align*}
    \max_{i \in [r_n]} \left(\tilde w_{i,n} - w_{i,n} \right)^2 \lesssim  \left\Vert \widehat{\mathcal{O}} - \mathcal{O} \right\Vert_{2}^2 \lesssim  \left\Vert \widehat{\mathcal{O}} - \widetilde{\mathcal{O}} \right\Vert_{2}^2 + \left\Vert \widetilde{\mathcal{O}} - \mathcal{O} \right\Vert_{2}^2,
\end{align*}
where $$\widetilde{\mathcal{O}} := \frac{P_\theta^{1/2} \Lambda P_\theta^{1/2}}{\sum_{i \in [n]} P_{\theta,ii} \hat \sigma_i^2}, $$ and note that for the second term
\begin{align*}
    \left\Vert \widetilde{\mathcal{O}} - \mathcal{O} \right\Vert_{2}^2 &= \left(\frac{1}{\sum_{i \in [n]} P_{\theta,ii} \hat \sigma_i^2} -  \frac{1}{\sum_{i \in [n]} P_{\theta,ii} \sigma_i^2} \right)^2 \left \Vert P_\theta^{1/2} \Lambda P_\theta^{1/2} \right\Vert_{2}^2 \\
    &\lesssim O \left(\frac{1}{K_\theta^2} \right) \times \left(\frac{1}{\frac{1}{K_\theta}\sum_{i \in [n]} P_{\theta,ii} \hat \sigma_i^2} -  \frac{1}{\frac{1}{K_\theta}\sum_{i \in [n]} P_{\theta,ii} \sigma_i^2} \right)^2 \\
    &= o_\Prob\left(\frac{1}{K_\theta^2}\right)
\end{align*}
by Lemma \ref{lem:variance-estimator-consistent}, and the fact that $\lambda_{\max} (\Lambda) = O(1)$ by Assumption \ref{ass:mean_zero}.3 and $\lambda_{\max}(P_\theta) = O(1)$ from the proof of Lemma \ref{lem:Ktheta_and_rn}. For the first term, note that 
\begin{align*}
    \left\Vert \widehat{\mathcal{O}} - \widetilde{\mathcal{O}} \right\Vert_{2}^2 &= \left(\frac{\sum_{i \in [n]} P_{\theta,ii} \sigma_i^2}{\sum_{i \in [n]} P_{\theta,ii} \hat \sigma_i^2} \right)^2 \left(\frac{1}{\sum_{i \in [n]} P_{\theta,ii} \sigma_i^2} \right)^2 \left\Vert P_\theta^{1/2} \widehat \Lambda P_\theta^{1/2} - P_\theta^{1/2} \Lambda P_\theta^{1/2} \right\Vert_{2}^2 \\
    &= \frac{1}{K_\theta^2} \left(\frac{\frac{1}{K_\theta}\sum_{i \in [n]} P_{\theta,ii} \sigma_i^2}{\frac{1}{K_\theta}\sum_{i \in [n]} P_{\theta,ii} \hat \sigma_i^2} \right)^2 \left(\frac{1}{\frac{1}{K_\theta}\sum_{i \in [n]} P_{\theta,ii} \sigma_i^2} \right)^2 \left\Vert P_\theta^{1/2} \widehat \Lambda P_\theta^{1/2} - P_\theta^{1/2} \Lambda P_\theta^{1/2} \right\Vert_{2}^2 \\
    &\lesssim O_\Prob\left(\frac{1}{K_\theta^2}\right) \left\Vert P_\theta^{1/2} \widehat \Lambda P_\theta^{1/2} - P_\theta^{1/2} \Lambda P_\theta^{1/2} \right\Vert_{2}^2,
\end{align*}
by Lemma \ref{lem:variance-estimator-consistent} and that
\begin{align*}
   \frac{1}{K_\theta}\sum_{i \in [n]} P_{\theta,ii} \sigma_i^2 \geq \frac{\underline{c} }{K_\theta}\sum_{i \in [n]} P_{\theta,ii} = \underline{c}>0
\end{align*}
by Assumption \ref{ass:mean_zero}.2. We also have
\begin{align*}
    \left\Vert P_\theta^{1/2} \widehat \Lambda P_\theta^{1/2} - P_\theta^{1/2} \Lambda P_\theta^{1/2} \right\Vert_{2}^2 &= \left\Vert P_\theta^{1/2} \left(\widehat \Lambda - \Lambda \right) P_\theta^{1/2} \right\Vert_{2}^2 \\
    &\leq \left\Vert P_\theta^{1/2} \left(\widehat \Lambda - \Lambda \right) P_\theta^{1/2} \right\Vert_{F}^2 \\
    &= \tr \left( P_\theta \left(\widehat \Lambda - \Lambda \right) P_\theta \left(\widehat \Lambda - \Lambda \right) \right) \\
    &= \sum_{i \in [n]} \sum_{j \in [n]} P_{\theta, ij}^2 \left(\hat \sigma_i^2 - \sigma_i^2 \right) \left(\hat \sigma_j^2 - \sigma_j^2 \right).
\end{align*}
Therefore, the second conclusion of Lemma \ref{lem:consistency:weights} follows if 
\begin{align*}
    \frac{1}{K_\theta} \sum_{i \in [n]} \sum_{j \in [n]} P_{\theta, ij}^2 \left(\hat \sigma_i^2 - \sigma_i^2 \right) \left(\hat \sigma_j^2 - \sigma_j^2 \right) = o_\Prob(1).
\end{align*}
Recall that $\widehat{\sigma}_i^2 = \epsilon_i^2 + V_i$, we have
\begin{align*}
    \frac{1}{K_\theta} \sum_{i \in [n]} \sum_{j \neq i} P_{\theta, ij}^2 \left(\hat \sigma_i^2 - \sigma_i^2 \right) \left(\hat \sigma_j^2 - \sigma_j^2 \right) = \frac{1}{K_\theta} \sum_{i \in [n]} \sum_{j \neq i} P_{\theta, ij}^2 \left(\eps_i^2 - \sigma_i^2 + V_i \right) \left(\eps_j^2 - \sigma_j^2 + V_j \right) = o_\Prob(1),
\end{align*}
as in the proof of Lemma \ref{lem:consistency:phi-hat}. It remains to show that
\begin{align*}
    \frac{1}{K_\theta} \sum_{i \in [n]} P_{\theta, ii}^2 (\hat \sigma_i^2 - \sigma_i^2)^2 = o_\Prob(1),
\end{align*}
and note that $p_n = O(K_\theta/n) = O(r/n) = o(1)$. We can write
\begin{align*}
    \frac{1}{K_\theta} \sum_{i \in [n]} P_{\theta, ii}^2 (\hat \sigma_i^2 - \sigma_i^2)^2 = \frac{1}{K_\theta} \sum_{i \in [n]} P_{\theta, ii}^2 (\eps_i^2 - \sigma_i^2 + V_i)^2 \lesssim \frac{1}{K_\theta} \sum_{i \in [n]} P_{\theta, ii}^2 (\eps_i^2 - \sigma_i^2)^2 + \frac{1}{K_\theta} \sum_{i \in [n]} P_{\theta, ii}^2 V_i^2,
\end{align*}
and since
\begin{align*}
    \mathbb E_\epsilon \left[ \frac{1}{K_\theta} \sum_{i \in [n]} P_{\theta, ii}^2 (\eps_i^2 - \sigma_i^2)^2 \right] \lesssim  \frac{p_n}{K_\theta} \sum_{i \in [n]} P_{\theta, ii} \mathbb E_\epsilon \left[ (\eps_i^2 - \sigma_i^2)^2 \right] = O(p_n) = o(1),
\end{align*}
by Markov inequality we have
\begin{align*}
    \frac{1}{K_\theta} \sum_{i \in [n]} P_{\theta, ii}^2 (\eps_i^2 - \sigma_i^2)^2 = o_\Prob(1).
\end{align*}
We also have
\begin{align*}
    \frac{1}{K_\theta} \sum_{i \in [n]} P_{\theta, ii}^2 V_i^2 \lesssim \frac{1}{K_\theta} \sum_{i \in [n]} P_{\theta, ii}^2 d_i^4 + \frac{1}{K_\theta} \sum_{i \in [n]} P_{\theta, ii}^2 d_i^2 \eps_i^2,
\end{align*}
where
\begin{align*}
    \frac{1}{K_\theta} \sum_{i \in [n]} P_{\theta, ii}^2 d_i^4 \leq \frac{p_n^2}{K_\theta} \sum_{i \in [n]} d_i^4  = o\left(\frac{\left\Vert {\bf{X}} \widehat \Delta \right\Vert_{\infty}^2 \times \left\Vert {\bf{X}} \widehat \Delta \right\Vert_2^2}{n}\right) = o(1),
\end{align*}
and by the Cauchy-Schwarz inequality,
\begin{align*}
    \frac{1}{K_\theta} \sum_{i \in [n]} P_{\theta, ii}^2 d_i^2 \eps_i^2 \leq \sqrt{\frac{1}{K_\theta} \sum_{i \in [n]} P_{\theta, ii}^2 d_i^4} \sqrt{\frac{1}{K_\theta} \sum_{i \in [n]} P_{\theta, ii}^2 \eps_i^4} = o_\Prob(1),
\end{align*}
where we use the Markov inequality and that
\begin{align*}
    \mathbb E_\epsilon \left[ \frac{1}{K_\theta} \sum_{i \in [n]} P_{\theta, ii}^2 \eps_i^4 \right] \lesssim \frac{p_n}{K_\theta} \sum_{i \in [n]} P_{\theta, ii} \mathbb E_\epsilon\left(\eps_i^4 \right) = O(p_n) = o(1).
\end{align*}
Finally, note that $\left \Vert \tilde{w}_{n} - w_{n} \right\Vert_1 \leq 2$ because both $\{w_i\}_{i \in [r_n]}$ and $\{\tilde w_i\}_{i \in [r_n]}$ are nonnegative and sum to one, and thus $\left \Vert \tilde{w}_{n} - w_{n} \right\Vert_2^2 \leq \left \Vert \tilde{w}_{n} - w_{n} \right\Vert_1 \left \Vert \tilde{w}_{n} - w_{n} \right\Vert_\infty = o_\Prob(1)$. This concludes the proof.
\qed

\subsection{Random Design}

We now derive the same results under random $\bf X$. The following lemma is useful for the proof.

\begin{lem} \label{lem:random_design_bound}
    Let Assumption \ref{ass:size} hold. For a suitable choice of $\lambda_0 \asymp \sqrt{\log(np)/n}$, we have
    \begin{align*}
        &\left\lVert {\bf X} (\widehat\beta - \beta^0) / \sqrt{n} \right\rVert_{2} = O_\Prob\left( \sqrt{\frac{s_0 \log (np)}{n}}\right), \\
        &\lVert \widehat\beta-\beta^0 \rVert_1 = O_\Prob \left(\sqrt{\frac{s_0^2 \log (np)}{n}}\right), \\
        &\lVert {\bf X} (\widehat\beta - \beta^0) \rVert_{\infty} = O_\Prob\left( \sqrt{\frac{s_0^2 (\log (np))^{1+2/\nu}}{n}}\right).
    \end{align*}
    In addition, for a suitable choice of $\lambda_j \asymp \sqrt{\log(np)/n}$ uniformly in $j \in [p]$, we have
    \begin{align*}
    &\max_{j \in [p]}\left\Vert \widehat \Theta_j^\top - \Theta_j^\top \right\Vert_1 = O_\mathbb{P} \left( s_0^*\sqrt{\frac{\log (np)}{n}} + \frac{(s_0^*)^{3/2}(\log(np))^{4/\nu}}{n} \right), \\
    &\max_{j \in [p]} \left\Vert \widehat \Theta_j^\top - \Theta_j^\top \right\Vert_2 = O_\mathbb{P} \left( \sqrt{\frac{s_0^* \log (np)}{n}} + \frac{s_0^* (\log(np))^{4/\nu}}{n} \right), \\
    &\max_{j \in [p]}\left| \hat \tau_j^2 - \tau_j^2 \right| = O_\mathbb{P} \left( \sqrt{\frac{s_0^* \log (np)}{n}} + \frac{s_0^* (\log(np))^{4/\nu}}{n}\right).
    \end{align*}
    Finally, for the choice of $\lambda_0$ and $\lambda_j$ above, we have
    \begin{align*}
        \left\Vert \Delta \right\Vert_{\infty} = O_\Prob \left( \frac{s_0 \log(np)}{\sqrt{n}} \right).
    \end{align*}
\end{lem}

\noindent \textbf{Proof of Lemma \ref{lem:random_design_bound}:} The first result follows from Lemma \ref{lem:bound_of_Xeps} (see Remark \ref{rem:LASSO_tuning_rate}) and Lemma \ref{lem:random_design_X}. The second result follows from Lemma \ref{lem:hattau}, Lemma \ref{lem:hatTheta} and Lemma \ref{lem:node-wise_tuning_rate}. For the last result, note that 
\begin{align*}
    \left\Vert \Delta \right\Vert_{\infty} = \sqrt{n} \left\Vert (\widehat \Theta \widehat \Sigma - I_p) (\widehat\beta-\beta^0)\right\Vert_{\infty} \leq \sqrt{n} \left\Vert \widehat \Theta \widehat \Sigma - I_p \right\Vert_{\max} \lVert \widehat\beta-\beta^0 \rVert_1.
\end{align*}
By the extended KKT conditions \eqref{node-wise_LASSO_kkt}, we have
\begin{align*}
    \left\Vert \widehat \Theta \widehat \Sigma - I_p \right \Vert_{\max} \leq \max_j \frac{\lambda_j}{\hat \tau_j^2}, 
\end{align*}
and by the choice of $\lambda_j$ we have
\begin{align*}
    \left\Vert \widehat \Theta \widehat \Sigma - I_p \right \Vert_{\max} \lesssim \sqrt{\frac{\log (np)}{n}} \max_j \frac{1}{\hat \tau_j^2-\tau_j^2 + \tau_j^2} = O_\mathbb{P} \left(\sqrt{\frac{\log (np)}{n}}\right),
\end{align*}
by using the second result, where we also use Lemma \ref{lem:aux:2}. Combining this with the first result concludes the proof. 
\qed

\begin{lem} \label{lem:variance_consistency_random_design}
    Let Assumption \ref{ass:size} hold, and suppose that $\lambda_0 \asymp \sqrt{\log(np)/n}$ and $\lambda_j \asymp \sqrt{\log(np)/n}$ (uniformly in $j \in [p]$) are suitably chosen. Then
    \begin{align*}
        \frac{\sum_{i \in [n]} P_{\theta,ii} \widehat{\sigma}_i^2}{K_\theta} = \frac{\sum_{i \in [n]} P_{\theta,ii} \epsilon_i^2}{K_\theta} + o_\Prob(1),
    \end{align*}
    and
    \begin{align*}
        \frac{\sum_{i \in [n]} P_{\theta,ii} \epsilon_i^2}{K_\theta} = \frac{\sum_{i \in [n]} P_{\theta,ii} \sigma_i^2}{K_\theta} + o_\Prob(1).
    \end{align*}
    In particular, we have
    \begin{align*}
        \frac{\sum_{i \in [n]} P_{\theta,ii} \widehat{\sigma}_i^2}{K_\theta}  = \frac{\sum_{i \in [n]} P_{\theta,ii} \sigma_i^2}{K_\theta} + o_\Prob(1).
    \end{align*}
    In addition, we have
    \begin{align*}
        \frac{\widehat\Phi}{\Phi} = 1 + o_\Prob(1)
    \end{align*}
    if $1/\Phi = O_\Prob(1)$, and the same results hold for the adjusted variance estimator $\widetilde \Phi$.
\end{lem}

\noindent{\bf Proof of Lemma \ref{lem:variance_consistency_random_design}:} We have $p_n = \max_{i \in [n]} P_{\theta,ii} =  O_\Prob(K_\theta/n)$ by assumption. We can also verify the rate conditions in Lemma \ref{lem:variance-estimator-consistent} and Lemma \ref{lem:consistency:phi-hat}. By Lemma \ref{lem:random_design_bound}, we have
\begin{align*}
    \frac{\lVert {\bf X} \widehat \Delta \rVert_2}{\sqrt{n}} = O_\Prob\left( \sqrt{\frac{s_0 \log (np)}{n}}\right) = o_\Prob(1),
\end{align*}
and
\begin{align*}
    \frac{\lVert {\bf X} \widehat \Delta  \rVert_{\infty} \times \lVert {\bf X} \widehat \Delta  \rVert_2}{\sqrt{n}} &= O_\Prob\left( \sqrt{\frac{s_0 \log (np)}{n}}\right) \times O_\Prob\left( \sqrt{\frac{s_0^2 (\log (np))^{1+2/\nu}}{n}}\right) \\
    &= O_\Prob\left(\sqrt{\frac{s_0^2 (\log (np))^{2}}{n}} \times \sqrt{\frac{s_0 (\log (np))^{2/\nu}}{n}} \right)\\
    &= o_\Prob(1).
\end{align*}
Note also that $r\leq m = o(\sqrt{n/\log(np)}) = o(n)$. We may then repeat the proof of Lemma \ref{lem:variance-estimator-consistent} and Lemma \ref{lem:consistency:phi-hat}. The only difference here is that $\mathbb E_\epsilon (\cdot)$ is now random as $\bf X$ is random. However, note that by Markov inequality and an application of \cite[Lemma 6.1]{chernozhukov2018double}, we have $\mathbb E_\epsilon \left| A_n \right| = O_\Prob(l_n)$ implies that $ A_n = O_\Prob(l_n)$, where $A_n$ is a sequence of random variables and $l_n$ is a sequence of positive constants. In particular, $l_n = O(1)$ implies that $A_n = O_\Prob(1)$ and $l_n = o(1)$ implies that $A_n = o_\Prob(1)$. The proof then follows essentially the same lines as in the fixed design case.
\qed

\begin{lem} \label{lem:weights_consistency_random_design}
    Let Assumption \ref{ass:size} hold, and suppose that $\lambda_0 \asymp \sqrt{\log(np)/n}$ and $\lambda_j \asymp \sqrt{\log(np)/n}$ (uniformly in $j \in [p]$) are suitably chosen. Then
    \begin{align*}
        \max_{i \in [r_n]} w_{i,n}^2 = O_\Prob(1/K_\theta^2), \quad \max_{i \in [r_n]} \left(\tilde w_{i,n} - w_{i,n} \right)^2 = o_\Prob(1/K_\theta).
    \end{align*}
    In particular, in the case where $K_\theta$ is bounded (with probability approaching one), we have $\left \Vert \tilde{w}_{n} - w_{n} \right\Vert_2 = o_\Prob(1)$; in the case where $K_\theta \xrightarrow{\Prob} \infty$, we have $ \max_{i \in [r_n]} w_{i,n}^2 = o_\Prob(1/K_\theta)$ and $ \max_{i \in [r_n]} \tilde {w}_{i,n}^2 = o_\Prob(1/K_\theta)$.
\end{lem}

\noindent{\bf Proof of Lemma \ref{lem:weights_consistency_random_design}:} Using the same argument as in the proof of Lemma \ref{lem:variance_consistency_random_design}, we can repeat the proof of Lemma \ref{lem:consistency:weights} to show Lemma \ref{lem:weights_consistency_random_design}.
\qed

\newpage
\section{Proofs of Main Results} \label{sec:appendix_main_thm}

In this section, we prove the main results in Section \ref{section:strong_approx}-\ref{section:size_control}.

\subsection{Proof of Theorem \ref{theorem:strong_approx}}

The proof of \eqref{eq:strong-approximation:1} follows from a minor modification to the proof of \cite[Theorem 3.1, {\bf Step 2.}]{LWZ2024}. Note that
\begin{align*}
&\quad \sup_{y \in \mathbb{R}}\left| \mathbb{P}_\epsilon \left( \sum_{i,j \in [n], j \neq i}  P_{ij}({\bf X}) \epsilon_i \epsilon_j  \leq y \right) - \mathbb{P}_\mathcal{E} \left( \sum_{i,j \in [n], j \neq i} P_{ij}({\bf X}) \mathcal{E}_i \mathcal{E}_j  \leq y \right) \right| \\
&= \sup_{y \in \mathbb{R}}\left| \mathbb{P}_\epsilon \left( \frac{\sum_{i,j \in [n], j \neq i}  P_{ij}({\bf X}) \epsilon_i \epsilon_j}{\sqrt{K_n}}  \leq y \right) - \mathbb{P}_\mathcal{E} \left( \frac{\sum_{i,j \in [n], j \neq i} P_{ij}({\bf X}) \mathcal{E}_i \mathcal{E}_j}{\sqrt{K_n}}  \leq y \right) \right|,
\end{align*}
and define 
\begin{align*}
    \breve{\mathcal Q} := \frac{\sum_{i,j \in [n], j \neq i}  P_{ij}({\bf X}) \epsilon_i \epsilon_j}{\sqrt{K_n}}, \quad \mathcal Q := \frac{\sum_{i,j \in [n], j \neq i} P_{ij}({\bf X}) \mathcal{E}_i \mathcal{E}_j}{\sqrt{K_n}}.
\end{align*}
Let $\hat p_n := \tilde p_n/K_n$ and note that $0 \leq \hat p_n \leq 1$. For any $\zeta \in (0,1)$, let $\tilde \zeta = \zeta/2$, the proof in \cite[Theorem 3.1, {\bf Step 2.}]{LWZ2024} implies that
\begin{align*}
    \sup_{y \in \mathbb{R}}\left| \mathbb{P}_{\eps} \left( \breve{\mathcal Q}  \leq y \right) - \mathbb{P}_{\mathcal{E}} \left( \mathcal Q  \leq y \right) \right| \lesssim_{\zeta} \frac{1}{1-B(C_n)} \times \frac{\hat p_n^{1/2}}{C_n h_n^3} + \frac{B(C_n)}{1-B(C_n)} + (C_n h_n)^{\frac{1-\tilde \zeta}{2}},
\end{align*}
where $\lesssim_{\zeta}$ means that the implicit constant depends on $\zeta$ (but not on $n$). Here we replace their $C_h$ with $C_n$ as we now need it to change with $n$, and $B(C_n) = \left(C_n^2/\exp(1-C_n^2)\right)^{1/2}$. We may assume without loss of generality that $\hat p_n \rightarrow 0$, as the bound can be made trivial by using a large $C_\zeta$ when $\hat p_n$ is not small. We can then optimize the right-hand side in terms of $\hat p_n$. Note that $1/(1-B(C_n)) = O(1)$ as $B(C_n)\rightarrow 0$ (i.e. $C_n \rightarrow \infty$). By choosing $$h_n = O\left( \hat p_n^{\frac{1}{7-\tilde \zeta}} C_n^{-\frac{3-\tilde \zeta}{7-\tilde \zeta}} \right),$$ the first term and the third term are both of order $$\hat p_n^{\frac{1-\tilde \zeta}{2(7-\tilde \zeta)}} C_n^{\frac{2(1-\tilde \zeta)}{7-\tilde \zeta}}.$$ By choosing $$C_n = \left(2 \left( \frac{1-\tilde \zeta}{2(7-\tilde \zeta)}+ \tilde \zeta \right) \log(1/\hat p_n)\right)^{1/2},$$ we have
\begin{align*}
    \hat p_n^{\frac{1-\tilde \zeta}{2(7-\tilde \zeta)}} C_n^{\frac{2(1-\tilde \zeta)}{7-\tilde \zeta}} \asymp_{\zeta} \hat p_n^{\frac{1-\tilde \zeta}{2(7-\tilde \zeta)}} (\log(1/\hat p_n))^{\frac{1-\tilde \zeta}{7-\tilde \zeta}},
\end{align*}
and
\begin{align*}
    B(C_n) \lesssim_{\zeta} C_n \times \exp(-C_n^2/2) \lesssim_{\zeta} \log^{1/2}(1/\hat p_n) \times \hat p_n^{\frac{1-\tilde \zeta}{2(7-\tilde \zeta)}+ \tilde \zeta} = o\left( \hat p_n^{\frac{1-\tilde \zeta}{2(7-\tilde \zeta)}} (\log(1/\hat p_n))^{\frac{1-\tilde \zeta}{7-\tilde \zeta}} \right),
\end{align*}
and thus the second term is negligible. It follows that
\begin{align*}
    \sup_{y \in \mathbb{R}}\left| \mathbb{P}_{\eps} \left( \breve{\mathcal Q}  \leq y \right) - \mathbb{P}_{\mathcal{E}} \left( \mathcal Q  \leq y \right) \right| \lesssim_{\zeta} \left( \hat p_n \log^2(1/\hat p_n) \right)^{\frac{1-\tilde \zeta}{2(7-\tilde \zeta)}} \lesssim_{\zeta} \hat p_n^{\frac{1-\zeta}{2(7-\zeta)}},
\end{align*}
because $\tilde \zeta < \zeta$ and 
\begin{align*}
    \hat p_n^{\left(\frac{1-\tilde \zeta}{2(7-\tilde \zeta)}-\frac{1-\zeta}{2(7-\zeta)}\right)}  (\log(1/\hat p_n))^{\frac{1-\tilde \zeta}{7-\tilde \zeta}} \rightarrow 0
\end{align*}
as $\hat p_n \rightarrow 0$.
This gives the first part of Theorem \ref{theorem:strong_approx}.

\vspace{5mm}
\noindent For the second part of Theorem \ref{theorem:strong_approx}, we first derive an anti-concentration inequality that extends \cite[Lemma H.1.]{LWZ2024}'s Gaussian errors to arbitrary errors, which may be of independent interest.
\begin{lem}[Anti-Concentration]
    Suppose Assumption \ref{ass:mean_zero} holds. In addition, assume that $$\frac{1}{K_n} \sum_{i \in [n]} \sum_{j \neq i} P_{ij}^2({\bf X}) \gtrsim 1.$$ Then for any $\ell>0$ and $\zeta \in (0,1)$, there exists a $C_\zeta$ only depending on $\zeta$ such that 
    \begin{align*}
        \sup_{y \in \mathbb{R}} \mathbb{P}_\epsilon\left( \left|\frac{\sum_{i \in [n]} \sum_{j \neq i} P_{ij}({\bf X}) \epsilon_i \epsilon_j}{\sqrt{K_n}} - y\right| \leq \ell  \right) \leq C_\zeta \ell^{(1-\zeta)/2} + C_\zeta \left(\frac{\tilde p_n}{K_n} \right)^{(1-\zeta)/(14-2\zeta)}.
    \end{align*}
    \label{lem:anti-concentration:main}
\end{lem}

\noindent{\bf Proof of Lemma \ref{lem:anti-concentration:main}:}
We first require an additional lemma\footnote{This Lemma only requires minimal modification, hence we state it without proof.}.
\begin{lem}[\cite{LWZ2024}, Lemma H.1.]
    Suppose Assumption \ref{ass:mean_zero} holds. In addition, assume that $$\frac{1}{K_n} \sum_{i \in [n]} \sum_{j \neq i} P_{ij}^2({\bf X}) \gtrsim 1.$$ Then for any $\ell >0$ and $\zeta \in (0,1)$, there exists a $C_\zeta$ only depending on $\zeta$ such that 
    \begin{align*}
        \sup_{y \in \mathbb{R}} \mathbb{P}_\mathcal{E}\left( \left|\frac{\sum_{i \in [n]} \sum_{j \neq i} P_{ij}({\bf X}) \mathcal{E}_i \mathcal{E}_j}{\sqrt{K_n}} - y\right| \leq \ell  \right) \leq C_\zeta \ell^{(1-\zeta)/2},
    \end{align*}
    where $\mathcal{E}_i \overset{d}{=} \mathcal{N}\left(0, \sigma_i^2 \right)$ are independent.
    \label{lem:LWZ:Lemma_H.1.}
\end{lem}
\noindent Define
\begin{align*}
    \varsigma_n := \sup_{y \in \mathbb{R}}\left| \mathbb{P}_{\epsilon} \left( \frac{1}{\sqrt{K_n}} \sum_{i \in [n]} \sum_{j \neq i}  P_{ij}({\bf X}) \epsilon_i \epsilon_j \leq y \right) - \mathbb{P}_{\mathcal{E}} \left( \frac{1}{\sqrt{K_n}} \sum_{i \in [n]} \sum_{j \neq i} P_{ij}({\bf X}) \mathcal{E}_i \mathcal{E}_j \leq y \right) \right|.
\end{align*}
Then observe that for any $y \in \mathbb{R}$ and $\ell > 0$, we have
\begin{align*}
    &\mathbb{P}_\epsilon\left( \left|\frac{\sum_{i \in [n]} \sum_{j \neq i} P_{ij}({\bf X}) \epsilon_i \epsilon_j}{\sqrt{K_n}} - y\right| \leq \ell  \right) \\
    &\leq \mathbb{P}_\epsilon\left( \frac{\sum_{i \in [n]} \sum_{j \neq i} P_{ij}({\bf X}) \epsilon_i \epsilon_j}{\sqrt{K_n}} \leq y+ 2\ell \right) - \mathbb{P}_\epsilon\left( \frac{\sum_{i \in [n]} \sum_{j \neq i} P_{ij}({\bf X}) \epsilon_i \epsilon_j}{\sqrt{K_n}} \leq y- 2\ell \right) \\
    &\leq \mathbb{P}_\mathcal{E}\left( \frac{\sum_{i \in [n]} \sum_{j \neq i} P_{ij}({\bf X}) \mathcal{E}_i \mathcal{E}_j}{\sqrt{K_n}} \leq y+ 2\ell \right) - \mathbb{P}_\mathcal{E}\left( \frac{\sum_{i \in [n]} \sum_{j \neq i} P_{ij}({\bf X}) \mathcal{E}_i \mathcal{E}_j}{\sqrt{K_n}} \leq y- 2\ell \right)  + 2 \varsigma_n  \\
    &\leq C_\zeta (2\ell)^{(1-\zeta)/2} + 2 \varsigma_n,
\end{align*}
where we used Lemma \ref{lem:LWZ:Lemma_H.1.} in the last inequality. Finally, we apply \eqref{eq:strong-approximation:1} to complete the proof.
\qed

\vspace{5mm}
Define
\begin{align*}
    A_{1,n} := \frac{\sum_{i \in [n]} P_{ii}({\bf X})\epsilon_i^2}{\sqrt{K_n}} - B_n, \quad A_{2,n} := \frac{\sum_{i \in [n]} P_{ii}({\bf X})\mathcal{E}_i^2}{\sqrt{K_n}} - B_n, \quad B_n := \frac{\sum_{i \in [n]} P_{ii}({\bf X})\sigma_i^2}{\sqrt{K_n}},
\end{align*}
where we note that we have
\begin{align*}
    \mathbb{E}_{\eps} \left( \frac{\sum_{i \in [n]} P_{ii}({\bf X}) (\epsilon_i^2 - \sigma_i^2)}{\sqrt{K_n}} \right)^2  \lesssim \frac{\sum_{i \in [n]} P_{ii}({\bf X})^2}{K_n},
\end{align*}
and the same result holds if we replace $\eps$ with $\mathcal E$, whence
\begin{align}
    \mathbb{E}_\epsilon [A_{1,n}^2] \lesssim \tilde p_n, \quad \mathbb{E}_\epsilon [A_{2,n}^2] \lesssim \tilde p_n.
    \label{eq:Theorem_3.1:1}
\end{align}
To control $A_{1,n}$ and $A_{2,n}$, we use the following lemma.

\begin{lem}[Kolmogorov Distance Inequality]\label{lem:anti-concentration}
    Let $X_n$ and $Y_n$ be random variables. Then for any $\ell >0$,
    \begin{align*}
        \sup_{y \in \mathbb{R}} \left| \mathbb{P}(X_n+Y_n \leq y) - \mathbb{P}(X_n \leq y)  \right| \leq \sup_{y \in \mathbb{R}}\mathbb{P}(|X_n - y|\leq \ell) + \mathbb{P}(|Y_n|>\ell).
    \end{align*}
\end{lem}
\noindent \textbf{Proof of Lemma \ref{lem:anti-concentration}:} For any $y \in \mathbb{R}$ and $\ell>0$, note that
\begin{align*}
    \left\{ X_n + Y_n > y  \right\} \supseteq \left\{  X_n > y \right\} \cap \left\{ |X_n-y| > \ell \right\} \cap \left\{ |Y_n| \leq \ell   \right\}
\end{align*}
implies 
\begin{align*}
    \left\{ X_n + Y_n \leq y  \right\} \subseteq \left\{  X_n \leq y \right\} \cup \left\{ |X_n-y| \leq \ell \right\} \cup \left\{ |Y_n| > \ell  \right\}.
\end{align*}
Therefore,
\begin{align*}
    \mathbb{P}\left( X_n + Y_n \leq y \right) - \mathbb{P}\left( X_n \leq y  \right) \leq \mathbb{P}(|X_n - y|\leq \ell) + \mathbb{P}(|Y_n| > \ell).
\end{align*}
Similarly, we show that
\begin{align*}
    \left\{ X_n \leq y  \right\} \subseteq \left\{  X_n + Y_n \leq y \right\} \cup \left\{ |X_n-y| \leq \ell \right\} \cup \left\{ |Y_n| > \ell   \right\}.
\end{align*}
So
\begin{align*}
    \left|\mathbb{P}\left( X_n + Y_n \leq y \right) - \mathbb{P}\left( X_n \leq y  \right)\right| &\leq \mathbb{P}(|X_n - y|\leq \ell) + \mathbb{P}(|Y_n| > \ell).
\end{align*}
Taking supremum over $y \in \mathbb{R}$ yields the result.
\qed

\vspace{5mm}
For notational simplicity, we simply write $\mathbb{P}$ instead of $\mathbb{P}_\epsilon$ and $\mathbb{P}_\mathcal{E}$, and note that we take $P({\bf X})$ as fixed here. Then observe that for any $\ell>0$,
\begin{align*}
    \sup_{y \in \mathbb{R}} \left| \mathbb{P} \left( \epsilon^\top P({\bf X}) \epsilon \leq y \right) - \mathbb{P} \left( \mathcal{E}^\top P({\bf X}) \mathcal{E} \leq y \right)  \right| &= \sup_{y \in \mathbb{R}} \left| \mathbb{P} \left( \epsilon^\top P({\bf X}) \epsilon \leq \sqrt{K_n} y \right) - \mathbb{P} \left( \mathcal{E}^\top P({\bf X}) \mathcal{E} \leq \sqrt{K_n} y \right)  \right| \\
    & \leq C_\zeta \ell^{(1-\zeta)/2} + C_\zeta \left( \frac{\tilde p_n}{K_n} \right)^{(1-\zeta)/(14-2\zeta)} + \overline{C} \frac{\tilde p_n}{\ell^2},
\end{align*}
where the inequality follows from observing that $\forall \ell > 0$,
\begin{align*}
    &\sup_{y \in \mathbb{R}} \left| \mathbb{P} \left( \epsilon^\top P({\bf X}) \epsilon \leq \sqrt{K_n} y \right) - \mathbb{P} \left( \mathcal{E}^\top P({\bf X}) \mathcal{E} \leq \sqrt{K_n} y \right)  \right| \\
    & = \sup_{y \in \mathbb{R}}  \Bigg| \mathbb{P} \left(\frac{\sum_{i \in [n]} \sum_{j \neq i} P_{ij}({\bf X}) \epsilon_i \epsilon_j}{\sqrt{K_n}} +A_{1,n} \leq y -B_n\right)  -  \mathbb{P} \left( \frac{\sum_{i \in [n]} \sum_{j \neq i} P_{ij}({\bf X}) \mathcal{E}_i \mathcal{E}_j}{\sqrt{K_n}}+A_{2,n} \leq y -B_n \right)  \Bigg| \\
    & = \sup_{y \in \mathbb{R}}  \Bigg| \mathbb{P} \left(\frac{\sum_{i \in [n]} \sum_{j \neq i} P_{ij}({\bf X}) \epsilon_i \epsilon_j}{\sqrt{K_n}} +A_{1,n} \leq y \right)  -  \mathbb{P} \left( \frac{\sum_{i \in [n]} \sum_{j \neq i} P_{ij}({\bf X}) \mathcal{E}_i \mathcal{E}_j}{\sqrt{K_n}}+A_{2,n} \leq y\right)  \Bigg| \\
    & \leq \sup_{y \in \mathbb{R}}  \Bigg| \mathbb{P} \left(\frac{\sum_{i \in [n]} \sum_{j \neq i} P_{ij}({\bf X}) \epsilon_i \epsilon_j}{\sqrt{K_n}} +A_{1,n} \leq y \right)  -  \mathbb{P} \left( \frac{\sum_{i \in [n]} \sum_{j \neq i} P_{ij}({\bf X}) \epsilon_i \epsilon_j}{\sqrt{K_n}} \leq y\right)  \Bigg| \\
    & \quad + \sup_{y \in \mathbb{R}}  \Bigg| \mathbb{P} \left(\frac{\sum_{i \in [n]} \sum_{j \neq i} P_{ij}({\bf X}) \epsilon_i \epsilon_j}{\sqrt{K_n}} \leq y \right)  -  \mathbb{P} \left( \frac{\sum_{i \in [n]} \sum_{j \neq i} P_{ij}({\bf X}) \mathcal{E}_i \mathcal{E}_j}{\sqrt{K_n}} \leq y\right)  \Bigg| \\
    & \quad + \sup_{y \in \mathbb{R}}  \Bigg| \mathbb{P} \left(\frac{\sum_{i \in [n]} \sum_{j \neq i} P_{ij}({\bf X}) \mathcal{E}_i \mathcal{E}_j}{\sqrt{K_n}} +A_{2,n} \leq y \right)  -  \mathbb{P} \left( \frac{\sum_{i \in [n]} \sum_{j \neq i} P_{ij}({\bf X}) \mathcal{E}_i \mathcal{E}_j}{\sqrt{K_n}} \leq y\right)  \Bigg| \\
    & \leq \sup_{y \in \mathbb{R}}  \Bigg| \mathbb{P} \left(\frac{\sum_{i \in [n]} \sum_{j \neq i} P_{ij}({\bf X}) \epsilon_i \epsilon_j}{\sqrt{K_n}} \leq y \right)  -  \mathbb{P} \left( \frac{\sum_{i \in [n]} \sum_{j \neq i} P_{ij}({\bf X}) \mathcal{E}_i \mathcal{E}_j}{\sqrt{K_n}} \leq y\right)  \Bigg| \\
    &\quad + \sup_{y \in \mathbb{R}} \mathbb{P} \left(\left|\frac{\sum_{i \in [n]} \sum_{j \neq i} P_{ij}({\bf X}) \epsilon_i \epsilon_j}{\sqrt{K_n}} - y \right| \leq \ell \right) + \sup_{y \in \mathbb{R}} \mathbb{P} \left(\left|\frac{\sum_{i \in [n]} \sum_{j \neq i} P_{ij}({\bf X}) \mathcal{E}_i \mathcal{E}_j}{\sqrt{K_n}} - y \right| \leq \ell \right) \\
    &\quad + \mathbb{P}(|A_{1,n}|>\ell) + \mathbb{P}(|A_{2,n}|>\ell) \\
    & \leq \sup_{y \in \mathbb{R}}  \Bigg| \mathbb{P} \left(\frac{\sum_{i \in [n]} \sum_{j \neq i} P_{ij}({\bf X}) \epsilon_i \epsilon_j}{\sqrt{K_n}} \leq y \right)  -  \mathbb{P} \left( \frac{\sum_{i \in [n]} \sum_{j \neq i} P_{ij}({\bf X}) \mathcal{E}_i \mathcal{E}_j}{\sqrt{K_n}} \leq y\right)  \Bigg| \\
    & \quad + C_\zeta \ell^{(1-\zeta)/2} + C_\zeta \left( \frac{\tilde p_n}{K_n} \right)^{(1-\zeta)/(14-2\zeta)} +\mathbb{P}(|A_{1,n}|>\ell) + \mathbb{P}(|A_{2,n}|>\ell) \\
    & \leq C_\zeta \ell^{(1-\zeta)/2} + C_\zeta \left( \frac{\tilde p_n}{K_n} \right)^{(1-\zeta)/(14-2\zeta)} +\mathbb{P}(|A_{1,n}|>\ell) + \mathbb{P}(|A_{2,n}|>\ell) \\
    & \leq C_\zeta \ell^{(1-\zeta)/2} + C_\zeta \left( \frac{\tilde p_n}{K_n} \right)^{(1-\zeta)/(14-2\zeta)} + \overline{C} \frac{\tilde p_n}{\ell^2},
\end{align*}
where we used Lemma \ref{lem:anti-concentration} in the second inequality,  Lemma \ref{lem:anti-concentration:main} and Lemma \ref{lem:LWZ:Lemma_H.1.} in the third  inequality, \eqref{eq:strong-approximation:1} in the second-last inequality, and \eqref{eq:Theorem_3.1:1} in the last inequality. The proof is complete. 
\qed

\begin{rem} \label{rem:strong_approx_random_design}
    Under Assumption \ref{ass:size}, by setting $\ell = (r/n)^{1/4}$ and noting that $r \leq m = o(n)$, we have
    \begin{align*}
        \sup_{y \in \mathbb{R}} \left| \mathbb{P}_\epsilon \left( \epsilon^\top P_\theta \epsilon \leq y \mid \bf X \right) 
        - \mathbb{P}_\mathcal{E} \left( \mathcal{E}^\top P_\theta \mathcal{E} \leq y \mid \bf X \right) \right| = o_\Prob(1).
    \end{align*}
    It follows that 
    \begin{align*}
        &\quad \sup_{y \in \mathbb{R}} \left| \mathbb{P}_\epsilon \left( \epsilon^\top P_\theta \epsilon \leq y  \right) 
        - \mathbb{P}_\mathcal{E} \left( \mathcal{E}^\top P_\theta \mathcal{E} \leq y  \right) \right| \\
        &= \sup_{y \in \mathbb{R}} \left|\mathbb E \left( \mathbb{P}_\epsilon \left( \epsilon^\top P_\theta \epsilon \leq y \big| \bf X \right) 
        - \mathbb{P}_\mathcal{E} \left( \mathcal{E}^\top P_\theta \mathcal{E} \leq y \big| \bf X \right) \right) \right| \\
        &\leq \mathbb E \left( \sup_{y \in \mathbb{R}} \left| \mathbb{P}_\epsilon \left( \epsilon^\top P_\theta \epsilon \leq y \big| \bf X \right) 
        - \mathbb{P}_\mathcal{E} \left( \mathcal{E}^\top P_\theta \mathcal{E} \leq y \big| \bf X \right) \right| \right)\\
        &= o(1),
    \end{align*}
    and thus Corollary \ref{cor:strong_approx} holds for random $\bf X$ as well.
\end{rem}

\subsection{Proof of Theorem \ref{thm:size-control_fixed_design}}

Consider any sequence of DGPs $\lambda_{n_k} \in \Lambda_{n_k}$. We aim to show that
\begin{align}
    & \lim_{n_k \rightarrow \infty} \mathbb{P}_{\lambda_{n_k}} \left( \widehat{Q}(q) > \mathcal C_{\alpha} \right) = \alpha. \label{eq:assumption_B*} 
\end{align}
Any sub-sequence of $\lambda_{n_k}$ admits a further sub-sequence along which $K_\theta$ (or $r_n$, by \eqref{eq:K_lambda:summary}) is either bounded or diverging, and we will show that we have asymptotically exact size control in both cases. Then \eqref{eq:assumption_B*} satisfies \textbf{Assumption B*} of \cite{Andrews-cheng-guggenberger(2020)}. By \textbf{Corollary 2.1(c)} of their paper, Theorem \ref{thm:size-control_fixed_design} follows. Without loss of generality, we implicitly consider some arbitrary sequence $\lambda_n \in \Lambda_n$ and show that it satisfies \eqref{eq:assumption_B*}. We break the proof into two parts, part $I$ and $II$.

\vspace{5mm}
\noindent \begin{Large} \textbf{Part $I$ (Bounded $K_\theta$ case):} \end{Large}
Assume that $K_\theta$ is bounded. We break our proof into three steps.

\vspace{5mm}
\noindent {\it Step 1:} We show that in any case (bounded or diverging $K_\theta$), 
\begin{align*}
    \widehat{Q}(q) = \frac{\epsilon^\top P_\theta \epsilon}{\sum_{i \in [n]} P_{\theta,ii} \sigma_i^2} + o_\Prob(K_\theta^{-1/2}).
\end{align*}
Note that we need the remainder term to be $o_\Prob(K_\theta^{-1/2})$ instead of $o_\Prob(1)$ for the diverging $K_\theta$. Denote $\mathcal{A} \equiv R W - R \Delta$ where $W = \widehat\Theta {\bf X}^\top \epsilon / \sqrt{n}$, and we want to show that
\begin{align*}
    \frac{\mathcal{A}^\top M_\theta \mathcal{A}}{\widehat{D}} = \frac{\epsilon^\top P_\theta \epsilon }{D} + o_\Prob(K_\theta^{-1/2}),
\end{align*}
where
\begin{align*}
    \widehat{D} \equiv \sum_{i \in [n]} P_{\theta, ii} \widehat{\sigma}_i^2, \qquad D \equiv \sum_{i \in [n]} P_{\theta,ii} \sigma_i^2, \qquad M_\theta \equiv \left(R \widehat{\Theta} \widehat\Sigma \widehat\Theta^\top R^\top + \theta I_r\right)^{-1}.\quad 
\end{align*}
We have
\begin{align*}
    \frac{\mathcal{A}^\top M_\theta \mathcal{A}}{\widehat{D}} &= \frac{\epsilon^\top P_\theta \epsilon }{\widehat{D}} -  \frac{2 \epsilon^\top {\bf X} \widehat{\Theta}^\top R^\top M_\theta R \Delta}{\sqrt{n}\widehat{D}} + \frac{\Delta^\top R^\top M_\theta R \Delta}{\widehat{D}},
\end{align*}
and we claim that
\begin{align}
   &\frac{\epsilon^\top P_\theta \epsilon }{\widehat{D}} = \frac{\epsilon^\top P_\theta \epsilon }{{D}} + o_\Prob(K_\theta^{-1/2}), \\
   &\frac{ \epsilon^\top {\bf X} \widehat{\Theta}^\top R^\top M_\theta R \Delta}{\sqrt{n}\widehat{D}} = O_\Prob\left(\sqrt{\frac{m \lVert \Delta \rVert_\infty^2}{ K_\theta} }\right) = o_\Prob(K_\theta^{-1/2}),\label{eq:step1:1} \\
   & \frac{\Delta^\top R^\top M_\theta R \Delta}{\widehat{D}} = O_\Prob\left(\frac{m \lVert \Delta \rVert_\infty^2}{ K_\theta} \right) = o_\Prob(K_\theta^{-1/2}).\label{eq:step1:2}
\end{align}
First, the proof of Lemma \ref{lem:variance-estimator-consistent} implies that
\begin{align*}
    \frac{1}{K_\theta} \left( \widehat{D} - D \right) = O_\Prob \left( \frac{1+\lVert {\bf X} (\widehat{\beta} - \beta^0 ) \rVert_2}{\sqrt{n}} + \frac{\lVert {\bf X} (\widehat{\beta} - \beta^0 ) \rVert_2^2}{n} \right) = o_\Prob(r^{-1/2})
\end{align*}
by assumption. As $o_\Prob(r^{-1/2}) = o_\Prob(K_\theta^{-1/2})$ by $K_\theta \leq r_n \leq r$, this further implies that
\begin{align*}
    \frac{\epsilon^\top P_\theta \epsilon }{\widehat{D}} = \frac{\epsilon^\top P_\theta \epsilon }{{D}} \frac{D / K_\theta}{\widehat{D} / K_\theta} = \frac{\epsilon^\top P_\theta \epsilon }{{D}} \left(1 + o_\Prob\left(K_\theta^{-1/2}\right)\right) = \frac{\epsilon^\top P_\theta \epsilon }{{D}} + o_\Prob (K_\theta^{-1/2}),
\end{align*}
by $K_\theta/D = O(1)$, $\mathbb{E}_\epsilon \left[ \left|{\eps^\top P_{\theta} \epsilon}/{D}\right|\right] = 1$ and the Markov inequality. Second, observe by H\"older's inequality
\begin{align*}
    \left\lVert \epsilon^\top n^{-1/2} {\bf X} \widehat{\Theta}^\top R^\top M_\theta R \Delta \right\rVert_1 & \leq \underbrace{\left\lVert \epsilon^\top n^{-1/2} {\bf X} \widehat{\Theta}^\top R^\top M_\theta \right\rVert_2}_{\Xi_1} \underbrace{\left\lVert R \Delta \right\rVert_2}_{\Xi_{2}}.
\end{align*}
Denote $S \subseteq \{1,...,p\}$ to be the set of indexes corresponding to the non-zero column vectors of $R$. By assumption, $ |S| \equiv m$, where we recall $m$ is the number of non-zero columns in $R$, noting that $r \leq m$ by the full-ranked assumption of $R$. For $\Xi_2$, we have 
\begin{align*}
    \Xi_2 = \lVert R_S \Delta_S \rVert_2 \leq \lVert R_S \rVert_2 \lVert \Delta_S \rVert_2 = \lVert R \rVert_2 \lVert \Delta_S \rVert_2 \lesssim \lVert \Delta_S \rVert_2 \leq \sqrt{m}  \lVert \Delta_S \rVert_\infty \leq \sqrt{m}  \lVert \Delta \rVert_\infty,
\end{align*}
where $R_S \in \mathbb{R}^{r \times m}$ is the $S$-indexed column sub-matrix of $R$, and we use the fact that $\lVert R \rVert_2 = O(1)$ by assumption. For $\Xi_1$, we have that
\begin{align*}
    \Xi_1^2 \leq \frac{\epsilon^\top P_{\theta} \epsilon}{\theta} \lesssim \epsilon^\top P_{\theta} \epsilon = O_\Prob\left(K_\theta\right).
\end{align*}
where we use the fact that $\theta \gtrsim 1$ by condition {\it 4} of Theorem \ref{thm:size-control_fixed_design}, and the last equality follows from Markov inequality and
\begin{align*}
    \mathbb{E}_\epsilon \left[\eps^\top P_{\theta} \epsilon\right] = \sum_{i \in [n]} P_{\theta,ii} \sigma_i^2 \lesssim K_\theta.
\end{align*}
Combining with Lemma \ref{lem:variance-estimator-consistent} yields the first equality of \eqref{eq:step1:1}. Third, note that
\begin{align*}
    \Delta^\top R^\top M_\theta R \Delta = O_\Prob\left(m \lVert \Delta \rVert_\infty^2 \right),
\end{align*}
which follows from
\begin{align*}
     \Delta^\top R^\top M_\theta R \Delta \leq \frac{\lVert R \Delta \rVert_2^2}{\theta}  = \frac{\lVert R_S \Delta_S \rVert_2^2}{\theta} \lesssim m \lVert \Delta \rVert_\infty^2.
\end{align*}
Combining with Lemma \ref{lem:variance-estimator-consistent} yields the first equality of \eqref{eq:step1:2}. The second equality of \eqref{eq:step1:1} and \eqref{eq:step1:2} follow from assumption {\it 3} of Theorem \ref{thm:size-control_fixed_design}.

\vspace{5mm}
\noindent {\it Step 2:} We can write the rejection criteria as
\begin{align}
    \widehat{Q}(q) > q_{1-\alpha}(F_{\widetilde{w}_{n}}) + (q_{1-\alpha}(F_{\widetilde{w}_{n}})-1) \left( \frac{  {\sqrt{ K_\theta \widehat{\Phi}}}  }{\left({\sum_{i \in [n]} P_{\theta,ii} \widehat{\sigma}_i^2}\right)\sqrt{2 \sum_{i \in [r_n]} \widetilde{w}_{i,n}^2 } } -1 \right),
    \label{eq:limit:fixed_K}
\end{align} 
and we show that the second term is $o_\Prob(1)$. We note the following holds:
\begin{align*}
& (a) \quad \lVert w_n \rVert^2_2 \times \left(\sum_{i \in [n]} P_{\theta,ii} \sigma_i^2 \right)^2 = \tr (\Omega^2) \left(\sum_{i \in [n]} P_{\theta,ii} \sigma_i^2 \right)^2 =  \tr(P_\theta \Lambda P_\theta \Lambda) = \sum_{i \in [n]} \sum_{j \in [n]} P_{\theta,ij}^2 \sigma_i^2 \sigma_j^2, \\
& (b) \quad {\widehat\Phi}/{\Phi} = 1+ o_\Prob(1), \\
& (c) \quad \frac{1}{K_\theta} \sum_{i \in [n]} P_{\theta, ii} \widehat{\sigma}_i^2 = \frac{1}{K_\theta} \sum_{i \in [n]} P_{\theta, ii} \sigma_i^2 + o_\Prob(1) .
\end{align*}
where $(b)$ follows from Lemma \ref{lem:consistency:phi-hat} and the equality in $(c)$ follows from Lemma \ref{lem:variance-estimator-consistent}. Then from (c) we obtain
\begin{align*}
  & (d) \quad \frac{\frac{1}{K_\theta} \sum_{i \in [n]} P_{\theta,ii} \sigma_i^2 }{\frac{1}{K_\theta} \sum_{i \in [n]} P_{\theta,ii} \widehat{\sigma}_i^2  } = \frac{\frac{1}{K_\theta} \sum_{i \in [n]} P_{\theta,ii} \sigma_i^2  }{ \frac{1}{K_\theta} \sum_{i \in [n]} P_{\theta,ii} \sigma_i^2 + o_\Prob(1) } \overset{\Prob}{\rightarrow } 1.
\end{align*}
Furthermore, we note that
\begin{align*}
    \Phi = \frac{2}{K_\theta}\sum_{i,j \in [n]}  P_{\theta,ij}^2 \sigma_i^2 \sigma_j^2 + o(1) 
\end{align*}
by
\begin{align*}
    \frac{2}{K_\theta} \sum_{i \in [n]} P_{\theta,ii}^2 
    \sigma_i^4 \lesssim \frac{1}{K_\theta} \sum_{i \in [n]} P_{\theta,ii}^2 \leq p_n = O\left(\frac{K_\theta}{n}\right) = o(1),
\end{align*}
so that by (b) we can obtain
\begin{align*}
 &  (e) \quad \frac{\sqrt{\widehat{\Phi}} }{\sqrt{ \frac{1}{K_\theta} \sum_{i \in [n]} \sum_{j \in [n]} P_{\theta,ij}^2 \sigma_i^2 \sigma_j^2 }} = \sqrt{\frac{ \Phi}{ \frac{1}{K_\theta} \sum_{i,j \in [n]} P_{\theta,ij}^2 \sigma_i^2 \sigma_j^2}} (1+o_\Prob(1)) = \sqrt{2} + o_\Prob(1).
\end{align*}
Putting it together,
\begin{align*}
\frac{\sqrt{\widehat{\Phi}}}{\frac{1}{\sqrt{K_\theta}} \sum_{i \in [n]} P_{\theta,ii} \widehat{\sigma}_i^2 } &= \frac{\sqrt{\frac{1}{K_\theta} \sum_{i \in [n]} \sum_{j \in [n]} P_{\theta,ij}^2 \sigma_i^2 \sigma_j^2 }}{\frac{1}{\sqrt{K_\theta}} \sum_{i \in [n]} P_{\theta,ii} \sigma_i^2 } \times \frac{\frac{1}{K_\theta} \sum_{i \in [n]} P_{\theta,ii} \sigma_i^2 }{\frac{1}{K_\theta} \sum_{i \in [n]} P_{\theta,ii} \widehat{\sigma}_i^2  } \times \frac{\sqrt{\widehat{\Phi}} }{\sqrt{ \frac{1}{K_\theta} \sum_{i \in [n]} \sum_{j \in [n]} P_{\theta,ij}^2 \sigma_i^2 \sigma_j^2 }} \\
&\overset{(d),(e)}{=}  \frac{\sqrt{\frac{1}{K_\theta} \sum_{i \in [n]} \sum_{j \in [n]} P_{\theta, ij}^2 \sigma_i^2 \sigma_j^2 }}{\frac{1}{\sqrt{K_\theta}} \sum_{i \in [n]} P_{\theta, ii} \sigma_i^2 }(1+ o_\Prob(1)) (\sqrt{2} + o_\Prob(1))  \\
&= \sqrt{2} \frac{\sqrt{ \sum_{i \in [n]} \sum_{j \in [n]} P_{\theta, ij}^2 \sigma_i^2 \sigma_j^2 }}{ \sum_{i \in [n]} P_{\theta, ii} \sigma_i^2 } + o_\Prob(1) \\
& \overset{(a)}{=} \sqrt{2} \lVert w_n \rVert_2 + o_\Prob(1),
\end{align*}
so that 
\begin{align*}
    \frac{\sqrt{\widehat{\Phi}}}{\frac{1}{\sqrt{K_\theta}}\sum_{i \in [n]} P_{\theta,ii} \widehat{\sigma}_i^2} \times \frac{1}{\sqrt{2 \sum_{i \in [r_n]} \widetilde{w}_{i,n}^2  } } = \frac{\sqrt{2} \lVert w_n \rVert_2 + o_\Prob(1)}{\sqrt{2} \lVert \widetilde{w}_n \lVert_2} = 1 + o_\Prob(1),
\end{align*}
where the last equality follows from
\begin{align*}
    \left| \lVert \widetilde{w}_n \rVert_2 - \lVert w_n \rVert_2 \right| \leq \lVert \widetilde{w}_n - w_n \rVert_2 = o_\Prob(1),
\end{align*}
by Lemma \ref{lem:consistency:weights}, and that $\lVert w_n \rVert_2 \geq 1/\sqrt{r_n}$ is bounded away from zero. Therefore, 
\begin{align*}
   (q_{1-\alpha} (F_{\widetilde{w}_n}) - 1) \left( \frac{\sqrt{\widehat{\Phi}}}{\frac{1}{\sqrt{K_\theta}}\sum_{i \in [n]} P_{\theta,ii} \widehat{\sigma}_i^2}  \frac{  1 }{\sqrt{2 \sum_{i \in [r_n]} \widetilde{w}_{i,n}^2 } } -1 \right) = O_\Prob(1) \times o_\Prob(1) = o_\Prob(1),
\end{align*}
since $q_{1-\alpha} (F_{\widetilde{w}_n}) = O_\Prob(1)$ by Markov inequality.

\vspace{5mm}
\noindent {\it Step 3:} We require two additional lemmas before proceeding.

\begin{lem}
    Let $F_n$ and $G_n$ be two cumulative distribution functions, and let $q_\alpha$ be the quantile function for fixed $\tilde \alpha \in (0,1)$. Denote $q_n := q_{\tilde \alpha}(F_n)$. If
    \begin{align*}
        \sup_{x \in \mathbb R} \left| F_n(x) - G_n(x) \right| = o_\Prob (1),
    \end{align*}
    and for every $\eps > 0$, there exists $\eta(\eps) > 0$ such that
    \begin{align*}
        \mathbb P \left( \left\{F_n(q_n - \eps) \leq \tilde \alpha -\eta(\eps) \right\} \cap \left\{F_n(q_n + \eps) \geq \tilde \alpha + \eta(\eps) \right\} \right) \rightarrow 1,
    \end{align*}
    then 
    \begin{align*}
        q_{\tilde \alpha}(G_n) = q_n + o_\Prob(1).
    \end{align*}
    \label{lem:quantile_function}
\end{lem}

\noindent \textbf{Proof of Lemma \ref{lem:quantile_function}:} Denote $\tilde q_n := q_{\tilde \alpha}(G_n)$ and $\Delta_n := \sup_{x \in \mathbb R} \left| F_n(x) - G_n(x) \right|$. Fix a $\eps > 0$, on the event $$\left\{F_n(q_n - \eps) \leq \tilde \alpha -\eta(\eps) \right\} \cap \left\{F_n(q_n + \eps) \geq \tilde \alpha + \eta(\eps) \right\} \cap \left\{ \Delta_n \leq \eta(\eps)/2 \right\},$$
which has probability approaching one, we have
\begin{align*}
    G_n(q_n - \eps) \leq F_n(q_n - \eps) + \Delta_n < \tilde \alpha,
\end{align*}
so that $q_n - \eps \leq \tilde q_n$. Similarly, on the same event,
\begin{align*}
    G_n(q_n + \eps) \geq F_n(q_n + \eps) - \Delta_n \geq \tilde \alpha,
\end{align*}
so that $q_n + \eps \geq \tilde q_n$. It follows that
\begin{align*}
    \mathbb P \left( \left| \tilde q_n - q_n \right| \geq \eps \right) \rightarrow 0,
\end{align*}
and since $\eps > 0$ is arbitrary, the proof is complete. \qed

\begin{lem}
    Let $F_n$ be the cumulative distribution function of $X_n = \sum_{i \in [r_n]} w_{i,n}\,\chi^2_{1,i}$ with $r_n$ bounded, where $\{\chi_{1,i}^2\}_{i \in [r_n]}$ are $r_n$ independent chi-squared random variables with one degree of freedom, and $\{w_{i,n}\}_{i \in [r_n]}$ are non-negative and sum up to one. For every $\eps > 0$, there exists $\eta(\eps) > 0$ such that
    \begin{align*}
        F_n(q_n - \eps) \leq \tilde \alpha -\eta(\eps) \quad \text{and} \quad F_n(q_n + \eps) \geq \tilde \alpha + \eta(\eps)
    \end{align*}
    for every $n$, where $q_n := q_{\tilde \alpha}(F_n)$ for fixed $\tilde \alpha \in (0,1)$.
    \label{lem:weighted_chi_squared}
\end{lem}

\noindent \textbf{Proof of Lemma \ref{lem:weighted_chi_squared}:} Let $R < \infty$ be an upper bound for $r_n$. Padding with zeros if necessary, we can write $X_n = \sum_{i \in [R]} w_{i,n}\,\chi^2_{1,i}$, where $\{\chi_{1,i}^2\}_{i \in [R]}$ are $R$ independent chi-squared random variables with one degree of freedom.  Because $\{w_{i,n}\}_{i \in [R]}$ are non-negative and sum up to one, we may assume without loss of generality that $w_{1,n} \geq 1/R$. This can be made uniformly in $n$ by switching labels if necessary. By choosing $\underline c_{\tilde \alpha}$ small enough, we have
\begin{align*}
    \mathbb P \left( X_n \leq \underline c_{\tilde \alpha} \right) \leq \mathbb P \left( \chi^2_{1,1} \leq R \underline c_{\tilde \alpha} \right) < \tilde \alpha, 
\end{align*}
so that $q_n \geq \underline c_{\tilde \alpha}$. On the other hand, by choosing $\bar c_{\tilde \alpha}$ large enough, we have
\begin{align*}
    \mathbb P \left( X_n \leq \bar c_{\tilde \alpha} \right) = 1 - \mathbb P \left( X_n >  \bar c_{\tilde \alpha} \right) \geq 1 - \frac{1}{\bar c_{\tilde \alpha}} \geq \tilde \alpha
\end{align*}
by Markov inequality, so that $q_n \leq \bar c_{\tilde \alpha}$. Let $X_{1,n} := w_{1,n} \chi^2_{1,1}$ and $X_{2,n} := \sum_{i=2}^R w_{i,n}\,\chi^2_{1,i}$, then by conditioning on $X_{2,n}$, we can write the density function of $F_n$ as
\begin{align*}
    f_n(x) = \mathbb E_{X_{2,n}} \left[ f_{1,n} \left( x - X_{2,n} \right) \right],
\end{align*}
where $f_{1,n}(x)$ is the density function of $X_{1,n}$. For $x \in [\underline c_{\tilde \alpha}/2,\bar c_{\tilde \alpha}+\underline c_{\tilde \alpha}/2]$, and on the event $X_{2,n} \leq \underline c_{\tilde \alpha}/4$, we have $x - X_{2,n} \in [\underline c_{\tilde \alpha}/4,\bar c_{\tilde \alpha}+\underline c_{\tilde \alpha}/2]$ and note that $w_{1,n} \in [1/R, 1]$, so that $f_{1,n}(x-X_{2,n})$ can be lower bounded on this compact set. In addition, we have
\begin{align*}
    \mathbb P \left(X_{2,n} \leq \underline c_{\tilde \alpha}/4 \right) \geq \mathbb P \left(\bigcap_{i=2}^{R} \{ \chi^2_{1,i} \leq \underline c_{\tilde \alpha}/4 \} \right) = \mathbb P \left( \chi^2_{1} \leq \underline c_{\tilde \alpha}/4 \right)^{R-1}
\end{align*}
which can also be lower bounded as $R$ is bounded. Therefore, we have
\begin{align*}
    f_n(x) \geq \mathbb E_{X_{2,n}} \left[ f_{1,n} \left( x - X_{2,n} \right) 1_{\{ X_{2,n} \leq \underline c_{\tilde \alpha}/4 \} } \right] \geq c(\tilde \alpha) > 0, \quad x \in [\underline c_{\tilde \alpha}/2,\bar c_{\tilde \alpha} +\underline c_{\tilde \alpha}/2],
\end{align*}
for some $c(\tilde \alpha)$ independent of $n$. Finally, for every $\eps > 0$, let $\delta = \min\{\eps, \underline c_{\tilde \alpha}/2\}$, so that $[q_n - \delta, q_n + \delta] \subset [\underline c_{\tilde \alpha}/2,\bar c_{\tilde \alpha}+\underline c_{\tilde \alpha}/2]$. We have
\begin{align*}
    \tilde \alpha - F_n(q_n - \eps) \geq \tilde \alpha - F_n(q_n - \delta) = \int_{q_n-\delta}^{q_n} f_n(x) dx \geq c(\tilde \alpha) \delta,
\end{align*}
and
\begin{align*}
    F_n(q_n + \eps) - \tilde \alpha \geq F_n(q_n + \delta) - \tilde \alpha = \int_{q_n}^{q_n+\delta} f_n(x) dx \geq c(\tilde \alpha) \delta,
\end{align*}
where we also use the fact that $F_n(q_n) = \tilde \alpha$. Thus we can choose $\eta(\eps) = c(\tilde \alpha) \min\{\eps, \underline c_{\tilde \alpha}/2\}$. This concludes the proof. 
\qed

\vspace{5mm}
From {\it Step 1} and {\it Step 2}, we can write \eqref{eq:limit:fixed_K} as 
\begin{align*}
    \frac{\epsilon^\top P_\theta \epsilon}{\sum_{i \in [n]} P_{\theta,ii} \sigma_i^2} > q_{1-\alpha}(F_{\widetilde{w}_{n}}) + R_n,
\end{align*}
with $R_n = o_\Prob(1)$. By using \cite[Theorem 1.5]{rudelson2015small}, similar to the proof of 
\cite[Lemma H.1]{LWZ2024}, we have the anti-concentration inequality
\begin{align}
    \sup_{y \in \mathbb{R}} \mathbb{P}\left( \left|\sum_{i \in [r_n]} w_{i,n}\,\chi^2_{1,i} - y\right| \leq \ell  \right) \leq C_\zeta \ell^{(1-\zeta)/2}
    \label{eq:Eps_anti}
\end{align}
for any $\ell >0$ and $\zeta \in (0,1)$, where $C_\zeta$ only depends on $\zeta$. From Lemma \ref{lem:consistency:weights}, we have
\begin{align*}
    \sum_{i \in [r_n]} \widetilde{w}_{i,n}\,\chi^2_{1,i} = \sum_{i \in [r_n]} w_{i,n}\,\chi^2_{1,i} + o_\Prob(1),
\end{align*}
where we also use the fact that $r_n$ is bounded. By using the same argument as in the proof of 
Lemma \ref{lem:anti-concentration}, we have
\begin{align*}
        &\quad \sup_{y \in \mathbb{R}} \left| \mathbb{P}\left(\sum_{i \in [r_n]} \widetilde{w}_{i,n}\,\chi^2_{1,i}\leq y \bigg| Y \right) - \mathbb{P} \left(\sum_{i \in [r_n]} w_{i,n}\,\chi^2_{1,i} \leq y \right)  \right| \\
        &\leq \sup_{y \in \mathbb{R}}\mathbb{P}\left( \left|\sum_{i \in [r_n]} w_{i,n}\,\chi^2_{1,i} - y \right|\leq \ell \right) + \mathbb{P}\left( \left|\sum_{i \in [r_n]} \widetilde{w}_{i,n}\,\chi^2_{1,i} - \sum_{i \in [r_n]} w_{i,n}\,\chi^2_{1,i} \right|>\ell \bigg| Y  \right).
    \end{align*}
Together with \eqref{eq:Eps_anti}, this implies that the Kolmogorov distance between the distribution of $F_{\widetilde{w}_n} $ conditional on the data and the distribution of $F_{w_n}$ is $o_\Prob(1)$, which further implies that $q_{1-\alpha}(F_{\widetilde{w}_n}) = q_{1-\alpha}(F_{w_n}) + o_\Prob(1)$, in view of Lemmas \ref{lem:quantile_function} and \ref{lem:weighted_chi_squared}. In addition, we note that
\begin{align*}
    \frac{\mathcal{E}^\top P_\theta \mathcal{E}}{\sum_{i \in [n]}P_{\theta,ii}\sigma_i^2} &= \frac{\mathcal{Z}^\top \Lambda^{1/2} P_\theta \Lambda^{1/2} \mathcal{Z}}{\sum_{i \in [n]}P_{\theta,ii}\sigma_i^2} = \mathcal{Z}^\top \mathcal{G}_n \mathcal{Z},
\end{align*}
where $\mathcal{Z}\overset{d}{=} \mathcal{N}(0,I_n)$, $\Lambda = \diag (\sigma_1^2, \cdots, \sigma_n^2)$ and 
\begin{align*}
    \mathcal{G}_n = \frac{ \Lambda^{1/2} P_\theta \Lambda^{1/2}}{\sum_{i \in [n]}P_{\theta,ii}\sigma_i^2} = \frac{ \Lambda^{1/2} {\bf X} \widehat{\Theta}^\top R^\top (R \widehat{\Theta} \widehat{\Sigma} \widehat{\Theta}^\top R^\top + \theta I_r)^{-1} R \widehat{\Theta} {\bf X}^\top \Lambda^{1/2}}{n \sum_{i \in [n]}P_{\theta,ii}\sigma_i^2}.
\end{align*}
By Theorem 1.3.22 of \cite{Horn-Johnson-Matrix-Analysis}, the matrix $\mathcal{G}_n$ has the same $r_n$ non-zero eigenvalues as $\Omega$, where the remaining $(n - r_n)$ eigenvalues of $\mathcal{G}_n$ are zero. Then by singular value decomposition, we have 
\begin{align*}
    \mathcal{G}_n = U D^* U^\top,
\end{align*}
where $UU^\top = U^\top U = I_n$ and $D^* = \diag(w_{1,n},...,w_{r_n,n}, \mathbf{0}_{n - r_n})$ collects the eigenvalues of $\mathcal{G}_n$. Since $U^\top \mathcal{Z} \overset{d}{=} \mathcal{N}(0,I_n)$, we obtain
\begin{align*}
     \mathcal{Z}^\top \mathcal{G}_n \mathcal{Z} = \mathcal{Z}^\top U D^* U^\top \mathcal{Z} \overset{d}{=} \sum_{i \in [r_n]} w_{i,n} \chi^2_{1,i}.
\end{align*}
By Corollary \ref{cor:strong_approx} and \eqref{eq:Eps_anti}, we have
\begin{align}
    \sup_{y \in \mathbb{R}} \mathbb{P}\left( \left|\frac{\epsilon^\top P_\theta \epsilon}{\sum_{i \in [n]} P_{\theta,ii} \sigma_i^2} - y\right| \leq \ell  \right) \leq C'_\zeta \ell^{(1-\zeta)/2} + o(1),
    \label{eq:eps_anti}
\end{align}
as in the proof of Lemma \ref{lem:anti-concentration:main}, where the $o(1)$ term does not depend on $y \in \mathbb{R}$. It follows that, with $\tilde R_n = R_n + o_\Prob(1)$,
\begin{align*}
    \mathbb P \left( \widehat Q(q) > \mathcal{C}_{\alpha} \right) &= \mathbb P \left(  \frac{\epsilon^\top P_\theta \epsilon}{\sum_{i \in [n]} P_{\theta,ii} \sigma_i^2} > q_{1-\alpha}(F_{w_{n}}) + \tilde R_n \right) \\
    &= \mathbb P \left(  \frac{\epsilon^\top P_\theta \epsilon}{\sum_{i \in [n]} P_{\theta,ii} \sigma_i^2} > q_{1-\alpha}(F_{w_{n}}) \right) + o(1) \\
    &= \mathbb P \left(  \frac{\mathcal{E}^\top P_\theta \mathcal{E}}{\sum_{i \in [n]}P_{\theta,ii}\sigma_i^2} > q_{1-\alpha}(F_{w_{n}}) \right) + o(1) \\
    &= \mathbb P \left(  \sum_{i \in [r_n]} w_{i,n}\,\chi^2_{1,i} > q_{1-\alpha}(F_{w_{n}})  \right) + o(1) \\
    &= \alpha + o(1).
\end{align*}
where the second equality follows from Lemma \ref{lem:anti-concentration} (with an additional $o(1)$ term) and \eqref{eq:eps_anti}, and the third equality follows from Corollary \ref{cor:strong_approx}. This completes the proof of \eqref{eq:assumption_B*}.

\vspace{5mm}
\noindent \begin{Large} \textbf{Part $II$ (Diverging $K_\theta$ case):} \end{Large} Assume now that $K_\theta \rightarrow \infty$. We break our proof into three steps.

\vspace{5mm}
\noindent {\it Step 1:} We claim that
\begin{align}
    \frac{F_{\widetilde{w}_n} - 1}{\sqrt{2 \sum_{i \in [r_n]} \widetilde{w}_{i,n}^2 } } \convD \mathcal{N}(0,1)
    \label{eq:distributional_conV}
\end{align}
in probability conditional on the data. Note that, conditional on the data, $F_{\widetilde{w}_n}$ has mean one and variance $2 \sum_{i \in [r_n]} \widetilde{w}_{i,n}^2$. Additionally, we claim that 
\begin{align}
    \frac{\max_i \widetilde{w}_{i,n}^2  }{\sum_{i \in [r_n]} \widetilde{w}_{i,n}^2}  = o_\Prob(1).
    \label{eq:ratio_max_w_sum_w}
\end{align}
By the AM-QM inequality, $\sum_{i \in [r_n]} \widetilde{w}_{i,n} = \tr (\widehat\Omega) = 1$, Lemma \ref{lem:consistency:weights}, and $K_\theta \asymp r_n$ (c.f. \eqref{eq:K_lambda:summary}), we have 
\begin{align*}
    \frac{\max_{i \in [r_n]} \widetilde{w}_{i,n}^2}{\sum_{i \in [r_n]} \widetilde{w}_{i,n}^2} \leq \frac{\max_{i\in [r_n]} \widetilde w_{i,n}^2}{r_n \left(\sum_{i \in [r_n]} \widetilde{w}_{i,n} / r_n\right)^2} = r_n \max_{i \in [r_n]} \widetilde{w}_{i,n}^2 = o_\Prob(\frac{r_n}{K_\theta}) = o_\Prob(1), 
\end{align*}
showing that \eqref{eq:ratio_max_w_sum_w} holds. To show \eqref{eq:distributional_conV} conditional on the data, we apply the Lindeberg-Feller Central Limit Theorem (CLT). The conditions of the CLT are satisfied by noting that for any $\eta>0$, $T_{i,n} \equiv \frac{ \widetilde{w}_{i,n} (\chi^2_{1,i}-1)}{\sqrt{2 \sum_{i \in [r_n]} \widetilde{w}_{i,n}^2}} $ and $a_{i,n} \equiv \frac{2 \sum_{i \in [r_n]}\widetilde{w}_{i,n}^2}{\widetilde{w}_{i,n}^2}$ for $i \in \mathcal I_n$ where $\mathcal I_n = \left\{i \in [r_n]: \widetilde{w}_{i,n} > 0\right\}$,
\begin{align*}
    \sum_{i \in [r_n]} \mathbb{E} \left[ T_{i,n}^2 \mathbb{1}\{ |T_{i,n}| > \eta \} \right] &= \sum_{i \in \mathcal I_n} \frac{\widetilde{w}_{i,n}^2}{2 \sum_{i \in [r_n]}\widetilde{w}_{i,n}^2} \mathbb{E}\left[ (\chi^2_{1,i}-1)^2 \mathbb{1}\left\{ (\chi^2_{1,i}-1)^2 \geq \eta^2\; a_{i,n} \right\}\right] \\
    &\lesssim \max_{i \in \mathcal I_n}  \mathbb{E}\left[ (\chi^2_{1}-1)^2 \mathbb{1}\left\{ (\chi^2_{1}-1)^2 \geq \eta^2\;a_{i,n} \right\}\right] \\
    &= \mathbb{E}\left[ (\chi^2_{1}-1)^2 \mathbb{1}\left\{ (\chi^2_{1}-1)^2 \geq \eta^2\;\min_{i \in \mathcal I_n} a_{i,n} \right\}\right].
\end{align*}
By the dominated convergence theorem and $\mathbb E(\chi^2_1-1)^2 = O(1)$, we have $\mathbb{E}\left[ (\chi^2_{1}-1)^2 \mathbb{1}\left\{ (\chi^2_{1}-1)^2 \geq \eta^2\;a \right\}\right] \rightarrow 0$ as $a \rightarrow \infty$. It follows that 
\begin{align*}
    \sum_{i \in [r_n]} \mathbb{E} \left[ T_{i,n}^2 \mathbb{1}\{ |T_{i,n}| > \eta \} \right] \convP 0
\end{align*}
by noting that
\begin{align*}
    \min_{i \in \mathcal I_n} a_{i,n} = \frac{2 \sum_{i \in [r_n]}\widetilde{w}_{i,n}^2}{\max_{i \in [r_n]} \widetilde{w}_{i,n}^2}  \overset{\mathbb{P}}{\rightarrow} \infty
\end{align*}
by \eqref{eq:ratio_max_w_sum_w}. 

\vspace{5mm}
\noindent {\it Step 2:} We claim that
\begin{align}
     \frac{ \frac{1}{\sqrt{K_\theta}} \sum_{i \in [n]} P_{\theta,ii} \widehat{\sigma}_i^2 }{\sqrt{\widehat{\Phi} } } \left( \widehat{Q}(q) - 1 \right) &= \frac{\epsilon^\top P_\theta \epsilon - \sum_{i \in [n]} P_{\theta,ii} \sigma_i^2}{\sqrt{ \widehat{\Phi} K_\theta}} + o_\Prob(1) \notag \\
     &= \frac{\sum_{i \in [n], j \neq i} P_{\theta,ij} \epsilon_i \epsilon_j }{\sqrt{\widehat{\Phi} K_\theta}} + o_\Prob(1) \notag \\
     &\convD \mathcal{N}\left(0, 1 \right).
     \label{eq:moving_K_convergence}
 \end{align}
Indeed, the first equality follows from Step 1 of the bounded $K_\theta$ case and Lemma \ref{lem:consistency:phi-hat}. For 
the second equality, note that
\begin{align*}
    \frac{\sum_{i \in [n]} P_{\theta,ii} (\epsilon_i^2 - \sigma_i^2)}{\sqrt{K_\theta}} = o_\Prob(1)
\end{align*}
by Markov inequality. Then, by Lemma \ref{lem:variance-estimator-consistent} and the fact that
\begin{align*}
    \Phi \gtrsim \frac{\sum_{i \in [n]}\sum_{j \neq i}P_{\theta,ij}^2}{K_\theta} \gtrsim 1 
\end{align*}
by condition {\it 4} of Theorem \ref{thm:size-control_fixed_design}, we have
\begin{align*}
    \frac{\sum_{i \in [n]} P_{\theta,ii} (\epsilon_i^2 - \sigma_i^2)}{\sqrt{\widehat{\Phi}K_\theta}} &= \frac{\sum_{i \in [n]} P_{\theta,ii} (\epsilon_i^2 - \sigma_i^2)}{\sqrt{K_\theta}} \sqrt{\frac{\Phi}{\widehat\Phi}} \frac{1}{\sqrt{\Phi}} = o_\Prob(1) \left(1+o_\Prob(1)\right) O(1) = o_\Prob(1).
\end{align*}
Finally, note that by \cite[Theorem 3.1]{dovi2024ridge}
 \begin{align*}
     \frac{\sum_{i \in [n], j \neq i} P_{\theta,ij} \epsilon_i \epsilon_j }{\sqrt{\Phi K_\theta}} \convD \mathcal{N}(0,1), 
 \end{align*}
which is justified because 
\begin{align*}
    \frac{1}{r_n} \sum_{i \in [n]} \sum_{j \neq i} P_{\theta,ij}^2 \gtrsim 1,
\end{align*}
by the fact that $K_\theta \asymp r_n$ and condition {\it 4} of Theorem \ref{thm:size-control_fixed_design}. Therefore, by Lemma \ref{lem:consistency:phi-hat}, 
\begin{align*}
    \frac{\sum_{i \in [n], j \neq i} P_{\theta,ij} \epsilon_i \epsilon_j }{\sqrt{\widehat{\Phi} K_\theta}} = \frac{\sum_{i \in [n], j \neq i} P_{\theta,ij} \epsilon_i \epsilon_j }{\sqrt{\Phi K_\theta}} \sqrt{\frac{\Phi}{\widehat\Phi}} = \frac{\sum_{i \in [n], j \neq i} P_{\theta,ij} \epsilon_i \epsilon_j }{\sqrt{\Phi K_\theta}} (1+o_\Prob(1)) \convD \mathcal{N}(0,1).
\end{align*}

\vspace{5mm}
\noindent {\it Step 3: } We can re-express our rejection criteria $\widehat{Q}(q)>\mathcal{C}_\alpha$ as
\begin{align*}
    \frac{ \frac{1}{\sqrt{K_\theta}} \sum_{i \in [n]} P_{\theta,ii} \widehat{\sigma}_i^2 }{\sqrt{\widehat{\Phi} } } \left( \widehat{Q}(q) - 1 \right) > q_{1-\alpha} \left( \frac{F_{\widetilde{w}_n} - 1}{\sqrt{2 \sum_{i \in [r_n]} \widetilde{w}_{i,n}^2}} \right).
\end{align*}
By \cite[Example 1.4.7]{VW96}, 
\begin{align*}
    \left( \frac{ \frac{1}{\sqrt{K_\theta}} \sum_{i \in [n]} P_{\theta,ii} \widehat{\sigma}_i^2 }{\sqrt{\widehat{\Phi} } } \left( \widehat{Q}(q) - 1 \right), q_{1-\alpha} \left( \frac{F_{\widetilde{w}_n} - 1}{\sqrt{2 \sum_{i \in [r_n]} \widetilde{w}_{i,n}^2}} \right) \right) \convD \left(\mathcal N , q_{1-\alpha} (\mathcal N) \right),
\end{align*}
with $\mathcal N \overset{d}{=} \mathcal N(0,1)$, by using \eqref{eq:distributional_conV} and \eqref{eq:moving_K_convergence}. Then, an application of \cite[Theorem 1.3.6]{VW96} yields
\begin{align*}
    \frac{ \frac{1}{\sqrt{K_\theta}} \sum_{i \in [n]} P_{\theta,ii} \widehat{\sigma}_i^2 }{\sqrt{\widehat{\Phi} } } \left( \widehat{Q}(q) - 1 \right) - q_{1-\alpha} \left( \frac{F_{\widetilde{w}_n} - 1}{\sqrt{2 \sum_{i \in [r_n]} \widetilde{w}_{i,n}^2}} \right) \convD \mathcal N - q_{1-\alpha} (\mathcal N).
\end{align*}
Applying \cite[Theorem 1.3.4(vi)]{VW96} yields \eqref{eq:assumption_B*}. The proof is complete. 
\qed

\subsection{Proof of Proposition \ref{prop:ridge_parameter}} \label{subsec:ridge_parameter}

Let $D_1 \geq D_2 \geq \cdots \geq D_{r_n}>0$ be the non-zero singular values of $M \equiv { R \widehat\Theta {\bf X}^\top}/{\sqrt{n}}$. We require two additional lemmas for the proof. 

\begin{lem}
    Suppose there exists constants of $0<t_0,u_0<\infty$ such that,
    \begin{enumerate}
        \item $\frac{\lVert M \rVert_F^2}{r_n} \geq t_0$ and
        \item $D_1 \leq u_0$.
    \end{enumerate}
    Then for any fixed  $\varpi \in \left(0, \min \left\{1,\frac{t_0}{(1+\mu) u_0^2}\right\}\right]$ with $\mu > 0$, we have $\theta \geq \underline{c}$ for some $\underline{c} > 0$. In addition, assume that $\inf_{i \in [n]} \sigma_i^2 \geq \underline{c} >0$, then if we use the $\varpi$ defined above for $\theta$, we have 
    \begin{align*}
        \Phi \geq 2\underline{c}^2 \left(\frac{t_0^2}{4 u_0^4} - p_n\right),
    \end{align*}
    where $p_n = \max_{i \in [n]} P_{\theta,ii}$.
    \label{lem:ridge_parameter_sufficient_condition}
\end{lem}

\noindent \textbf{Proof of Lemma \ref{lem:ridge_parameter_sufficient_condition}:} For notational simplicity, denote $k_n \equiv \lceil \varpi r_n \rceil$ and observe
\begin{align*}
    \lVert M \rVert_F^2 =\sum_{i \in [r_n]} D_i^2 &\leq (k_n-1)D_1^2 + (r_n - k_n+1) D_{k_n}^2 \\
    &= r_n D^2_{k_n} + (k_n-1)(D_1^2 - D^2_{k_n}) \\
    &\leq r_n D^2_{k_n} + \varpi r_n D_1^2,
\end{align*}
where we used $k_n \leq \varpi r_n + 1$ in the last inequality. Rearranging yields
\begin{align*}
    \theta = D_{k_n}^2 \geq \frac{\lVert M \rVert_F^2}{r_n} - \varpi D_1^2 \geq t_0 - \frac{t_0}{1+\mu} \geq \frac{\mu}{1+\mu} t_0,
\end{align*}
by the assumptions that $\varpi \leq \frac{t_0}{(1+\mu) u_0^2}$ and $D_1 \leq u_0$. This proves the first part of the lemma.

\vspace{5mm}
For the second part, observe
\begin{align*}
    \sum_{i \in [n]} \sum_{j \neq i} P_{\theta,ij}^2 = \sum_{i \in [n]} \sum_{j \in [n]} P_{\theta,ij}^2 - \sum_{i \in [n]} P_{\theta,ii}^2,
\end{align*}
where we note
\begin{align*}
    \sum_{i \in [n]} \sum_{j \in [n]} P_{\theta,ij}^2 = \tr(P_\theta^2) &= \sum_{i \in [r_n]} \left( \frac{D_i^2}{D_i^2+\theta} \right)^2 \\
    &\geq \frac{1}{(2u_0^2)^2} \sum_{i \in [r_n]} D_i^4 \\
    &\geq \frac{1}{4u_0^4r_n} \left( \sum_{i \in [r_n]} D_i^2 \right)^2 \\
    &\geq \frac{t_0^2}{4u_0^4} r_n \\
    &\geq \frac{t_0^2}{4u_0^4} K_\theta,
\end{align*}
where we use the definition of $\theta$ and the assumption that $D_1 \leq u_0$ in the first inequality, Cauchy-Schwarz inequality in the second inequality, the assumption that $\frac{\lVert M \rVert_F^2}{r_n} \geq t_0$ in the third inequality, and Lemma \ref{lem:Ktheta_and_rn} in the last inequality.
We also have
\begin{align*}
    \sum_{i \in [n]} P_{\theta,ii}^2 \leq p_n K_\theta,
\end{align*}
whence
\begin{align*}
    \frac{1}{K_\theta} \sum_{i \in [n]} \sum_{j \neq i} P_{\theta,ij}^2 \geq \frac{t_0^2}{4u_0^4} - p_n.
\end{align*}
By the assumption that $\inf_{i \in [n]} \sigma_i^2 \geq \underline{c}>0$, the result follows.
\qed

\begin{lem}\label{lem:op_norm_bound}
    Suppose Assumption \ref{ass:size} holds, and $\lambda_j \asymp \sqrt{\log(np)/n}$ (uniformly in $j \in [p]$) are suitably chosen, then we have 
    \begin{align*}
        \left\lVert R \widehat \Theta \widehat \Sigma \widehat \Theta^\top R^\top - R \Theta R^\top \right\rVert_2 = o_\mathbb{P}(1).
    \end{align*}    
\end{lem}

\noindent {\bf Proof of Lemma \ref{lem:op_norm_bound}:} Define 
\begin{align*}
    A := R \widehat \Theta \widehat \Sigma \widehat \Theta^\top R^\top - R \Theta R^\top = R B R^\top, \quad B := \widehat \Theta \widehat \Sigma \widehat \Theta^\top -  \Theta,  
\end{align*}
and we aim to bound $\left\Vert A \right \Vert_2$. Fix an arbitrary $x \in \mathbb R^r$ with $\lVert x \rVert_2 = 1$, and define $y = R^\top x \in \mathbb R^p$. Note that $\lVert y \rVert_2 \lesssim 1$, by the assumption that $\lambda_{\max}(R  R^\top) \leq \overline{c} < \infty$. Without loss of generality, we may assume that the last $p-m$ column vectors of $R$ are zero, which implies that the last $p-m$ elements of $y$ are zero. We have
\begin{align*}
    \left|x^\top A x \right| = \left| y^\top B y \right| &= \left| \sum_{i,j \in [m]} B_{ij} y_i y_j \right| \\
    &\leq \left\Vert B  \right \Vert_{\max} \left(\sum_{i \in [m]} |y_i| \right) \left(\sum_{j \in [m]} |y_j| \right) \\
    &\leq m \left\Vert B  \right \Vert_{\max} \left(\sum_{i \in [m]} y_i^2 \right)^{1/2} \left(\sum_{j \in [m]} y_j^2 \right)^{1/2} \\
    &= m \left\Vert B  \right \Vert_{\max} \left(\sum_{i \in [p]} y_i^2 \right)^{1/2} \left(\sum_{j \in [p]} y_j^2 \right)^{1/2} \\
    &\lesssim m \left\Vert B  \right \Vert_{\max},
\end{align*}
where we use Cauchy-Schwarz inequality in the second inequality. Since this holds uniformly in $x$, we have $\left\Vert A \right \Vert_2 \lesssim m \left\Vert B  \right \Vert_{\max}$ by the Rayleigh quotient property. It remains to bound $\left\Vert B  \right \Vert_{\max}$, note that
\begin{align*}
	B = \widehat \Theta \widehat \Sigma \widehat \Theta^\top -  \Theta = 	\widehat \Theta \widehat \Sigma \widehat \Theta^\top - \widehat \Theta  + \widehat \Theta -  \Theta = \widehat \Theta \left( \widehat \Sigma \widehat \Theta^\top - I_p \right) + \left(\widehat \Theta -  \Theta \right),
\end{align*}
thus
\begin{align*}
	\left\Vert B  \right \Vert_{\max} &\leq \left\Vert \widehat \Theta \left( \widehat \Sigma \widehat \Theta^\top - I_p \right) \right \Vert_{\max} + \left\Vert\widehat \Theta -  \Theta  \right\Vert_{\max} \\
	&\leq \left\Vert  \widehat \Theta \right\Vert_{\infty} \left\Vert \widehat \Sigma \widehat \Theta^\top - I_p \right \Vert_{\max} + \left\Vert\widehat \Theta -  \Theta  \right\Vert_{\max}.
\end{align*}
We have 
\begin{align*}
    \left\Vert\widehat \Theta -  \Theta  \right\Vert_{\max} \leq \max_j \left\Vert \widehat \Theta_j^\top - \Theta_j^\top \right\Vert_2 = O_\mathbb{P} \left( \sqrt{\frac{s_0^* \log (np)}{n}} + \frac{s_0^* (\log(np))^{4/\nu}}{n} \right)
\end{align*}
by Lemma \ref{lem:random_design_bound}. In addition, note that
\begin{align*}
	\left\Vert  \widehat \Theta \right\Vert_{\infty} \leq \left\Vert  \widehat \Theta - \Theta \right\Vert_{\infty} + \left\Vert  \Theta \right\Vert_{\infty}, 
\end{align*}
and 
\begin{align*}
    \left\Vert  \widehat \Theta - \Theta \right\Vert_{\infty} = \max_j \left\Vert \widehat \Theta_j^\top - \Theta_j^\top \right\Vert_1 = O_\mathbb{P} \left( s_0^*\sqrt{\frac{\log (np)}{n}} + \frac{(s_0^*)^{3/2}(\log(np))^{4/\nu}}{n} \right),
\end{align*}
and since $\lambda_{\max}(\Theta) \leq \overline{c} < \infty$, we have $\max_j \lVert \Theta_j^\top \rVert_2 = O(1)$ and
\begin{align*}
    \left\Vert  \Theta \right\Vert_{\infty}  = \max_j ||\Theta_j^\top||_1 \leq \max_j \sqrt{1+ s_j} ||\Theta_j^\top||_2 \lesssim \sqrt{s_0^*}.
\end{align*}
Finally, we have
\begin{align*}
    \left\Vert \widehat \Sigma \widehat \Theta^\top - I_p \right \Vert_{\max} = O_\Prob\left(\sqrt{\frac{\log (np)}{n}} \right),
\end{align*}
as shown in the proof of Lemma \ref{lem:random_design_bound}. Putting everything together, we obtain
\begin{align*}
    \left\Vert B  \right \Vert_{\max} = O_\mathbb{P} \left( \sqrt{\frac{s_0^* \log (np)}{n}} + \frac{s_0^* (\log(np))^{4/\nu}}{n} \right),
\end{align*}
and thus
\begin{align*}
    \left \Vert A \right \Vert_2 = O_\mathbb{P} \left( \sqrt{\frac{m^2 s_0^* \log (np)}{n}} + \frac{ms_0^* (\log(np))^{4/\nu}}{n} \right) = o_\Prob(1).
\end{align*}
\qed 

\vspace{5mm}
We now verify that the conditions in Lemma \ref{lem:ridge_parameter_sufficient_condition} hold with probability approaching one. Note that the conditions in Assumption \ref{ass:size} imply that
\begin{enumerate}
        \item $ \lambda_{\max}(\Theta) \leq \overline{c} < \infty$,
        \item $\tr(R \Theta R^\top)/r \geq \underline{c}>0$,
        \item $\lambda_{\max}(RR^\top) \leq \overline{c}<\infty $,
        \item $\inf_{i \in [n]} \sigma_i^2 \geq \underline{c} >0$.
\end{enumerate}
Define
\begin{align*}
    u_0 \equiv 2\overline{c}, \quad\quad and \quad\quad t_0 \equiv {\underline{c}}/{2}.
\end{align*}
For the {\bf upper bound} of Lemma \ref{lem:ridge_parameter_sufficient_condition}, note that
\begin{align*}
     \lVert R \Theta R^\top \rVert_2 \leq \lVert R \rVert_2^2 \lVert \Theta \rVert_2 \leq \overline{c}^2,
\end{align*}
and by Lemma \ref{lem:op_norm_bound} and Assumption \ref{ass:size}, we have
\begin{align*}
    D_1^2 = \lVert R \widehat\Theta \widehat\Sigma \widehat{\Theta}^\top R^\top \rVert_2 = \lVert R \Theta R^\top \rVert_2 + o_\Prob(1),
\end{align*}
so that with probability approaching one, 
\begin{align*}
    |D_1| \leq  u_0.
\end{align*}
For the {\bf lower bound} of Lemma \ref{lem:ridge_parameter_sufficient_condition}, note that 
\begin{align*}
    \frac{1}{r} \tr (R \widehat{\Theta} \widehat\Sigma \widehat\Theta^\top R^\top) = \frac{1}{r} \tr (R \Theta R^\top) + \frac{1}{r} \left( \tr (R \widehat{\Theta} \widehat\Sigma \widehat\Theta^\top R^\top) - \tr (R \Theta R^\top) \right),
\end{align*}
and by Lemma \ref{lem:op_norm_bound} and Weyl's theorem, we have
\begin{align*}
    \left|\frac{1}{r} \left( \tr (R \widehat{\Theta} \widehat\Sigma \widehat\Theta^\top R^\top) - \tr (R \Theta R^\top) \right) \right| &\leq \frac{1}{r} \sum_{i \in [r]} \left| \lambda_i(R \widehat{\Theta} \widehat\Sigma \widehat\Theta^\top R^\top) - \lambda_i(R \Theta R^\top) \right| \\
    &\leq \lVert R \widehat \Theta \widehat \Sigma \widehat \Theta^\top R^\top - R \Theta R^\top\rVert_2 = o_\Prob(1).
\end{align*}
Combining with our assumption yields
\begin{align*}
   \frac{\lVert M \rVert_F^2}{r_n} \geq \frac{\lVert M \rVert_F^2}{r} \geq t_0
\end{align*}
by recalling $r_n \leq r$. Finally, for $\Phi$, note that $p_n = O_\Prob(K_\theta/n) = O_\Prob(r/n) = o_\Prob(1)$, since $r \leq m$ and $m/n \leq m s_0^2 \log^2(np)/n = o(1)$. It follows that, with probability approaching one, $p_n \leq t_0^2/(8u_0^4)$, whence
\begin{align*}
    \Phi \geq 2\underline{c}^2 \left(\frac{t_0^2}{4u_0^4} - p_n\right) \geq \frac{\underline{c}^2 t_0^2}{4u_0^4} > 0.
\end{align*}
The proof is complete.
\qed

\begin{rem} \label{rem:varpi}

We can choose $c_{\varpi} = \min\left\{ 1, \frac{\lVert M \rVert_F^2}{(1+\mu) r_n \lVert M \rVert_2^2} \right\}$ for any $\mu>0$. The proof above implies that $\theta \geq \lVert M \rVert_F^2/r_n - \varpi \lVert M \rVert_2^2 \geq (\mu/(1+\mu))\times (\lVert M \rVert_F^2/r_n)$ for any $\varpi \in (0,c_\varpi)$, and since $\lVert M \rVert_F^2/r_n$ can be bounded away from zero with probability approaching one, we have $\theta \gtrsim 1$ asymptotically. Note also that we can choose $\varpi = c_{\varpi}/(1 + \omega)$ for any $\omega > 0$, so that $\varpi$ is also bounded away from zero asymptotically, and thus $K_\theta$ is of the same order as $r_n$ by Lemma \ref{lem:Ktheta_and_rn}. In our simulation we choose $\mu=3$ and $\omega=1$.
    
\end{rem}

\begin{rem} \label{rem:rn_and_r}

From the proof above, we can also deduce that $r_n/r \geq \underline{c} > 0$ with probability approaching one under Assumption \ref{ass:size}. To see this, note that 
\begin{align*}
    \frac{ \tr(R \widehat\Theta \widehat\Sigma \widehat\Theta^\top R^\top)}{r} = \frac{ \tr(R \Theta R^\top)}{r} + o_\mathbb{P}(1),
\end{align*}
and 
\begin{align*}
   \tr(R \widehat\Theta \widehat\Sigma \widehat\Theta^\top R^\top) \leq r_n \lambda_{\max}(R \widehat\Theta \widehat\Sigma \widehat\Theta^\top R^\top) = r_n \left(\lambda_{\max}(R \Theta R^\top) + o_\mathbb{P}(1)\right) \leq r_n (\overline{c} + o_\mathbb{P}(1)),
\end{align*}
where the equality follows from Lemma \ref{lem:op_norm_bound} and Weyl's Theorem. Therefore, 
\begin{align*}
    \underline{c} \leq \frac{\tr(R \Theta R^\top)}{r}+ o_\mathbb{P}(1)  \leq \frac{r_n}{r} (\overline{c}+o_\mathbb{P}(1)) = \frac{r_n}{r}\overline{c} + o_\mathbb{P}(1),
\end{align*}
where we note $r_n \leq r$. It follows that, with probability approaching one,
\begin{align*}
    \frac{r_n}{r} \geq \frac{\underline{c}}{2\overline{c}} > 0.
\end{align*}
\end{rem}

\subsection{Proof of Proposition \ref{prop:sufficient_conditions_for_fixed_design}}
By Lemma~\ref{lem:random_design_bound}, we have
\begin{align*}
    \frac{\sqrt{r}\lVert {\bf X}\widehat\Delta\rVert_2}{\sqrt{n}}
    &=
    O_\Prob\left(
        \sqrt{\frac{r s_0\log(np)}{n}}
    \right)
    =
    o_\Prob(1),\\
    \frac{
        \lVert {\bf X}\widehat\Delta\rVert_\infty
        \lVert {\bf X}\widehat\Delta\rVert_2
    }{\sqrt{n}}
    &=
    O_\Prob\left(
        \sqrt{\frac{s_0^2\log^2(np)}{n}}
        \sqrt{\frac{s_0(\log(np))^{2/\nu}}{n}}
    \right)
    =
    o_\Prob(1),\\
    \sqrt{m}\lVert\Delta\rVert_\infty
    &=
    O_\Prob\left(
        \sqrt{\frac{m s_0^2\log^2(np)}{n}}
    \right)
    =
    o_\Prob(1).
\end{align*}
For the first bound, since $r\leq m$, Assumption~\ref{ass:size}.4.1 implies
\[
    \sqrt{\frac{r s_0\log(np)}{n}}
    \leq
    \sqrt{\frac{m s_0\log(np)}{n}}
    =
    o(1).
\]
Hence, condition~{\it 3} of
Theorem~\ref{thm:size-control_fixed_design} holds. Furthermore, by
Proposition~\ref{prop:ridge_parameter}, we can choose
$\varpi\in(0,1)$ such that condition~{\it 4} of
Theorem~\ref{thm:size-control_fixed_design} holds. Finally, Assumption~\ref{ass:size} directly implies conditions~{\it 1}
and~{\it 2} of Theorem~\ref{thm:size-control_fixed_design}. It remains to
verify the additional requirement $r=o(n)$. Since
$s_0\log(np)\geq1$ and
\[
    \frac{r s_0\log(np)}{n}=o(1),
\]
we have
\[
    \frac{r}{n}
    \leq
    \frac{r s_0\log(np)}{n}
    =
    o(1).
\]
Therefore, $r=o(n)$, completing the proof.

\subsection{Proof of Theorem \ref{thm:size-control_random_design}} 
By Proposition~\ref{prop:sufficient_conditions_for_fixed_design}, the proof proceeds essentially along the same lines as that of Theorem~\ref{thm:size-control_fixed_design}. The main difference is that \({\bf X}\) is now random, requiring the relevant arguments to be established conditionally on \({\bf X}\). For completeness, we provide the details below.

We begin by proving \eqref{eq:assumption_B*} for both bounded and diverging $r$, and break the proof into two parts, part $I$ and $II$.

\vspace{5mm}
\noindent \begin{Large} \textbf{Part $I$ (Bounded $r$ case):} \end{Large} Assume that $r$ is bounded, so that $r_n$ is bounded as $r_n \leq r$, and the same holds for $K_\theta$ by Lemma \ref{lem:Ktheta_and_rn}. We have
\begin{align*}
    \widehat{Q}(q) = \frac{\epsilon^\top P_\theta \epsilon}{\sum_{i \in [n]} P_{\theta,ii} \sigma_i^2} + o_\Prob(1),
\end{align*}
and
\begin{align*}
   (q_{1-\alpha} (F_{\widetilde{w}_n}) - 1) \left( \frac{\sqrt{\widehat{\Phi}}}{\frac{1}{\sqrt{K_\theta}}\sum_{i \in [n]} P_{\theta,ii} \widehat{\sigma}_i^2}  \frac{  1 }{\sqrt{2 \sum_{i \in [r_n]} \widetilde{w}_{i,n}^2 } } -1 \right) = o_\Prob(1),
\end{align*}
where we now apply Lemmas \ref{lem:variance_consistency_random_design} and \ref{lem:weights_consistency_random_design}, and use the fact that $\theta \gtrsim 1$ with probability approaching one by Proposition \ref{prop:ridge_parameter}. By \cite[Theorem 1.5]{rudelson2015small}, we have
\begin{align*}
    \sup_{y \in \mathbb{R}} \mathbb{P}\left( \left|\frac{\mathcal{E}^\top P_\theta \mathcal{E}}{\sum_{i \in [n]}P_{\theta,ii}\sigma_i^2}  - y\right| \leq \ell \bigg| \bf X  \right) \leq C_\zeta \ell^{(1-\zeta)/2}
\end{align*}
with probability approaching one, for any $\ell >0$ and $\zeta \in (0,1)$, where $C_\zeta$ only depends on $\zeta$. By using the same argument as in the proof of 
Lemma \ref{lem:anti-concentration}, we have
\begin{align*}
    &\quad \sup_{y \in \mathbb{R}} \left| \mathbb{P}\left(\sum_{i \in [r_n]} \widetilde{w}_{i,n}\,\chi^2_{1,i}\leq y \bigg| Y, \bf X \right) - \mathbb{P} \left(\sum_{i \in [r_n]} w_{i,n}\,\chi^2_{1,i} \leq y \bigg| \bf X\right)  \right| \\
    &\leq \sup_{y \in \mathbb{R}}\mathbb{P}\left( \left|\sum_{i \in [r_n]} w_{i,n}\,\chi^2_{1,i} - y \right|\leq \ell \bigg| \bf X \right) + \mathbb{P}\left( \left|\sum_{i \in [r_n]} \widetilde{w}_{i,n}\,\chi^2_{1,i} - \sum_{i \in [r_n]} w_{i,n}\,\chi^2_{1,i} \right|>\ell \bigg| Y, \bf X \right).
    \end{align*}
Together with the anti-concentration inequality and Lemma \ref{lem:weights_consistency_random_design}, this implies that the Kolmogorov distance between the distribution of $F_{\widetilde{w}_n} $ conditional on the data and the distribution of $F_{w_n}$ conditional on $\bf X$ is $o_\Prob(1)$, which further implies that $q_{1-\alpha}(F_{\widetilde{w}_n}) = q_{1-\alpha}(F_{w_n}) + o_\Prob(1)$, in view of Lemmas \ref{lem:quantile_function} and \ref{lem:weighted_chi_squared} (note that we can set $R=r$ in Lemma \ref{lem:weighted_chi_squared}). In addition, we have 
\begin{align*}
    \sup_{y \in \mathbb{R}} \mathbb{P}\left( \left|\frac{\epsilon^\top P_\theta \epsilon}{\sum_{i \in [n]}P_{\theta,ii}\sigma_i^2}  - y\right| \leq \ell \bigg| \bf X  \right) \leq C'_\zeta \ell^{(1-\zeta)/2} + o_\Prob(1)
\end{align*}
with probability approaching one, similar to the proof of Theorem \ref{thm:size-control_fixed_design}. This implies that, uniformly in $y \in \mathbb{R}$,
\begin{align*}
    &\quad \mathbb{P}\left( \left|\frac{\epsilon^\top P_\theta \epsilon}{\sum_{i \in [n]}P_{\theta,ii}\sigma_i^2} - y\right| \leq \ell \right) \\
    &\leq \mathbb E \left[1_{A_n} \mathbb{P}\left( \left|\frac{\epsilon^\top P_\theta \epsilon}{\sum_{i \in [n]}P_{\theta,ii}\sigma_i^2} - y\right| \leq \ell \bigg| \bf X  \right) \right] + \mathbb P \left( A_n^c \right) \\
    &\leq C'_\zeta \ell^{(1-\zeta)/2}  + o(1),
\end{align*}
where $A_n$ denotes the event that the anti-concentration inequality holds. It follows that, with $\tilde R_n = o_\Prob(1)$, 
\begin{align*}
    \mathbb P \left( \widehat Q(q) > \mathcal{C}_{\alpha} \right) &= \mathbb P \left(  \frac{\epsilon^\top P_\theta \epsilon}{\sum_{i \in [n]} P_{\theta,ii} \sigma_i^2} > q_{1-\alpha}(F_{w_{n}}) + \tilde R_n \right) \\
    &= \mathbb P \left( \frac{\epsilon^\top P_\theta \epsilon}{\sum_{i \in [n]} P_{\theta,ii} \sigma_i^2} > q_{1-\alpha}(F_{w_{n}})  \right) + o(1) \\
    &= \mathbb P \left( \frac{\mathcal{E}^\top P_\theta \mathcal{E}}{\sum_{i \in [n]}P_{\theta,ii}\sigma_i^2} > q_{1-\alpha}(F_{w_{n}})  \right) + o(1) \\
    &= \mathbb E \left[ \mathbb P \left(  \frac{\mathcal{E}^\top P_\theta \mathcal{E}}{\sum_{i \in [n]}P_{\theta,ii}\sigma_i^2} > q_{1-\alpha}(F_{w_{n}})  \bigg| \bf X \right) \right] + o(1) \\
    &= \mathbb E \left[ \mathbb P \left(\sum_{i \in [r_n]} w_{i,n}\,\chi^2_{1,i} > q_{1-\alpha}(F_{w_{n}})  \bigg| \bf X \right) \right] + o(1) \\
    &= \alpha + o(1),
\end{align*}
where the second equality follows from Lemma \ref{lem:anti-concentration} (with an additional $o(1)$ term), and the third equality now follows from Remark \ref{rem:strong_approx_random_design}. This completes the proof of \eqref{eq:assumption_B*}.

\vspace{5mm}
\noindent \begin{Large} \textbf{Part $II$ (Diverging $r$ case):} \end{Large} Assume now $r \rightarrow \infty$, so that $r_n$ is diverging with probability approaching one by Remark \ref{rem:rn_and_r}, and the same holds for $K_\theta$ by Lemma \ref{lem:Ktheta_and_rn}. By using now Lemma \ref{lem:weights_consistency_random_design}, we have
\begin{align*}
    \frac{F_{\widetilde{w}_n} - 1}{\sqrt{2 \sum_{i \in [r_n]} \widetilde{w}_{i,n}^2 } } \convD \mathcal{N}(0,1)
\end{align*}
in probability conditional on the data. In addition, conditional on $\bf X$, we can follow the same lines as in the proof of Theorem \ref{thm:size-control_fixed_design} to obtain
\begin{align*}
    \frac{ \frac{1}{\sqrt{K_\theta}} \sum_{i \in [n]} P_{\theta,ii} \widehat{\sigma}_i^2 }{\sqrt{\widehat{\Phi} } } \left( \widehat{Q}(q) - 1 \right) \convD \mathcal{N}\left(0, 1 \right)
\end{align*}
in probability by \cite[Theorem 3.1]{dovi2024ridge} and Lemma \ref{lem:variance_consistency_random_design}, where we also use the fact that $\Phi$ is bounded away from zero with probability approaching one by Proposition \ref{prop:ridge_parameter}. Since conditional distribution functions are bounded by one, convergence in probability implies convergence in mean, whence the weak convergence holds unconditionally by the dominated convergence theorem. Finally, applying \cite[Example 1.4.7, Theorem 1.3.6, Theorem 1.3.4(vi)]{VW96} yields \eqref{eq:assumption_B*}. The proof is complete. 
\qed

\subsection{Proof of Theorem \ref{thm:size-control_random_design_extension}} The proof follows the same lines as in the proof of Theorem \ref{thm:size-control_random_design} once we show that $\hat w_j = 1, j \in S_0^c$ where $S_0 = \{ s \in[p] : \beta^0_s \neq 0\}$, and that $\hat w_{j,k} = 1, k \in S_{j}^c$ where $S_{j} = \{k \neq j: \gamma_{j,k} \neq 0\}$, with probability approaching one. 

Consider first our original LASSO estimator $\widehat{\beta}$ defined in \eqref{eq:beta-hat}. Using the KKT condition, we have
\begin{align*}
    (\widehat{\beta}- \beta^0) &= \Theta \Sigma (\widehat{\beta}- \beta^0) \\
    &= \Theta \left( \widehat{\Sigma} (\widehat{\beta}- \beta^0) + (\Sigma-\widehat{\Sigma}) (\widehat{\beta}- \beta^0) \right) \\
    &= \Theta \left(\frac{{\bf X}^\top \eps}{n} - \lambda_0 \widehat{\kappa} + (\Sigma-\widehat{\Sigma}) (\widehat{\beta}- \beta^0) \right),
\end{align*}
where $\lVert \widehat{\kappa} \rVert_{\infty} \leq 1$. Therefore, we have
\begin{align*}
    \lVert \widehat{\beta}- \beta^0 \rVert_\infty \leq \lVert \Theta \rVert_\infty \left( \lVert {\bf X}^\top \eps/n \rVert_\infty + \delta_n \lVert \widehat{\beta}- \beta^0 \rVert_1 + \lambda_0 \right),
\end{align*}
where $\delta_n = \lVert \widehat{\Sigma}-\Sigma \rVert_{\max}$. On the events that $\lVert {\bf X}^\top \eps/n \rVert_\infty \leq \lambda_0/2$ and that $\bf X$ satisfies the RE-condition with parameter $(s_0,3)$, both of which have probability approaching one, the proof in Lemma \ref{lem:main_LASSO} shows that
\begin{align*}
    \lVert \widehat{\beta}- \beta^0 \rVert_1 \lesssim s_0 \lambda_0. 
\end{align*}
By \cite[Theorem 4.1]{kuchibhotla2022moving}, $s_0 \delta_n = o_\Prob(1)$ under our assumption, and thus 
\begin{align*}
    \lVert \widehat{\beta}- \beta^0 \rVert_\infty \leq 2 \lambda_0\lVert \Theta \rVert_\infty \lesssim \lambda_0,
\end{align*}
since we assume that $s_0^*=O(1)$. Thus, by suitably choosing $\lambda_{{\rm prec},0} \asymp \lambda_0 \asymp \sqrt{\log(np)/n}$, we have
\begin{align*}
    \lVert \widehat{\beta}- \beta^0 \rVert_\infty \leq \lambda_{{\rm prec},0} \quad \Longrightarrow \quad \frac{\lambda_{{\rm prec},0}}{|\widehat{\beta}_{j}|\vee \lambda_{{\rm prec},0}} = 1, j \in S_0^c
\end{align*}
with probability approaching one. 

Consider now the original node-wise LASSO estimator $\widehat{\gamma}_{j}$ defined in \eqref{eq:gamma_j}. On the events that $\lVert {\bf X}_{-j}^\top \eta_j / n \rVert_{\infty} \leq \lambda_j/2$ and that $\bf X$ satisfies the RE condition with parameter $(s_0^*+1,3)$, both of which have probability approaching one, Lemma \ref{lem:hatgamma} implies that 
\begin{align*}
    \max_{j \in [p]} \lVert \widehat{\gamma}_{,j} - \gamma_j \rVert_\infty \leq \max_{j \in [p]} \lVert \widehat{\gamma}_{j} - \gamma_j \rVert_2 \lesssim \max_{j \in [p]} \lambda_j \sqrt{s_j} \lesssim  \max_{j \in [p]} \lambda_j,
\end{align*}
since we assume that $s_0^*=O(1)$. Thus, by suitably choosing $\lambda_{{\rm prec},j} \asymp \lambda_j \asymp \sqrt{\log(np)/n}$, we have
\begin{align*}
    \lVert \widehat{\gamma}_j - \gamma_j \rVert_\infty \leq \lambda_{{\rm prec},j} \quad \Longrightarrow \quad \frac{\lambda_{{\rm prec},j}}{|\widehat{\gamma}_{j}(k)|\vee \lambda_{{\rm prec},j}} = 1, k \in S_{j}^c
\end{align*}
uniformly in $j \in [p]$ with probability approaching one. This concludes the proof.

\newpage
\section{Power analysis}
\label{sec:power_analysis}

\begin{lem}
\label{lem:power}
Suppose that ${\bf X}$ is fixed and that the
regularity conditions of
Theorem~\ref{thm:size-control_fixed_design}
hold along a sequence of alternatives
$R\beta^0\neq q$. Let
\begin{align*}
    D &:=
    \sum_{i\in[n]}P_{\theta,ii}\sigma_i^2,\\
    M_\theta &:=
    \left(
    R\widehat\Theta\widehat\Sigma
    \widehat\Theta^\top R^\top+\theta I_r
    \right)^{-1},
\end{align*}
and define
\begin{align}
    \mathcal S_n(q)
    :=
    \frac{
    n(R\beta^0-q)^\top
    M_\theta(R\beta^0-q)
    }{D}.
    \label{eq:power_signal}
\end{align}
Then
\begin{align}
    \widehat Q(q)
    &=
    \frac{\epsilon^\top P_\theta\epsilon}{D}
    +\mathcal S_n(q)
    +O_\Prob\left(
        \sqrt{\frac{\mathcal S_n(q)}{K_\theta}}
    \right)
    +o_\Prob\left(1+\mathcal S_n(q)\right).
    \label{eq:power_expansion}
\end{align}
In particular, if $K_\theta\rightarrow\infty$
and $\mathcal S_n(q)=O(1)$, then
\begin{align}
    \widehat Q(q)
    =
    \frac{\epsilon^\top P_\theta\epsilon}{D}
    +\mathcal S_n(q)+o_\Prob(1).
    \label{eq:power_bounded_signal}
\end{align}
Furthermore, if
$\mathcal S_n(q)\rightarrow\infty$, then
\begin{align}
    \frac{\widehat Q(q)}{\mathcal S_n(q)}
    \convP 1,
    \label{eq:power_diverging_signal}
\end{align}
and the test based on the rejection rule
\eqref{eq:rejection-criteria} is consistent:
\begin{align}
    \Prob\left(
        \widehat Q(q)>\mathcal C_\alpha
    \right)\rightarrow1.
    \label{eq:power_consistency}
\end{align}
\end{lem}

\noindent \textbf{Proof of Lemma~\ref{lem:power}:}

\noindent Recall that
\begin{align*}
    \sqrt{n}(\widehat b-\beta^0)
    =
    \frac{\widehat\Theta{\bf X}^\top\epsilon}
    {\sqrt{n}}-\Delta.
\end{align*}
It follows that
\begin{align}
    b(q)
    &=
    \frac{R\widehat\Theta{\bf X}^\top\epsilon}
    {\sqrt{n}}
    -R\Delta
    +\sqrt{n}(R\beta^0-q).
    \label{eq:power_decomposition}
\end{align}
By the definition of $\widehat Q(q)$,
\begin{align}
    \widehat Q(q)
    &=
    \frac{1}{\widehat D}
    \left[
    \frac{R\widehat\Theta{\bf X}^\top\epsilon}
    {\sqrt{n}}
    -R\Delta
    +\sqrt{n}(R\beta^0-q)
    \right]^\top
    M_\theta
    \left[
    \frac{R\widehat\Theta{\bf X}^\top\epsilon}
    {\sqrt{n}}
    -R\Delta
    +\sqrt{n}(R\beta^0-q)
    \right],
    \label{eq:power_quadratic}
\end{align}
where
\begin{align*}
    \widehat D
    :=
    \sum_{i\in[n]}
    P_{\theta,ii}\widehat\sigma_i^2.
\end{align*}
Expanding the quadratic form gives
\begin{align}
    \widehat Q(q)
    &=
    \frac{\epsilon^\top P_\theta\epsilon}
    {\widehat D}
    +
    \frac{
    n(R\beta^0-q)^\top
    M_\theta(R\beta^0-q)}
    {\widehat D}
    \notag\\
    &\quad+
    \frac{
    2(R\beta^0-q)^\top
    M_\theta R\widehat\Theta
    {\bf X}^\top\epsilon}
    {\widehat D}
    \notag\\
    &\quad-
    \frac{
    2\epsilon^\top{\bf X}
    \widehat\Theta^\top R^\top
    M_\theta R\Delta}
    {\sqrt{n}\widehat D}
    \notag\\
    &\quad-
    \frac{
    2\sqrt{n}(R\beta^0-q)^\top
    M_\theta R\Delta}
    {\widehat D}
    +
    \frac{
    (R\Delta)^\top M_\theta R\Delta}
    {\widehat D}.
    \label{eq:power_expanded}
\end{align}
We bound each term separately.

\vspace{5mm}
First, by Lemma~\ref{lem:variance-estimator-consistent},
\begin{align}
    \frac{\widehat D}{D}
    =
    1+o_\Prob(1).
    \label{eq:power_variance_consistency}
\end{align}
Indeed, the lemma establishes
\begin{align*}
    \frac{\widehat D-D}{K_\theta}
    =
    o_\Prob(1),
\end{align*}
while the uniformly bounded and nondegenerate
conditional variances imply
\begin{align*}
    D\asymp K_\theta.
\end{align*}
Furthermore,
\begin{align*}
    \mathbb E_\epsilon
    \left[
    \epsilon^\top P_\theta\epsilon
    \right]
    =
    D.
\end{align*}
Since $P_\theta$ is positive semi-definite,
Markov's inequality yields
\begin{align*}
    \frac{\epsilon^\top P_\theta\epsilon}{D}
    =
    O_\Prob(1).
\end{align*}
Consequently,
\begin{align}
    \frac{\epsilon^\top P_\theta\epsilon}
    {\widehat D}
    =
    \frac{\epsilon^\top P_\theta\epsilon}{D}
    +o_\Prob(1).
    \label{eq:power_first_term}
\end{align}

Next, by the definition of $\mathcal S_n(q)$
and \eqref{eq:power_variance_consistency},
\begin{align}
    \frac{
    n(R\beta^0-q)^\top
    M_\theta(R\beta^0-q)}
    {\widehat D}
    &=
    \mathcal S_n(q)\frac{D}{\widehat D}
    \notag\\
    &=
    \mathcal S_n(q)
    +o_\Prob\left(1+\mathcal S_n(q)\right).
    \label{eq:power_second_term}
\end{align}

We now consider the interaction between the
stochastic component and the signal.
Conditional on ${\bf X}$, the errors are
independent and have mean zero. Hence,
\begin{align*}
    &\mathbb E_\epsilon
    \left[
    \left\{
    (R\beta^0-q)^\top
    M_\theta R\widehat\Theta
    {\bf X}^\top\epsilon
    \right\}^2
    \right]\\
    &\quad=
    (R\beta^0-q)^\top
    M_\theta R\widehat\Theta
    {\bf X}^\top\Lambda{\bf X}
    \widehat\Theta^\top R^\top
    M_\theta(R\beta^0-q).
\end{align*}
By the uniform upper bound on the conditional
variances,
\begin{align*}
    {\bf X}^\top\Lambda{\bf X}
    \preceq
    \overline c\,{\bf X}^\top{\bf X}.
\end{align*}
Therefore,
\begin{align*}
    &\mathbb E_\epsilon
    \left[
    \left\{
    (R\beta^0-q)^\top
    M_\theta R\widehat\Theta
    {\bf X}^\top\epsilon
    \right\}^2
    \right]\\
    &\quad\lesssim
    n(R\beta^0-q)^\top
    M_\theta
    \left(
    R\widehat\Theta\widehat\Sigma
    \widehat\Theta^\top R^\top
    \right)
    M_\theta(R\beta^0-q).
\end{align*}
Since
\begin{align*}
    M_\theta
    \left(
    R\widehat\Theta\widehat\Sigma
    \widehat\Theta^\top R^\top
    \right)
    M_\theta
    =
    M_\theta-\theta M_\theta^2
    \preceq M_\theta,
\end{align*}
we obtain
\begin{align*}
    &\mathbb E_\epsilon
    \left[
    \left\{
    (R\beta^0-q)^\top
    M_\theta R\widehat\Theta
    {\bf X}^\top\epsilon
    \right\}^2
    \right]\\
    &\quad\lesssim
    n(R\beta^0-q)^\top
    M_\theta(R\beta^0-q)
    =
    D\mathcal S_n(q).
\end{align*}
By Chebyshev's inequality, together with
$D\asymp K_\theta$ and
\eqref{eq:power_variance_consistency},
\begin{align}
    \frac{
    (R\beta^0-q)^\top
    M_\theta R\widehat\Theta
    {\bf X}^\top\epsilon}
    {\widehat D}
    =
    O_\Prob\left(
    \sqrt{\frac{\mathcal S_n(q)}{K_\theta}}
    \right).
    \label{eq:power_cross_term}
\end{align}

It remains to control the terms involving
the de-biasing remainder $\Delta$.
By the assumptions of
Theorem~\ref{thm:size-control_fixed_design},
\begin{align*}
    \|\Delta\|_\infty
    =
    o_\Prob(1/\sqrt m).
\end{align*}
Since $R$ has at most $m$ nonzero columns
and $\lambda_{\max}(RR^\top)\lesssim1$,
\begin{align*}
    \|R\Delta\|_2^2
    \leq
    \lambda_{\max}(RR^\top)
    \|\Delta_{\operatorname{supp}(R)}\|_2^2
    \lesssim
    m\|\Delta\|_\infty^2
    =
    o_\Prob(1).
\end{align*}
Moreover, $\theta\gtrsim1$ implies
\begin{align*}
    \lambda_{\max}(M_\theta)
    \leq\theta^{-1}
    \lesssim1.
\end{align*}
It follows that
\begin{align}
    \frac{
    (R\Delta)^\top M_\theta R\Delta}
    {\widehat D}
    =
    o_\Prob(1).
    \label{eq:power_remainder_squared}
\end{align}

By the Cauchy--Schwarz inequality,
\begin{align*}
    &\left|
    \frac{
    \epsilon^\top{\bf X}
    \widehat\Theta^\top R^\top
    M_\theta R\Delta}
    {\sqrt n\widehat D}
    \right|\\
    &\quad\leq
    \left(
    \frac{\epsilon^\top P_\theta\epsilon}
    {\widehat D}
    \right)^{1/2}
    \left(
    \frac{
    (R\Delta)^\top M_\theta R\Delta}
    {\widehat D}
    \right)^{1/2}
    =
    o_\Prob(1).
\end{align*}
Similarly,
\begin{align*}
    &\left|
    \frac{
    \sqrt n(R\beta^0-q)^\top
    M_\theta R\Delta}
    {\widehat D}
    \right|\\
    &\quad\leq
    \left(
    \frac{
    n(R\beta^0-q)^\top
    M_\theta(R\beta^0-q)}
    {\widehat D}
    \right)^{1/2}
    \left(
    \frac{
    (R\Delta)^\top M_\theta R\Delta}
    {\widehat D}
    \right)^{1/2}\\
    &\quad=
    o_\Prob\left(
    1+\mathcal S_n(q)
    \right),
\end{align*}
where the final equality follows from
\eqref{eq:power_second_term} and
\eqref{eq:power_remainder_squared}.

\vspace{5mm}
Combining these bounds with
\eqref{eq:power_expanded}, we conclude that
\begin{align*}
    \widehat Q(q)
    &=
    \frac{\epsilon^\top P_\theta\epsilon}{D}
    +\mathcal S_n(q)
    +O_\Prob\left(
    \sqrt{\frac{\mathcal S_n(q)}{K_\theta}}
    \right)
    +o_\Prob\left(1+\mathcal S_n(q)\right),
\end{align*}
which proves \eqref{eq:power_expansion}.

If $K_\theta\rightarrow\infty$ and
$\mathcal S_n(q)=O(1)$, then
\begin{align*}
    \sqrt{\frac{\mathcal S_n(q)}{K_\theta}}
    =o(1),
\end{align*}
so that
\begin{align*}
    \widehat Q(q)
    =
    \frac{\epsilon^\top P_\theta\epsilon}{D}
    +\mathcal S_n(q)+o_\Prob(1).
\end{align*}

If $\mathcal S_n(q)\rightarrow\infty$,
dividing the expansion by
$\mathcal S_n(q)$ yields
\begin{align*}
    \frac{\widehat Q(q)}{\mathcal S_n(q)}
    &=
    1+
    O_\Prob\left(
    \frac{1}{\mathcal S_n(q)}
    \right)
    +
    O_\Prob\left(
    \frac{1}{
    \sqrt{K_\theta\mathcal S_n(q)}
    }
    \right)
    +o_\Prob(1)\\
    &=1+o_\Prob(1).
\end{align*}
Thus,
\begin{align*}
    \frac{\widehat Q(q)}{\mathcal S_n(q)}
    \convP1.
\end{align*}

Finally, we establish the boundedness of
the critical value. Recall from
\eqref{eq:critical_value} that
\begin{align*}
    \mathcal C_\alpha
    =
    1+
    \frac{\sqrt{K_\theta\widehat\Phi}}
    {\widehat D}
    \left(
    \frac{
    q_{1-\alpha}(F_{\widetilde w_n})-1}
    {\sqrt{2\sum_{i\in[r_n]}
    \widetilde w_{i,n}^2}}
    \right).
\end{align*}
Under the conditions of
Theorem~\ref{thm:size-control_fixed_design},
the existing variance-consistency results
imply
\begin{align*}
    \widehat D\asymp_\Prob K_\theta,
    \qquad
    \widehat\Phi=O_\Prob(1).
\end{align*}
Furthermore, the standardized weighted
chi-squared distribution has conditional
mean zero and variance one, implying that
its $(1-\alpha)$-quantile is bounded
for any fixed $\alpha\in(0,1)$.
Consequently,
\begin{align*}
    \mathcal C_\alpha
    =
    1+O_\Prob(K_\theta^{-1/2})
    =
    O_\Prob(1).
\end{align*}

Since
$\mathcal S_n(q)\rightarrow\infty$
and
$\widehat Q(q)/\mathcal S_n(q)\convP1$,
we conclude that
\begin{align*}
    \Prob\left(
    \widehat Q(q)>\mathcal C_\alpha
    \right)\rightarrow1.
\end{align*}
This completes the proof.


\newpage

\section{Further Details for Monte Carlo Study}

\subsection{Implementation Details} 
\label{sec:appendix_implement}

The implementation details for all the test statistics are as follows:
\begin{enumerate}
    \item {\bf Wald-HE (BIC)/(GIC)}: Heteroskedasticity-robust Wald test proposed by \cite{caner2018asymptotically}, with ordinary LASSO using either Bayesian Information Criterion (BIC) or Generalized Information Criterion (GIC) to select $\lambda_0$ in the main LASSO regression. According to \cite{caner2018asymptotically} and our preliminary experiments, using cross-validation to select this tuning parameter does not improve the results. The same parameter is also used for the node-wise LASSO regressions. It is also possible to use BIC or GIC to select the tuning parameters for the node-wise LASSO regressions, but it turns out that this does not change the results too much and can be computationally demanding for large $r$.
    
    \item {\bf Wald-HO (BIC)/(GIC)}: Homoskedastic Wald test using the same approximate inverse as above, and we follow \cite{van-de-geer(2014)} to use the scaled LASSO in \cite{sun2012scaled} for the main LASSO regression.

    \item {\bf LMTW (BIC)/(GIC)}: Our proposed test, defined in \eqref{eq:rejection-criteria}, is implemented using ordinary LASSO estimators for both the initial regression in \eqref{eq:beta-hat} and the nodewise regressions in \eqref{eq:gamma_j}. Following \cite{Fan-Tang2013} and \cite{caner2018asymptotically}, the regularization parameters $\lambda_0$ and $\lambda_j$, $j\in[p]$, are selected in a data-driven manner using either the Bayesian information criterion (BIC) or the generalized information criterion (GIC).

For the BIC-based implementation, we define
\begin{align*}
    \mathrm{BIC}_0(\lambda)
    &=
    \log\left(
    \frac{\lVert Y-{\bf X}\widehat{\beta}(\lambda)
    \rVert_2^2}{n}
    \right)
    +
    \frac{\log n}{n}
    \lVert\widehat{\beta}(\lambda)\rVert_0,
    \\
    \mathrm{BIC}_j(\lambda)
    &=
    \log\left(
    \frac{\lVert X_j-{\bf X}_{-j}
    \widehat{\gamma}_j(\lambda)\rVert_2^2}{n}
    \right)
    +
    \frac{\log n}{n}
    \lVert\widehat{\gamma}_j(\lambda)\rVert_0,
    \qquad j\in[p],
\end{align*}
where $\widehat{\beta}(\lambda)$ and
$\widehat{\gamma}_j(\lambda)$ are the ordinary LASSO estimators defined in \eqref{eq:beta-hat} and \eqref{eq:gamma_j}, respectively, evaluated at the candidate regularization parameter $\lambda$. Similarly, for the GIC-based implementation, we define
\begin{align*}
    \mathrm{GIC}_0(\lambda)
    &=
    \log\left(
    \frac{\lVert Y-{\bf X}\widehat{\beta}(\lambda)
    \rVert_2^2}{n}
    \right)
    +
    \frac{\log\log n\,\log p}{n}
    \lVert\widehat{\beta}(\lambda)\rVert_0,
    \\
    \mathrm{GIC}_j(\lambda)
    &=
    \log\left(
    \frac{\lVert X_j-{\bf X}_{-j}
    \widehat{\gamma}_j(\lambda)\rVert_2^2}{n}
    \right)
    +
    \frac{\log\log n\,\log p}{n}
    \lVert\widehat{\gamma}_j(\lambda)\rVert_0,
    \qquad j\in[p].
\end{align*}
The regularization parameters are selected by minimizing the corresponding information criteria over a sequence of candidate values generated using the \texttt{glmnet} package in R. Specifically, for $\mathrm{IC}\in
\{\mathrm{BIC},\mathrm{GIC}\}$, we select
\begin{align*}
    \widehat{\lambda}_0
    &\in
    \arg\min_{\lambda\in\Lambda_0}
    \mathrm{IC}_0(\lambda),
    \\
    \widehat{\lambda}_j
    &\in
    \arg\min_{\lambda\in\Lambda_j}
    \mathrm{IC}_j(\lambda),
    \qquad j\in[p],
\end{align*}
where $\Lambda_0 = \{ \lambda_{0,1},...,\lambda_{0,N_0} \}$ and $\Lambda_j = \{ \lambda_{j,1},...,\lambda_{j,N_j} \}$ denote the candidate regularization parameter sets for the initial and $j$th nodewise regressions, respectively. Thus, the initial and nodewise regressions are tuned separately.

    For the quantile of the weighted chi-square random variable, we compute the quantile using $10,000$ random samples, and we set $c_{\varpi} = \min\left\{ 1, \lVert M \rVert_F^2/(4 r_n \lVert M \rVert_2^2) \right\}$ where $M \equiv { R \widehat\Theta {\bf X}^\top}/{\sqrt{n}}$ and $\varpi = c_{\varpi}/2$ for the ridge parameter; see Appendix \ref{subsec:ridge_parameter} for further discussion of these choices.

    \item {\bf LMTW-CL (BIC)/(GIC)}: Our proposed test, defined in \eqref{eq:rejection-criteria}, is implemented using the conservative LASSO extension described in Section~\ref{subsec:conservative_LASSO} for both the initial and nodewise regressions. Following \citet[Section~5]{caner2018asymptotically}, we first obtain preliminary {\it ordinary/plain} LASSO estimates, with regularization parameters selected using BIC or GIC as described above.

    For the initial regression, let $\widehat{\beta}^{L}$ denote the preliminary ordinary LASSO estimator obtained using the selected regularization parameter $\widehat{\lambda}_0^{L}$. We set
\[
\lambda_{{\rm prec},0}
=
\frac{9}{4}\widehat{\lambda}_0^{L}
\]
and construct the conservative LASSO weights as
\[
\widehat{w}_{0,k}
=
\frac{\lambda_{{\rm prec},0}}
{|\widehat{\beta}^{L}_k|
\vee\lambda_{{\rm prec},0}},
\qquad k\in[p].
\]
Using these weights, we compute the conservative LASSO estimator by solving
\[
\widehat{\beta}^{CL}(\lambda)
\in
\arg\min_{\beta\in\mathbb{R}^p}
\left\{
\frac{1}{n}\lVert Y-{\bf X}\beta\rVert_2^2
+
2\lambda\sum_{k=1}^p
\widehat{w}_{0,k}|\beta_k|
\right\}.
\]
The regularization parameter for this weighted LASSO regression is selected separately by minimizing the corresponding BIC or GIC over its candidate regularization parameter set, with the information criteria defined analogously to those for LMTW, using $\widehat{\beta}^{CL}(\lambda)$ in place of $\widehat{\beta}(\lambda)$.

The same two-stage procedure is applied to each nodewise regression. Specifically, preliminary ordinary LASSO estimates $\widehat{\gamma}_j^{L}$ are obtained using BIC- or GIC-selected regularization parameters $\widehat{\lambda}_j^{L}$, and the corresponding precision parameters are set to $\lambda_{{\rm prec},j}
=9\widehat{\lambda}_j^{L}/4$. The resulting conservative weights are then used in the weighted nodewise LASSO regressions, whose regularization parameters are again selected separately using BIC or GIC.

All remaining components of the test, including the construction of the test statistic, the weighted chi-square critical value, and the choice of ridge parameter, follow the same procedure as LMTW (BIC)/(GIC).

\end{enumerate}


\newpage
\subsection{Additional Simulation Results}
\label{sec:additional_simulation}

We present the simulation results for $(n,p)=(200,300)$ in Tables \ref{tab:size_power_n200_p300_r5_homo}-\ref{tab:size_power_n200_p300_r25_hetero}. We again observe that our test statistics have overall better size control across different designs.

\begin{table}[!htbp]
\centering
\begin{threeparttable}
\caption{Empirical Size and Power: $(n,p,r)=(200,300,5)$, Homoskedastic Errors}
\label{tab:size_power_n200_p300_r5_homo}
\scriptsize
\setlength{\tabcolsep}{2.7pt}
\begin{tabular}{lcccccccccccc}
\toprule\toprule
& \multicolumn{3}{c}{D1} & \multicolumn{3}{c}{D2} & \multicolumn{3}{c}{D3} & \multicolumn{3}{c}{D4} \\
\cmidrule(lr){2-4}\cmidrule(lr){5-7}\cmidrule(lr){8-10}\cmidrule(lr){11-13}
& $\delta=0$ & $\delta=0.2$ & $\delta=0.4$ & $\delta=0$ & $\delta=0.2$ & $\delta=0.4$ & $\delta=0$ & $\delta=0.2$ & $\delta=0.4$ & $\delta=0$ & $\delta=0.2$ & $\delta=0.4$ \\
\midrule
Wald-HO (BIC)
& 0.084 & 0.231 & 0.892 & 0.090 & 0.229 & 0.887 & 0.093 & 0.241 & 0.878 & 0.079 & 0.244 & 0.879 \\
Wald-HO (GIC)
& 0.092 & 0.214 & 0.858 & 0.097 & 0.216 & 0.848 & 0.104 & 0.232 & 0.853 & 0.093 & 0.229 & 0.840 \\
\addlinespace
Wald-HE (BIC)
& 0.074 & 0.478 & 0.985 & 0.079 & 0.475 & 0.979 & 0.088 & 0.484 & 0.986 & 0.064 & 0.530 & 0.988 \\
Wald-HE (GIC)
& 0.066 & 0.430 & 0.973 & 0.076 & 0.418 & 0.963 & 0.073 & 0.435 & 0.974 & 0.056 & 0.469 & 0.977 \\
\addlinespace
LMTW (BIC)
& 0.069 & 0.342 & 0.957 & 0.072 & 0.295 & 0.940 & 0.075 & 0.321 & 0.947 & 0.062 & 0.334 & 0.942 \\
LMTW (GIC)
& 0.072 & 0.274 & 0.928 & 0.081 & 0.250 & 0.901 & 0.090 & 0.273 & 0.913 & 0.070 & 0.262 & 0.899 \\
\addlinespace
LMTW-CL (BIC)
& 0.066 & 0.437 & 0.991 & 0.063 & 0.412 & 0.980 & 0.081 & 0.442 & 0.983 & 0.060 & 0.443 & 0.984 \\
LMTW-CL (GIC)
& 0.057 & 0.395 & 0.984 & 0.055 & 0.368 & 0.970 & 0.062 & 0.385 & 0.973 & 0.055 & 0.395 & 0.972 \\
\bottomrule
\end{tabular}
\begin{tablenotes}[flushleft]
\scriptsize
\item \textit{Notes:} The nominal size is $\alpha=0.05$. The column $\delta=0$ reports empirical size, while $\delta=0.2$ and $\delta=0.4$ report empirical power. D1: Gaussian $X$ and Gaussian errors; D2: centered-exponential $X$ and Gaussian errors; D3: Gaussian $X$ and standardized $t_5$ errors; D4: centered-exponential $X$ and standardized $t_5$ errors. The errors are homoskedastic in all four designs. Wald(HO) is the homoskedastic Wald test; Wald(HE) is the heteroskedasticity-robust Wald test using conservative-LASSO; LMTW(PL) and LMTW(CL) use ordinary/plain-- and conservative-LASSO respectively.
\end{tablenotes}
\end{threeparttable}
\end{table}

\begin{table}[!htbp]
\centering
\begin{threeparttable}
\caption{Empirical Size and Power: $(n,p,r)=(200,300,5)$, Heteroskedastic Errors}
\label{tab:size_power_n200_p300_r5_hetero}
\scriptsize
\setlength{\tabcolsep}{2.7pt}
\begin{tabular}{lcccccccccccc}
\toprule\toprule
& \multicolumn{3}{c}{D1} & \multicolumn{3}{c}{D2} & \multicolumn{3}{c}{D3} & \multicolumn{3}{c}{D4} \\
\cmidrule(lr){2-4}\cmidrule(lr){5-7}\cmidrule(lr){8-10}\cmidrule(lr){11-13}
& $\delta=0$ & $\delta=0.2$ & $\delta=0.4$ & $\delta=0$ & $\delta=0.2$ & $\delta=0.4$ & $\delta=0$ & $\delta=0.2$ & $\delta=0.4$ & $\delta=0$ & $\delta=0.2$ & $\delta=0.4$ \\
\midrule
Wald-HO (BIC)
& 0.140 & 0.301 & 0.882 & 0.254 & 0.409 & 0.879 & 0.144 & 0.320 & 0.886 & 0.235 & 0.415 & 0.889 \\
Wald-HO (GIC)
& 0.153 & 0.286 & 0.843 & 0.279 & 0.402 & 0.849 & 0.156 & 0.304 & 0.864 & 0.246 & 0.403 & 0.867 \\
\addlinespace
Wald-HE (BIC)
& 0.090 & 0.439 & 0.960 & 0.141 & 0.485 & 0.954 & 0.106 & 0.463 & 0.962 & 0.126 & 0.547 & 0.975 \\
Wald-HE (GIC)
& 0.085 & 0.393 & 0.947 & 0.144 & 0.447 & 0.930 & 0.088 & 0.412 & 0.940 & 0.120 & 0.492 & 0.949 \\
\addlinespace
LMTW (BIC)
& 0.065 & 0.274 & 0.894 & 0.101 & 0.254 & 0.766 & 0.077 & 0.273 & 0.887 & 0.074 & 0.257 & 0.791 \\
LMTW (GIC)
& 0.066 & 0.231 & 0.840 & 0.102 & 0.224 & 0.706 & 0.077 & 0.238 & 0.840 & 0.075 & 0.216 & 0.728 \\
\addlinespace
LMTW-CL (BIC)
& 0.072 & 0.343 & 0.952 & 0.108 & 0.304 & 0.830 & 0.085 & 0.368 & 0.944 & 0.087 & 0.325 & 0.855 \\
LMTW-CL (GIC)
& 0.063 & 0.308 & 0.931 & 0.098 & 0.269 & 0.806 & 0.066 & 0.320 & 0.912 & 0.072 & 0.292 & 0.811 \\
\bottomrule
\end{tabular}
\begin{tablenotes}[flushleft]
\scriptsize
\item \textit{Notes:} The nominal size is $\alpha=0.05$. The column $\delta=0$ reports empirical size, while $\delta=0.2$ and $\delta=0.4$ report empirical power. D1: Gaussian $X$ and Gaussian errors; D2: centered-exponential $X$ and Gaussian errors; D3: Gaussian $X$ and standardized $t_5$ errors; D4: centered-exponential $X$ and standardized $t_5$ errors. The errors are heteroskedastic in all four designs. Wald(HO) is the homoskedastic Wald test; Wald(HE) is the heteroskedasticity-robust Wald test using conservative-LASSO; LMTW(PL) and LMTW(CL) use ordinary/plain-- and conservative-LASSO respectively.
\end{tablenotes}
\end{threeparttable}
\end{table}

\begin{table}[!htbp]
\centering
\begin{threeparttable}
\caption{Empirical Size and Power: $(n,p,r)=(200,300,25)$, Homoskedastic Errors}
\label{tab:size_power_n200_p300_r25_homo}
\scriptsize
\setlength{\tabcolsep}{2.7pt}
\begin{tabular}{lcccccccccccc}
\toprule\toprule
& \multicolumn{3}{c}{D1} & \multicolumn{3}{c}{D2} & \multicolumn{3}{c}{D3} & \multicolumn{3}{c}{D4} \\
\cmidrule(lr){2-4}\cmidrule(lr){5-7}\cmidrule(lr){8-10}\cmidrule(lr){11-13}
& $\delta=0$ & $\delta=0.2$ & $\delta=0.4$ & $\delta=0$ & $\delta=0.2$ & $\delta=0.4$ & $\delta=0$ & $\delta=0.2$ & $\delta=0.4$ & $\delta=0$ & $\delta=0.2$ & $\delta=0.4$ \\
\midrule
Wald-HO (BIC)
& 0.191 & 0.245 & 0.630 & 0.196 & 0.229 & 0.626 & 0.198 & 0.261 & 0.664 & 0.200 & 0.274 & 0.678 \\
Wald-HO (GIC)
& 0.234 & 0.265 & 0.610 & 0.232 & 0.259 & 0.608 & 0.252 & 0.282 & 0.644 & 0.240 & 0.293 & 0.656 \\
\addlinespace
Wald-HE (BIC)
& 0.219 & 0.519 & 0.944 & 0.212 & 0.530 & 0.953 & 0.226 & 0.577 & 0.965 & 0.233 & 0.578 & 0.976 \\
Wald-HE (GIC)
& 0.233 & 0.491 & 0.925 & 0.239 & 0.502 & 0.931 & 0.241 & 0.563 & 0.953 & 0.254 & 0.548 & 0.958 \\
\addlinespace
LMTW (BIC)
& 0.104 & 0.157 & 0.655 & 0.080 & 0.123 & 0.611 & 0.082 & 0.163 & 0.654 & 0.083 & 0.160 & 0.636 \\
LMTW (GIC)
& 0.123 & 0.140 & 0.560 & 0.089 & 0.105 & 0.500 & 0.103 & 0.134 & 0.574 & 0.096 & 0.124 & 0.537 \\
\addlinespace
LMTW-CL (BIC)
& 0.073 & 0.189 & 0.804 & 0.063 & 0.165 & 0.797 & 0.081 & 0.200 & 0.797 & 0.068 & 0.191 & 0.794 \\
LMTW-CL (GIC)
& 0.071 & 0.159 & 0.741 & 0.066 & 0.136 & 0.725 & 0.062 & 0.156 & 0.731 & 0.058 & 0.157 & 0.736 \\
\bottomrule
\end{tabular}
\begin{tablenotes}[flushleft]
\scriptsize
\item \textit{Notes:} The nominal size is $\alpha=0.05$. The column $\delta=0$ reports empirical size, while $\delta=0.2$ and $\delta=0.4$ report empirical power. D1: Gaussian $X$ and Gaussian errors; D2: centered-exponential $X$ and Gaussian errors; D3: Gaussian $X$ and standardized $t_5$ errors; D4: centered-exponential $X$ and standardized $t_5$ errors. The errors are homoskedastic in all four designs. Wald(HO) is the homoskedastic Wald test; Wald(HE) is the heteroskedasticity-robust Wald test using conservative-LASSO; LMTW(PL) and LMTW(CL) use ordinary/plain-- and conservative-LASSO respectively.
\end{tablenotes}
\end{threeparttable}
\end{table}

\begin{table}[!htbp]
\centering
\begin{threeparttable}
\caption{Empirical Size and Power: $(n,p,r)=(200,300,25)$, Heteroskedastic Errors}
\label{tab:size_power_n200_p300_r25_hetero}
\scriptsize
\setlength{\tabcolsep}{2.7pt}
\begin{tabular}{lcccccccccccc}
\toprule\toprule
& \multicolumn{3}{c}{D1} & \multicolumn{3}{c}{D2} & \multicolumn{3}{c}{D3} & \multicolumn{3}{c}{D4} \\
\cmidrule(lr){2-4}\cmidrule(lr){5-7}\cmidrule(lr){8-10}\cmidrule(lr){11-13}
& $\delta=0$ & $\delta=0.2$ & $\delta=0.4$ & $\delta=0$ & $\delta=0.2$ & $\delta=0.4$ & $\delta=0$ & $\delta=0.2$ & $\delta=0.4$ & $\delta=0$ & $\delta=0.2$ & $\delta=0.4$ \\
\midrule
Wald-HO (BIC)
& 0.225 & 0.277 & 0.663 & 0.296 & 0.343 & 0.733 & 0.228 & 0.283 & 0.689 & 0.300 & 0.376 & 0.762 \\
Wald-HO (GIC)
& 0.264 & 0.303 & 0.635 & 0.338 & 0.368 & 0.716 & 0.274 & 0.303 & 0.666 & 0.342 & 0.402 & 0.747 \\
\addlinespace
Wald-HE (BIC)
& 0.246 & 0.591 & 0.970 & 0.310 & 0.653 & 0.983 & 0.266 & 0.643 & 0.979 & 0.310 & 0.713 & 0.992 \\
Wald-HE (GIC)
& 0.276 & 0.571 & 0.953 & 0.342 & 0.646 & 0.975 & 0.290 & 0.616 & 0.968 & 0.339 & 0.691 & 0.987 \\
\addlinespace
LMTW (BIC)
& 0.092 & 0.147 & 0.605 & 0.095 & 0.149 & 0.528 & 0.085 & 0.154 & 0.590 & 0.067 & 0.145 & 0.589 \\
LMTW (GIC)
& 0.108 & 0.130 & 0.519 & 0.096 & 0.130 & 0.445 & 0.103 & 0.136 & 0.513 & 0.074 & 0.115 & 0.499 \\
\addlinespace
LMTW-CL (BIC)
& 0.077 & 0.181 & 0.754 & 0.084 & 0.180 & 0.675 & 0.077 & 0.189 & 0.721 & 0.065 & 0.164 & 0.705 \\
LMTW-CL (GIC)
& 0.068 & 0.154 & 0.689 & 0.072 & 0.155 & 0.613 & 0.060 & 0.155 & 0.668 & 0.054 & 0.132 & 0.649 \\
\bottomrule
\end{tabular}
\begin{tablenotes}[flushleft]
\scriptsize
\item \textit{Notes:} The nominal size is $\alpha=0.05$. The column $\delta=0$ reports empirical size, while $\delta=0.2$ and $\delta=0.4$ report empirical power. D1: Gaussian $X$ and Gaussian errors; D2: centered-exponential $X$ and Gaussian errors; D3: Gaussian $X$ and standardized $t_5$ errors; D4: centered-exponential $X$ and standardized $t_5$ errors. The errors are heteroskedastic in all four designs. Wald(HO) is the homoskedastic Wald test; Wald(HE) is the heteroskedasticity-robust Wald test using conservative-LASSO; LMTW(PL) and LMTW(CL) use ordinary/plain-- and conservative-LASSO respectively.
\end{tablenotes}
\end{threeparttable}
\end{table}

\newpage
\section{Supplementary Results for the Empirical Application}
\label{sec:appendix_empirical}

\subsection{Calibration} \label{cali_ERR}

We first examine the marginal sub-Weibull condition on the covariates in Assumption \ref{ass:size}.3 using the standardized design matrix employed in the initial LASSO regressions. The tail parameter $\nu$ governs the strength of the distributional restriction on the covariates and enters the rate conditions in Assumption \ref{ass:size}.4. To provide empirical guidance on the choice of $\nu$, we consider three candidate orders: $\nu=2$ corresponds to sub-Gaussian tails, $\nu=1.5$ represents an intermediate case, and $\nu=1$ allows for the heavier sub-exponential tails. For each candidate order $\nu\in\{1,1.5,2\}$ and covariate $j$, we compute the empirical marginal Orlicz norm 
\begin{align*}
    K_{n,p}(\nu)\equiv \max_{j \in [p]} \Bigl( \inf \bigl\{K>0:\frac{1}{n}\sum_{i=1}^n\exp\left(\left|\frac{X_{ij}}{K}\right|^\nu\right)\leq 2\bigr\} \Bigr)
\end{align*}
and report this in Table \ref{tab:x-tail-calibration}. For a given $\nu$, this quantity measures the scale required to satisfy the empirical exponential moment bound uniformly across covariates. Table \ref{tab:x-tail-calibration} reports these quantities for the main and robustness-check specifications. The maximal empirical Orlicz norms are 2.26 and 2.69 at $\nu=1$, compared with 3.79 and 4.67 at $\nu=1.5$, and 4.99 and 6.17 at $\nu=2$, respectively. These comparisons illustrate how the empirical scale changes as increasingly stringent tail restrictions are considered, although the norms across different orders are not directly comparable measures of goodness of fit. As $\eta$ decreases from $2$ to $1$, the empirical
maximal Orlicz norm $\widehat K_{\max}(\eta)$
decreases, reflecting the weaker tail restrictions associated with smaller values of $\eta$. Since our theoretical framework accommodates any fixed $\eta>0$, it allows for substantially heavier-tailed covariates than those satisfying the sub-Gaussian condition. In particular, the relatively small empirical Orlicz norms obtained at $\eta=1$ provide descriptive support for considering a sub-exponential specification. Nevertheless, these finite-sample diagnostics do not establish the required asymptotic boundedness condition, and smaller values of $\eta$ may impose more stringent rate requirements.

Our test allows for $\nu=1$ as a conservative working specification to accommodate relatively heavy-tailed covariates without relying on stronger distributional restrictions. This choice is also consistent with the objective of our theoretical framework, which allows for non-sub-Gaussian covariates, albeit at the cost of more demanding rate conditions. The diagnostics provide empirical context for this choice rather than a formal identification of the population tail order or verification of Assumption \ref{ass:size}.3.

\begin{table}[H]
\centering
\footnotesize
\setlength{\tabcolsep}{8.0pt}
\begin{threeparttable}
\caption{Empirical calibration of the marginal sub-Weibull parameter of $\textbf{X}$}
\label{tab:x-tail-calibration}
\begin{tabular*}{\textwidth}{@{\extracolsep{\fill}}lccccc@{}}
\toprule
Specification & $p$ & $ K_{n,p}(2)$ & $K_{n,p}(1.5)$ & $ K_{n,p}(1)$ \\
\midrule
main & 430 & 4.99 & 3.79 & 2.26 \\
robustness & 858 & 6.17 & 4.67 & 2.69  \\
\bottomrule
\end{tabular*}
\begin{tablenotes}[flushleft]
\scriptsize
\item Notes: Calculations use the standardized LASSO design matrix. For each candidate order $\nu$, $\widehat K_{\max}(\nu)$ denotes the largest empirical marginal Orlicz norm across covariates. These quantities are descriptive and do not constitute a formal test of the population sub-Weibull condition.
\end{tablenotes}
\end{threeparttable}
\end{table}

We next examine the error moment condition in Assumption \ref{ass:mean_zero}.3. Since the structural errors are unobserved, we use the LASSO residuals from the main specification with $p=430$ as empirical proxies. Panel A of Figure \ref{fig:p429-residual-moment-diagnostics} plots the empirical $L_h$ norms of the standardized residuals for $h=1,\ldots,8$, together with Gaussian and standardized $t(5)$ benchmarks. Panel B reports the fraction of the fourth, sixth, and eighth sample absolute moments contributed by the $k$ largest absolute residuals. Under both BIC and GIC, the higher-order empirical residual moments exceed the Gaussian benchmark and become increasingly concentrated among a small number of observations. In particular, the ten largest absolute residuals account for approximately three quarters of the fourth sample absolute moment and more than 90 percent of the sixth and eighth sample absolute moments. This concentration illustrates the sensitivity of higher-order sample moments to extreme observations. We therefore retain $h=4$, the weakest moment requirement allowed by Assumption \ref{ass:mean_zero}.3, rather than relying on stronger higher-order moment assumptions. These diagnostics do not establish the conditional fourth-moment restriction or rule out the existence of higher-order population moments.

\begin{figure}[H]
\centering
\caption{Residual moment diagnostics for the $p=430$ LASSO residuals}
\label{fig:p429-residual-moment-diagnostics}
\vspace{0.25em}
\includegraphics[width=0.92\textwidth]{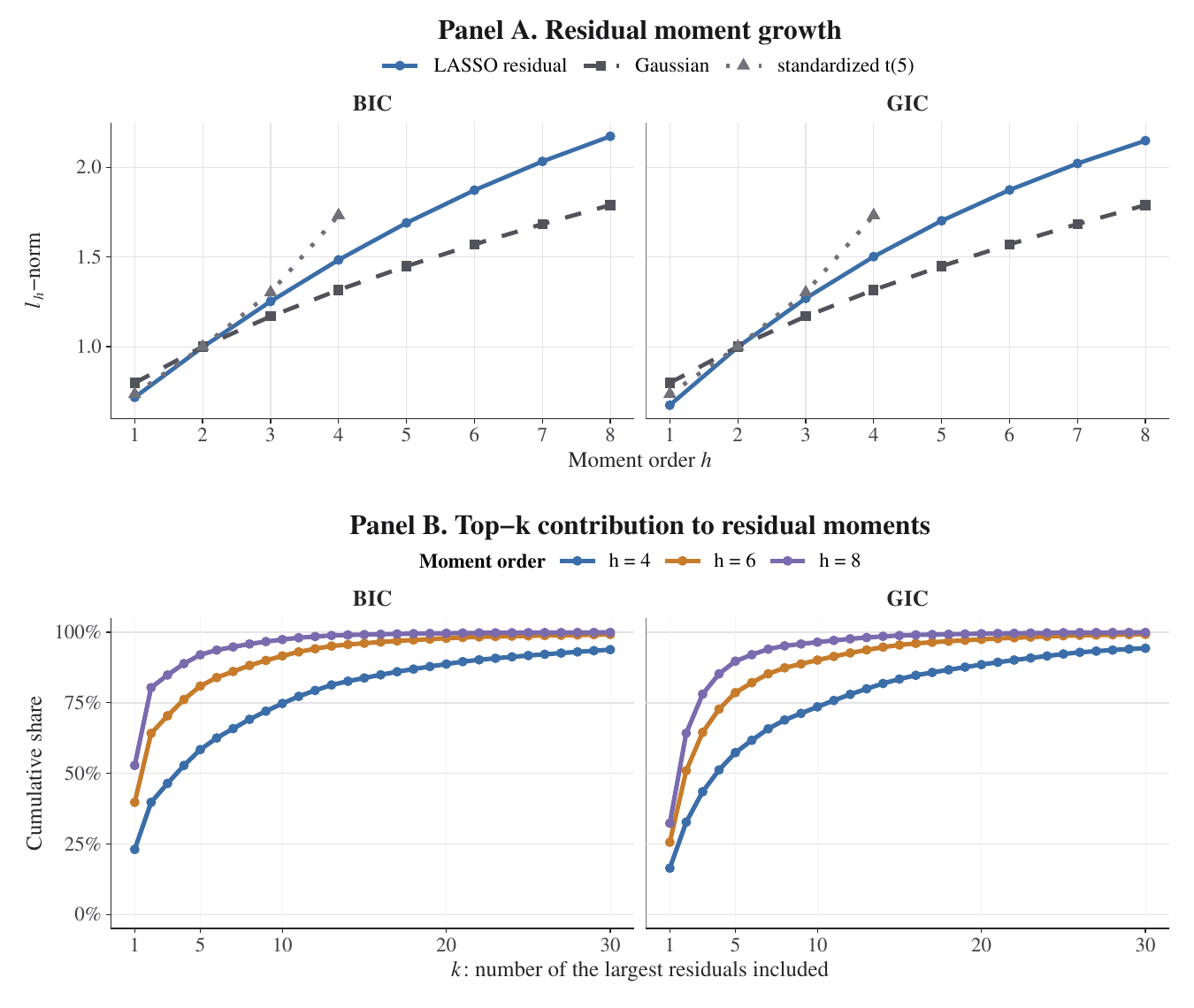}
\vspace{0.25em}
\begin{minipage}{0.92\textwidth}
\footnotesize
Notes: the LASSO residuals are obtained under the main setting $(n,p)=(210,430)$, centered and scaled by their sample standard deviation. The ordinary LASSO and conservative LASSO yield identical residuals under each criterion in Table \ref{tab:lag0-linear-main-p429}. The upper panel plots empirical $\ell_h$-norms for $h=1,\ldots,8$ with Gaussian and standardized $t(5)$ benchmarks; the latter is shown only for $h<5$, where its population moments exist. The lower panel plots the shares of the fourth, sixth, and eighth sample absolute moments contributed by the $k$ largest absolute standardized residuals.
\end{minipage}
\end{figure}

\subsection{Robustness Check} \label{ERR_RC}

This subsection examines the robustness of our empirical findings to a higher-dimensional specification that augments the baseline model with squared regressors. Specifically, starting from the $p=430$ regressors (originally 429 regressors, plus the constant term) in the main specification, we include the squared term of each regressor, introducing an additional 430 covariates and increasing the total dimension to $p=858$. Table~\ref{tab:lag0-linear-robustness-significance} reports the results of the joint-zero and coefficient-equality tests under this expanded specification, providing a robustness check for the findings reported in Table~\ref{tab:lag0-linear-main-p429}. Overall, the results remain qualitatively similar to those obtained under the baseline specification, suggesting that our main empirical conclusions are robust to the inclusion of quadratic terms.

\begin{table}[!htbp]
\centering
\footnotesize
\setlength{\tabcolsep}{3.0pt}
\renewcommand{\arraystretch}{1.0}
\begin{threeparttable}
\caption{Concept-group tests of zero coefficients and coefficient equality, $(n,p)=(210,858)$}
\label{tab:lag0-linear-robustness-significance}
\begin{tabular*}{\textwidth}{@{\extracolsep{\fill}}lcllll@{}}
\toprule
& & & & \multicolumn{2}{c}{$p$-value} \\
\cmidrule(l){5-6}
Group $G$ & $m$ & Method & Selected concepts
& $R_G=I_m$
& $R_G=[-\mathbb{1}_{m-1},\,I_{m-1}]$ \\
\midrule
credit & 8 & LMTW (BIC) & ``banks'' & $<0.001^{***}$ & $0.0020^{***}$ \\
 & & LMTW (GIC) & None & $<0.001^{***}$ & $0.0181^{**}$ \\
 & & LMTW-CL (BIC) & ``banks'' & $<0.001^{***}$ & $0.0051^{***}$ \\
 & & LMTW-CL (GIC) & None & $<0.001^{***}$ & $0.0493^{**}$ \\
\midrule
demand & 11 & LMTW (BIC) & ``business activity'' & $<0.001^{***}$ & $0.0230^{**}$ \\
 & & LMTW (GIC) & ``business activity'' & $<0.001^{***}$ & $0.0073^{***}$ \\
 & & LMTW-CL (BIC) & ``business activity'' & $0.0034^{***}$ & $0.0554^{*}$ \\
 & & LMTW-CL (GIC) & ``business activity'' & $<0.001^{***}$ & $0.0085^{***}$ \\
\midrule
housing & 9 & LMTW (BIC) & None & $0.0293^{**}$ & $0.3407$ \\
 & & LMTW (GIC) & None & $<0.001^{***}$ & $0.3094$ \\
 & & LMTW-CL (BIC) & None & $0.0542^{*}$ & $0.3637$ \\
 & & LMTW-CL (GIC) & None & $<0.001^{***}$ & $0.2461$ \\
\midrule
inflation & 10 & LMTW (BIC) & None & $0.5477$ & $0.6689$ \\
 & & LMTW (GIC) & None & $0.4025$ & $0.3733$ \\
 & & LMTW-CL (BIC) & None & $0.5986$ & $0.7640$ \\
 & & LMTW-CL (GIC) & None & $0.5474$ & $0.5340$ \\
\midrule
labor & 9 & LMTW (BIC) & None & $0.1096$ & $0.1182$ \\
 & & LMTW (GIC) & None & $<0.001^{***}$ & $0.0017^{***}$ \\
 & & LMTW-CL (BIC) & None & $0.1730$ & $0.2127$ \\
 & & LMTW-CL (GIC) & None & $0.0078^{***}$ & $0.0101^{**}$ \\
\midrule
oil & 4 & LMTW (BIC) & None & $0.3758$ & $0.1293$ \\
 & & LMTW (GIC) & None & $0.1642$ & $0.0448^{**}$ \\
 & & LMTW-CL (BIC) & None & $0.3540$ & $0.1697$ \\
 & & LMTW-CL (GIC) & None & $0.1874$ & $0.0621^{*}$ \\
\bottomrule
\end{tabular*}
\begin{tablenotes}[flushleft]
\scriptsize
\item Notes: Tests concern the lag-0 linear coefficients. For coefficient-equality tests $R_G = [-\mathbb{1}_{m-1}, I_{m-1}]$, the reference is the first concept listed for its group in Table~\ref{tab:lag0-linear-main-concepts}. LMTW and LMTW-CL denote implementations based on ordinary/plain and conservative LASSO, respectively. BIC and GIC represent the criteria used to select the penalty parameter in both the main and node-wise LASSO steps. ``Selected concepts'' lists the concepts selected by the main LASSO within each group. Simulated $p$-values and critical values use $\widetilde\Phi$ and 200,000 weighted chi-square draws. Stars denote rejection at the 10\% ($^{*}$), 5\% ($^{**}$), and 1\% ($^{***}$) levels.
\end{tablenotes}
\end{threeparttable}
\end{table}

\end{document}